\documentclass[a4paper,11pt]{article}
\usepackage[utf8x]{inputenc}
\usepackage{graphicx} 
\usepackage{amsmath} 
\usepackage{todonotes} 
\usepackage{amsthm} 
\usepackage{amsfonts}
\usepackage{amssymb}
\usepackage{array}
\usepackage{url}
\usepackage{graphicx}
\usepackage{caption}
\usepackage{subcaption}
\usepackage{bbm}

\newtheorem{thm}{Theorem}[section]
\newtheorem{prop}{Proposition}

\newtheorem{lem}[thm]{Lemma}
\newtheorem{mydef}{Definition}

\title{Critical probabilities and convergence time of Percolation Probabilistic Cellular Automata}
\author{Lorenzo Taggi \\ \textit{Max Planck Institute for Mathematics in the Sciences}}
\date{\today}

\begin{document}
\maketitle

\section*{Abstract}
This paper considers a class of probabilistic cellular automata
undergoing a phase transition with an absorbing state.
Denoting by ${\mathcal{U}}(x)$ the neighbourhood of
site $x$, the transition probability is
$T(\eta_x = 1 | \eta_{{\mathcal{U}}(x)}) = 0$ if $\eta_{{\mathcal{U}}(x)}= \mathbf{0}$
or $p$ otherwise, $\forall x \in \mathbb{Z}$.
For any $\mathcal{U}$ there exists
a non-trivial critical probability
$p_c( {\mathcal{U}})$ 
that separates a phase with an absorbing state from a fluctuating phase.
This paper studies how the neighbourhood affects the value of 
$p_c( {\mathcal{U}})$ and provides lower bounds for
$p_c( {\mathcal{U}})$. 
Furthermore, by using dynamic renormalization techniques, we 
prove that the expected convergence time of the processes
on a finite space with periodic boundaries grows exponentially
(resp. logarithmically) with the system size if $p > p_c$ (resp. $p<p_c$).
This provides a partial answer to an open problem in
Toom \textit{et al.} (1990, 1994).

\section{Introduction}
Probabilistic cellular automata (PCA) 
are discrete-time Markov processes
modelling the time evolution of a 
multicomponent system.
Their main feature is
the synchronous update of the states
of the components, which
take values in a finite set 
and interact with their neighbours
according to a given probabilistic interaction rule.

PCA are favourable models to study non-equilibrium phenomena.
Indeed, on the one hand, their definition is simple,
as the space of realizations is discrete and interactions
are local. On the other hand, despite this simplicity,
they show a variety of complex behaviours.

One of the interesting phenomena involving
probabilistic cellular automata
is the transition from ergodic to non-ergodic regime.
After setting a free parameter above or below a certain critical threshold,
at infinite time the process preserves part of the information 
on its initial condition (non-ergodic behaviour).
Namely, the probability measure at infinite time depends on the initial 
state of the dynamics.
On the contrary, if the process is ergodic, it admits a unique, attracting invariant measure.
In \cite{Lebowitz} it has been shown that the non-ergodic regime of a PCA is connected
to the existence of a phase transition for the PCA, considered as a statistical mechanics system.

Over the last 50 years, PCA have undergone intense analytical and numerical investigations
( e.g. \cite{Dawson, Lebowitz, ToomCont, ToomDiscr}). 
However, as far as we know, many questions 
involving the rate of convergence to equilibrium or the 
characterisation of the invariant measures 
still remain open, even for the simplest models
(see e.g. \cite{ToomCell, ToomDiscr}  for a survey).

In this paper we consider a class of PCA that has a correspondence with 
percolation. These models are refereed to as \textit{Percolation Systems} 
in \cite{ToomCell} and as \textit{Percolation Operators} in \cite{ToomDiscr}.
From now on we will refer to them as \textit{Percolation PCA}. This class includes the well studied \textit{Stavksaya's process} (see e.g. \cite{Depoorter,  Maere, Stavskaja, Stavskaja2, Stavskaja3, ToomCell,ToomDiscr, ToomCont}), in which  the neighbourhood of every corresponds to the site itself and its right nearest neighbour on the one dimensional lattice. On the contrary in Percolation PCA the neighbourhood of every site could be any finite (translation invariant) set.

The reason why we decided to consider Percolation PCA is that
their simplicity, combined with the presence of a phase transition,
make them an interesting test case
for attempts to characterise transient behaviour
and stationary measures for spatially extended stochastic dynamics.
Namely Percolation PCA are 
a prominent model for studying
\textit{absorbing state phase transitions}
(\cite{Hinrichsen}), 
i.e. there exists a phase characterised
by almost sure convergence into an ``absorbing state''
(a realisation where the process remains for ever whenever reached)
and a fluctuating phase, where the process remains active at all times.

In this paper we discuss two distinct aspects of the Percolation PCA.
In Section \ref{sect:criticalprobabilities} we study how the neighbourhood 
affects the critical probability.
We provide a lower bound
for critical probabilities $p_c(\mathcal{U})$ and
our result is stated in Theorem \ref{theo:maintheorem1}.
With our estimations we improve
the previous lower bound \cite{Pearce}
showing that  $p_c (\mathcal{U}) > 1/2$ strictly
if the neighbourhood $\mathcal{U} = \{ -1, 0, 1 \}$.
Furthermore, we provide
new bounds in case of neighbourhoods 
not considered before (as far as we know).
The comparison with numerical estimations, provided in the last section 
of this article,
shows that
our bounds are sharp.
In order to derive the lower bound we studied the temporal evolution 
evolution of ``absorbed sets''
(sets of adjacent sites all in state ``zero'').
If these sets on average are
expanding, the realisation at infinite time 
is ``all zeros'' almost surely. 
This idea comes from \cite[Chapter 6]{ToomDiscr}.
Our estimations take into account a certain aspect of the dynamics,
i.e. absorbed sets can dynamically merge one with the other.
This leads to an improvement of the bound.

In Section \ref{sect:timeofconvergence} we consider
Percolation PCA on a finite one dimensional lattice
with periodic boundaries and we study the convergence time
of the process into the absorbing state.
Our second main result is stated
in Theorem \ref{theo:maintheorem2}.
We show that at $p_c$ there exists a transition
from a fast to a slow convergence regime.
Namely we prove that the expected convergence time of the model
grows exponentially
(resp. logarithmically) as the size of the system grows 
if $p > p_c(\mathcal{U})$ (resp. $p<p_c(\mathcal{U})$).
This provides a partial answer to the \textit{Unsolved Problem 5.3.3} in
\cite{ToomCell} and to an open problem mentioned in \cite[Pag. 80-83]{ToomDiscr}.
If compared with \cite{ToomDiscr},
where the fast (resp. slow) convergence behaviour is proved for 
$p$ small enough  (resp. close enough to $1$),
our result provides a sharp estimation.
The slow convergence regime can be interpreted as a metastable
behaviour of the model, as the process spends an exceptionally long time
into a non-stable state before falling into the absorbing state.
Similar studies on the metastable behaviour of PCA models
were recently presented also in \cite{Bigelis, 
Cirillo, Louis},
although the methods used there do not apply in our case,
as Percolation PCA are not reversible and do
not have a naturally associated potential.
Numerical estimations of $p_c(\mathcal{U})$ 
(e.g. \cite{Jensen, Lubeck, Perlsman})
are obtained assuming that the metastable regime 
(the actual regime observed in numerical simulations, 
as there is no way to really simulate  ``infinite space'' in computers)
is observed only for all values of $p$ 
at which the infinite process is in the fluctuating phase. 
Although this fact might appear obvious in terms
of physical intuition, Theorem \ref{theo:maintheorem2} provides 
a justification for this assumption from a rigorous mathematical point of view.

The proof of our result relies almost entirely on the
correspondence between Percolation PCA and
oriented percolation in two dimensions.
This connection has been described for the first time in \cite{Uniform}.
The proof of the statement of the theorem
involving the case of $p<p_c$ is an application of some percolation estimations
presented in \cite{DurretOr}. We generalize these estimates
to the percolation model considered here, which differs from  \cite{DurretOr}
as here the neighbourhood is an arbitrary finite set and because sites (instead of bonds) can be open or closed.
The proof of the statement involving the case of $p>p_c$ is more technical
and is based on
\textbf{(1)} the generalization of the dynamic-block argument
provided by \cite{DurretOr}
to the case of non symmetric neighbourhood
with more than two elements 
and \textbf{(2)} the estimation
of the probability of a certain event
involving a dual lattice construction
provided by \cite{Uniform}.

We shall end this introductory section by presenting the structure of the paper.
In Section \ref{sect:themodel} we define
the model and we present our main results,
Theorem \ref{theo:maintheorem1} and Theorem \ref{theo:maintheorem2}.
In Section \ref{sect:criticalprobabilities}
we prove Theorem \ref{theo:maintheorem1}.
In Section \ref{sect:timeofconvergence}, divided into
three subsections,
we prove Theorem \ref{theo:maintheorem2}.
In Subsection \ref{sect:ergodicity}
we describe the correspondence
between Percolation PCA and oriented percolation
in two dimension, following \cite{ToomDiscr, Uniform}.
In Subsection \ref{sect:percolationestimates}
we present several percolation estimations
from \cite{DurretOr} used to prove
of the theorem.
Finally in Subsection \ref{sect:proofTheo2} we prove Theorem \ref{theo:maintheorem2}.

\section{Definition and Results}
\label{sect:themodel}
Probabilistic Cellular Automata (PCA)
are discrete-time Markov chains on a product space,
$\Sigma= X^{S}$.
In this paper we consider both the case of infinite space, 
$S = \mathbb{Z}$,
and of finite space,  
$S = \mathbb{S}_n$, $\mathbb{S}_n: = \{ -n, -n + 1, \ldots n-2, n-1 \}$.

We consider the case of boolean variables, $X = \{ 0, 1 \}$.
Realisations of the process are denoted by $\eta \in \Sigma$.
For any  $x \in S$ and any $K \subset S$, use $\eta_x$ to denote the $x$-th component of the vector $\eta$ and $\eta_K$ to designate the set of components
corresponding to the sites of $K$.

We introduce a \textit{neighbourhood function} on $S$.
We first fix a finite set $\mathcal{U} =\{  s_1, s_2, \ldots s_u \}
\subset \mathbb{Z}$, assuming that
$s_1 < s_2 < \ldots < s_u$.
If $S = \mathbb{Z}$,   $\forall x \in S$ we define the neighbourhood
of $x$ as $\mathcal{U}{ (x )} =\{  s_1, s_2, \ldots s_u \} + x$.
If $S = \mathbb{S}_n$ we consider periodic boundaries.
Namely, for any $\forall x \in \mathbb{S}_n$ we define the neighbourhood
of $x$ as 
\begin{multline}
\label{eq:neighborperiodic}
\begin{split}
& \mathcal{U}(x) =  
\\  &\{  | x + s_1 + n |_{2n} - n, \, | x + s_2 + n |_{2n} - n,\,  \ldots, \, | x + s_u + n |_{2n} - n  \},
\end{split}
\end{multline}
where $| x |_{2n}$ denotes $x$ (mod $2n$).
For example, if $\mathcal{U} = \{0,1\}$, the neighbourhood of 
the site $n-1$ is $\mathcal{U}(n-1)=\{n-1, -n\}$.
For any set $K \subset S$, we define the neighbourhood of $K$
as $\mathcal{U}(K)  = \bigcup _{x \in K} \mathcal{U}(x)$.

In Percolation PCA the states of the process are synchronously updated at every site according to the following \textit{transition probability},
 \begin{equation}
\label{eq:transitionprob}
T_{x}( \, \eta^{\prime}_x  = 1 \, | \, \eta_{\mathcal{U}(x)}  \, ) =
  \begin{cases}
 \, 0   & \text{if } \eta_{\mathcal{U}(x)} = \mathbf{0}  \\
 \, p	  & \mbox{otherwise}
  \end{cases},
\end{equation}
where $p \in [0,1]$ is a free parameter.~\footnote{We use a different notation from \cite{ToomCell, ToomDiscr, ToomCont}: here $p$ corresponds to $1 - \epsilon$ and zeroes and ones are inverted.}

The temporal evolution of the process can be represented by introducing a linear operator $\mathcal{P}$, which acts on the space of probability measures $\mathcal{M}(\Sigma)$.
For any $\mu \in \mathcal{M}(\Sigma)$, we
use $\mu \mathcal{P}$ to denote the measure
obtained applying $\mathcal{P}$ to $\mu$.
By using $\overline{ \mathbf{\eta}^{\prime} }_K$ to denote the cylinder set
$\overline{ \eta^{\prime} }_K = \{ \eta \in \Sigma \, :  \eta_K =  \eta^{\prime} _K \, \}$, 
with $K \subset S$, the measure $\mu \mathcal{P}$ is defined as
\begin{equation}
\label{eq:operator}
\mu \mathcal{P} ( \, \overline{ \eta^{\prime} }_K  \, ) = \sum_{ \eta_{\mathcal{U}(K)} \in {\{ 0,1\}}^{\mathcal{U}(K) }}\mu{(\eta_{\mathcal{U}(K)})}
\,
\prod_{x \in K} T_{x}( \, \eta^{\prime}_x \, | \, \eta_{\mathcal{U}(x)} \, ).
\end{equation}

In order to characterise the time evolution of PCA, it is useful
to introduce the set of \textit{space-time realisations},
$\tilde{\Sigma}= { \{ 0, 1 \} }^{V}$, where  $V = S \times \mathbb{N}$
is the \textit{space-time} set. The elements of  $\tilde{\Sigma}$ are the realisations of the process
at all times,  $\tilde{\eta} = {( \eta^t )}_{t=0}^{\infty} \in \tilde{\Sigma}$.
We then introduce an oriented graph $\mathcal{G}_{\mathcal{U}} = ( V ,  \vec{E}_{\mathcal{U}})$,
whose edges connect any vertex $(x, t) \in V$ to the 
vertices $(k, t-1) \in V$, where $k \in \mathcal{U}(x)$.
The vertices that can be reached from $(x, t) \in V$ 
through a path on $\mathcal{G}_{\mathcal{U}}$ 
constitute the \textit{evolution cone} of $(0,t)$.
\begin{figure}
\begin{center}
\includegraphics[width=0.70\linewidth]{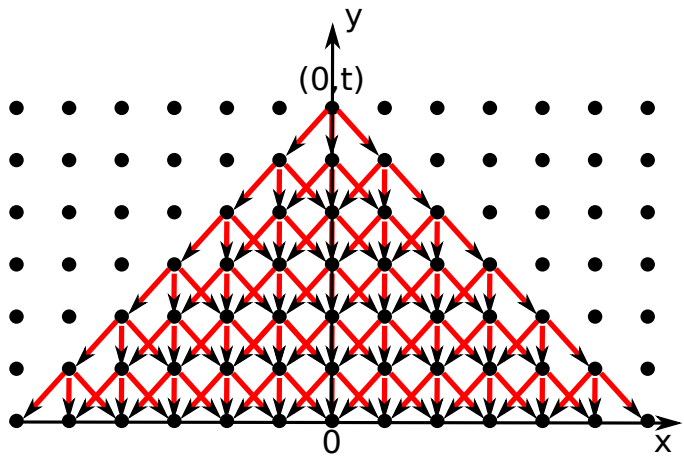}
\caption{Representation of the graph $\mathcal{G}_{\mathcal{U}}$ with neighbourhood $\mathcal{U} = \{ -1, 0, 1 \}$.
In this figure only edges between vertices belonging to the evolution cone of $(0,t)$ have been drawn. }
 \label{Fig0:Perc}
\end{center}
\end{figure}

We now introduce some definitions that will be used along
the whole article.
\begin{mydef}[Evolution Measure]
\label{def:evolutionmeasure}
Consider the Percolation PCA (\ref{eq:operator}) with $S = \mathbb{Z}$
(respectively $S = \mathbb{S}_n$ and periodic boundaries).
For every $\mu \in \mathcal{M}(\Omega)$,
we define the \textit{evolution measure} $\mathcal{E}_{\mu}$ 
(respectively $\mathcal{E}^{n}_{\mu}$) as the joint 
probability distribution of measures $\mu$, $\mu \mathcal{P}^1$,
$\mu \mathcal{P}^2$, $\ldots$.
\end{mydef}
For example, we use $\mathcal{E}^n_{\delta_\mathbf{1}}$ to denote the
 evolution measure of the Percolation PCA on 
finite space, starting from the realisation ``all ones''.
\begin{mydef}[Expectation on  the evolution space]
\label{def:expectationontheevolutiospace}
Consider the Percolation PCA  (\ref{eq:operator}) with $S = \mathbb{Z}$
(respectively $S = \mathbb{S}_n$ and periodic boundaries).
We use $\mathbb{E}_{\mu}[\, \, \cdot \, \, ]$ 
(respectively $\mathbb{E}^{(n)}_{\mu}[\, \, \cdot \, \, ]$)
to denote the expectation in relation to the evolution measure $\mathcal{E}_{\mu}$
(respectively $\mathcal{E}^n_{\mu}$).
\end{mydef}

\paragraph{Monotonicity}
It is immediate from the definition of transition probability that the Dirac measure $\delta_{\mathbf{0}}$,
where $\mathbf{0} = (0, 0, 0, \ldots )$, is stationary, i.e.  $\delta_{\mathbf{0}} = \delta_{\mathbf{0}}  \mathcal{P}$.  
Furthermore, the operator $\mathcal{P}$ of this stochastic process is \textit{monotone}.
Monotonicity of $\mathcal{P}$ means that it preserves partial order among elements of $\mathcal{M}(\Sigma)$. 
We first introduce partial order `` $\prec$ '' in $\Sigma$ by defining for any two realizations  $\eta, \eta^{\prime} \in \Sigma$,
 $\eta \prec \eta^{\prime}  \Leftrightarrow \forall x \in S$ $\eta_x \leq \eta_x^{\prime}$.
We then introduce the  functions
$\varphi : \Sigma \longmapsto \mathbb{R}$,
which only depend on a finite number of sites.
We call  $\varphi$ \textit{monotone} \textit{iff} for any $\eta, \eta^{\prime} \in \Sigma$, 
$\, \, \eta \prec \eta^{\prime} \Rightarrow \varphi(\eta) \leq \varphi(\eta^{\prime}) $.
We then introduce partial order in $\mathcal{M}(\Sigma)$ by
defining $\mu \prec  \mu^{\prime} \Leftrightarrow $  for any \textit{monotone} function $\varphi~$, $\int  \varphi \,  d\mu   \,  \leq \,    \int  \varphi \, d \mu^{\prime}$. 
Finally, we introduce an order relation between operators
and we introduce the notion of monotone operator.
\begin{mydef}[Monotone operator]
\label{def:monotoneoperator}
An operator $P \, : \mathcal{M}( \Sigma ) \longmapsto \mathcal{M}( \Sigma )$ is called monotone if for any pair of measures $\mu, \mu^{\prime} \in \mathcal{M}(\Sigma)$,  $\mu \prec \mu^{\prime} \Rightarrow \mu \mathcal{P} \prec \mu^{\prime} \mathcal{P}$.
\end{mydef}
The operator (\ref{eq:operator}) of the Percolation PCA is monotone.
This property follows from the fact that the transition probability (\ref{eq:transitionprob}) preserves order locally, i.e. for any $x \in S$,
$$
\eta^1_{\mathcal{U}(x)} \prec  \eta^2_{\mathcal{U}(x)}\, \,  \Rightarrow \, \,
 T_p ( \eta_x = 1 \, | \, \eta^1_{\mathcal{U}(x)} ) \leq T_p ( \eta_  = 1 \, | \,    \eta^2_{\mathcal{U}(x)} ) ,
$$
(see for example \cite[page 28]{ToomDiscr} for a proof of this).
Monotonicity of $\mathcal{P}$ implies that the probability measure,
\begin{equation}
\label{eq:limitmeasure}
\nu_p   : =  \lim_{ t \rightarrow \infty} \delta_{\mathbf{1}} \, \mathcal{P} ^t,
\end{equation}
exists and it is invariant.

\begin{mydef}[Critical Probability]
\label{def:critprob}
Consider the Percolation PCA on $\mathbb{Z}$ with finite neighbourhood $\mathcal{U} \subset \mathbb{Z}$.  We define the \textit{critical probability} as,
\begin{equation}
\label{eq:pc}
p_c(\mathcal{U}) = \sup_{p \in [0,1]} \{ \nu_p = \delta_{\mathbf{0}} \}.
\end{equation}
\end{mydef}
\begin{mydef}[Ergodic Operator]
\label{def:ergodic}
An operator $\mathcal{P}: \mathcal{M}(\Sigma) \rightarrow
\mathcal{M}(\Sigma)$ is ergodic if the two following conditions hold: \textbf{(a)}
there exists a unique $\varphi \in \mathcal{M}(\Sigma)$ such that $\varphi P = \varphi$
and \textbf{(b)} $\forall \mu \in \mathcal{M}(\Sigma)$, $\lim\limits_{t \rightarrow \infty} \mu P^t = \varphi$.
\end{mydef}
For any $p>p_c$ the evolution operator of the Percolation PCA is not ergodic. Indeed, in this case $\delta_{\mathbf{0}}$ and $\nu_p \neq \delta_{\mathbf{0}}$
(defined in \ref{eq:limitmeasure}) are two distinct invariant measures.
For any $p < p_c$, the Percolation PCA (\ref{eq:operator}) is ergodic.

In \cite{Stavskaja3, Uniform} it has been proved that 
$$p_c(\mathcal{U}) \in (0, 1)$$  
for the Stavskaya's process ($\mathcal{U} = \{0,1\}$) and 
a more general proof in case of general neighbourhood can be found in 
\cite{ToomDiscr}. The proofs are based on two methods widely used
in statistical mechanics, namely the counting
path method and the Peierls argument.
Our first result is stated in the following theorem and it 
involves the estimation of $p_c$.
\begin{thm}
\label{theo:maintheorem1}
Consider the Percolation PCA on $\mathbb{Z}$ with 
finite neighbourhood
$\mathcal{U} = \{ s_1, s_2, s_3, \ldots, s_u \}$,
where $s_1, s_2 \ldots s_u \in \mathbb{Z}$ are
such that $s_1 < s_2 < \ldots < s_u$.
Define  $p_1 := \frac{2}{2 + s_u - s_1}$
and $p_2$ as the unique solution in the interval $(0,1)$ of the following equation,
\begin{equation}
\label{eq:theorem1eq1}
p =   p_{1} \cdot  \frac{1}{1 - \frac{\varphi(p)}{s_u - s_1 + 2}},
\end{equation}
where $\varphi(p) = \frac{(1-p)^6 + (1-p)^{2(s_u - s_1)}}{p (2-p)}$.
Then $p_c(\mathcal{U}) \geq p_{2}$.
\end{thm}
The proof of the theorem is presented in Section \ref{sect:criticalprobabilities}.
From (\ref{eq:theorem1eq1}) it follows that  $p_{2} > p_1$, as $\varphi(p)$ is positive in $(0,1)$. Our analytical lower bound can be compared with the numerical estimations 
in the following tables. The plots of the numerical estimations can be found in the appendix of this article. The numerical estimation in case $\mathcal{U}=\{-1,0\}$ 
has been provided in \cite{Mendonca}.
\vspace{0.15cm}
\begin{center}
\begin{tabular}{|c | c | c | c |}
\hline
$\mathcal{U}$       &   $p_1$  &   $p_2$      & Num. Est. \\
\hline
$\{ -1, 0 \}$          &   $2/3$  &   $0.670$    & $0.705$   \\
\hline
$\{ -1, 0, 1 \}$       &   $1/2$  &   $0.505$    & $0.538$   \\
\hline
$\{ -1, 0, 1, 2 \}$    &   $2/5$  &   $0.407$    & $0.435$   \\
\hline 
$\{ -1, 0, 1, 2, 3 \}$ &   $1/3$  &   $0.343$    & $0.364$   \\
\hline
\end{tabular}
\end{center}
\vspace{0.25cm}
In the following table we consider neighbourhoods with $3$ elements and different values $s_u - s_1$. The comparison with the previous table shows that
our estimation is worse in case some sites between the two extremal ones are missing.
\vspace{0.15cm}
\begin{center}
\begin{tabular}{|c | c | c | c |}
\hline
$\mathcal{U}$ &   $p_1$ &  $p_2$ & Num. Est. \\
\hline
$\{ -1, 0, 1 \}$  & $1/2$ & $0.505$  & $0.538$ \\
\hline
$\{ -1, 0, 2 \}$  & $2/5$ & $0.407$  & $0.490$  \\
\hline
$\{ -1, 0, 3 \}$  & $1/3$ &  $0.343$  & $0.470$ \\
\hline
\end{tabular}
\end{center}
\vspace{0.25cm}

Our main result is stated in Theorem 
\ref{theo:maintheorem2} and it involves the convergence time 
into the absorbing state of the Percolation PCA
with finite space and periodic boundaries, as
defined at the beginning of this section.

When $S$ is finite, the process is always ergodic (Definition \ref{def:ergodic}).
Indeed, for any realisation of the process $\eta^t \in \Sigma$ at time $t$, 
the probability that $\eta^{t+1} = \mbox{ ``all zeroes''}$
is bounded from below by the constant ${(1-p)}^{|S|}$.
This implies that there exists almost surely a finite time $\tau \in \mathbb{N}$ such that $\eta^t = \mbox{ ``all zeroes''}$ for all $t \geq \tau$.
Hence, for any $\mu \in \mathcal{M}({\Sigma})$, $\lim\limits_{t \rightarrow \infty} \mu \mathcal{P}^t = \delta_{\mathbf{0}}$.

In order to estimate the convergence time into the
absorbing state
we define the \textit{absorption-time} $\tau \in \mathbb{N}_0$,
representing the first 
time all sites in the segment $\{-k, -k+1, \ldots, k-1\}$ are in state zero
for $\eta^\tau$.
\begin{mydef}
\label{def:absorbtime}
For all $k \in \mathbb{N}$, we call the absorption time of the 
interval $\{-k, -k+1, \ldots, k-1 \}$ 
the random variable $\tau_k : \tilde{\Sigma} \rightarrow \mathbb{N}$,
\begin{equation}
\label{eq:absorbingtime}
\tau_k(\tilde{\eta}) = \min \{ t \in \mathbb{N}_0 \mbox{ s.t. }  
\tilde{\eta}_x^t = 0 \, \, \, \, \forall x  \in \{-k, -k+1, \ldots, k-1 \}  \}.
\end{equation}
\end{mydef}
In case $S = \mathbb{S}_n$, this random variable is well defined only if 
$k \leq n$. 

We recall Definitions \ref{def:evolutionmeasure} and \ref{def:expectationontheevolutiospace} and we state our main result.
\begin{thm}[]
\label{theo:maintheorem2}
Consider the Percolation PCA with space $\mathbb{S}_n$, periodic boundaries and finite neighbourhood $\mathcal{U} = \{s_1, s_2, \ldots s_u\}$, where $s_1, s_2, \ldots s_u$ are some distinct elements of $\mathbb{Z}$.
For every $p \in [0,1]$ there exist $n_0 \in \mathbb{N}$ and some positive constants $K_1, K_2, K_3, K_4, c_1, c_2, c_3, c_4$ (dependent on $p$) such that
for all $n > n_0$,
\begin{center}
\begin{enumerate}
\item[a)]  if $p<p_c$,  ~ 
$K_1 \log (c_1 \, n)  \leq  \mathbb{E}_{{\delta}_{\mathbf{1}}}^{(n)} [ \tau_n] \leq K_2 \log (c_2 \, n)  $,
\item[b)]  if $p>p_c$, ~ 
$K_3 \exp (c_3 \, n) \leq  \mathbb{E}_{\delta_{\mathbf{1}}}^{(n)} [ \tau_n] \leq K_4 \exp (c_4 \, n)$.
\end{enumerate}
\end{center}
\end{thm}
The proof of the theorem is presented in Section \ref{sect:timeofconvergence}.

\section{Critical Probabilities}
\label{sect:criticalprobabilities}
In this section we prove Theorem \ref{theo:maintheorem1},
which provides a lower bound for $p_c$ as a function of
the neighbourhood.  The proof of Theorem \ref{theo:maintheorem1} 
requires Lemma \ref{lem:randomwalk} and  Propositions \ref{prop:containedinto} and \ref{prop:estimations}, which are stated in this section.
\begin{prop}
\label{prop:containedinto}
Consider two Percolation PCA in $\mathbb{Z}$ 
with neighbourhoods respectively $\mathcal{U}$ and 
$\mathcal{U}^{\prime}$,
both finite subsets of $\mathbb{Z}$,
such that $\mathcal{U} \subset \mathcal{U}^{\prime}$.
Then $$p_c (\mathcal{U}) \geq p_c (\mathcal{U}^{\prime}).$$ 
\end{prop}
\begin{proof}
From Proposition \ref{prop:percolation}, stated in Section \ref{sect:ergodicity}, and from the fact that the edge set of the graph $\mathcal{G}_{\mathcal{U}}$ is a subset of the edge set of the graph $\mathcal{G}_{\mathcal{U}^{\prime}}$, it follows that $\forall x \in \mathbb{Z}$, $\forall t \in \mathbb{N}_0$,
$\delta_{\mathbf{1}} \mathcal{P}^t_{\mathcal{U}} (\eta_x = 0  )
\geq \delta_{\mathbf{1}} \mathcal{P}^t_{\mathcal{U}^{\prime}} (\eta_x = 0 )$ (we added the subscript to the operator in order to distinguish between the two neighbourhoods). From Definition \ref{def:critprob} it follows that, 
$$\lim\limits_{t \rightarrow \infty} \delta_{\mathbf{1}} \mathcal{P}^t_{\mathcal{U}} (\eta_x = 0  ) < 1 \implies p \geq p_c(\mathcal{U}).$$ 
Hence, $p_c (\mathcal{U}) \geq p_c (\mathcal{U}^{\prime})$. 
\end{proof}
We introduce some notation.
\begin{mydef}[Massif of zeros] We call a segment of $\mathbb{Z}$,
$$
\{k, k+1, \ldots, k+\ell-1\} \subset \mathbb{Z}
$$
a \textit{massif of zeros} of length $\ell$ for a given $\eta \in \Sigma$, if 
$\eta_{k-1}=\eta_{k+\ell}=1$ and $\eta_{k}= \ldots = \eta_{k+\ell-1} =0$.
\end{mydef}
We use $[[a,b]]$ to denote the set of integers in the interval $[a,b]$.
We use $\eta^0 \in \Sigma$ to denote the initial realisation of the 
Percolation PCA (namely the initial probability measure 
is $\delta_{\eta^0}$) and $\eta^1$, $\eta^2$, $\ldots$
the random realisations of the process at different times.
For every $T \in \mathbb{Z}_+$, we introduce the following
notation (the role of $T$ will be clear later).
For every $\eta^0$, we enumerate somehow the massifs 
of zeros of length larger or equal to  $T (s_u - s_1)$.
This means that we assign to every massif a label $k \in \mathcal{I}$,
where $\mathcal{I} \subset \mathbb{N}_0$ is the set of labels.
We denote by $R_k^0$ and by $L_k^0$ respectively the rightmost and 
the leftmost zero of the $k$-th massif.
We observe that, by definition of the transition probability for the process (\ref{eq:transitionprob}), such massifs cannot have disappeared at time $t \leq T$ (see also Figure \ref{Fig:2}).
\begin{figure}
\centering        
\includegraphics[width=\textwidth]{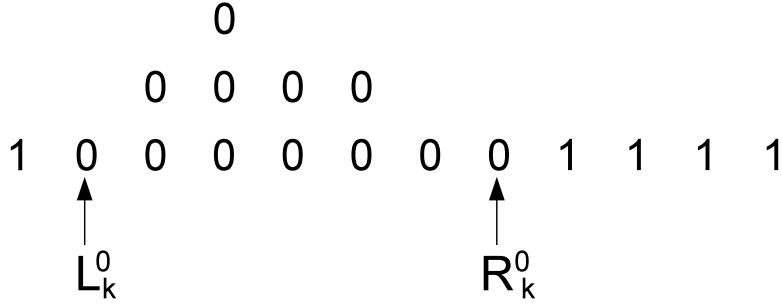}                		
\caption{In this example we consider a Percolation PCA  with $\mathcal{U}=\{ -1, 0, 1, 2 \}$. If the process starts from a realisation having a massif of zeros in 
$\{ L_k^0, \ldots, R_k^0 \}$ with $R_k^0 - L_k^0 \geq T (s_u - s_1)$, as in the figure, then the state of every site in $\{ L_k^0-s_1, L_k^0-s_1+1,\ldots,  R_k^0-s_u \}$ for the random realisation at time $1$ and in $\{ L_k^0\, - 2 \, s_1, L_k^0\, -\, 2  s_1 +1,\ldots ,  R_k^0\, -\, 2 s_u \}$ for the random realisation at time $2$ is $0$
almost surely.} 
\label{Fig:2}   
\end{figure}
For every $k \in \mathcal{I}$, we define the random variables $(R_k^t)_{t \geq 1}$ and $(L_k^t)_{t \geq 1}$ using recursion.
Namely, $\forall k \in \mathcal{I}$, $\forall t \in \mathbb{Z}_+$,
\begin{equation}
\label{eq:R}
R_k^t :=
\begin{cases}
\max\limits_{x \in \mathbb{Z}} \{\forall y \in [[L_k^{t-1} -  s_1, x]], \, \,  \eta^t_{y}  = 0
 \} & \mbox{ if }  R_k^{t-1} - L_k^{t-1} \geq  (s_u - s_1) \\
- \infty	  & \mbox{otherwise}\\
\end{cases}
\end{equation}
\begin{equation}
\label{eq:L}
L_k^t :=
\begin{cases}
\min\limits_{x \in \mathbb{Z}} \{\forall y \in [[x, R_k^{t-1}- s_u]], \, \,  \eta^t_{y}  = 0
 \} & \mbox{ if }  R_k^{t-1} - L_k^{t-1} \geq  (s_u - s_1) \\
+ \infty	  & \mbox{otherwise}\\
\end{cases}.
\end{equation}
Namely $R_k^t$ and $L_k^t$ keep track of the temporal evolution of
two extremal sites of the $k$-th massif. If the distance between such sites
at a given time is less than $ (s_u - s_1)$, then at all subsequent times $R_k^t= -\infty$ and $L_k^t = + \infty$.
Instead if at time $R_k^t - L_k^t \geq (s_u - s_1)$, then at time $t+1$ the massif still exists almost surely.
Note that it might happen that two or more massifs merge 
at a certain time. In this case more than one label is used to denote the same massif.
The next lemma shows that if the massifs of zeros are ``on average'' expanding,
then the state of the system at infinite time is zero almost surely.
As this happens independently on the initial realisation, the process is ergodic.
\begin{lem}
\label{lem:randomwalk}
For every $T \in \mathbb{Z}_+$, the following statement holds.
If there exist two families of independent and identically distributed random variables $({\pi}_k^{i})_{i , k \in \mathbb{N}}$, 
$({\xi}_k^{i})_{i, k \in \mathbb{N}}$,
such that 
$\forall \eta^0 \in \Sigma$, $\forall k \in \mathcal{I}, \forall i \in \mathbb{N}_0$, $\forall j \in \mathbb{Z}$, the conditions (\ref{eq:lemcond1}),  (\ref{eq:lemcond2}),
(\ref{eq:lemcond3}) hold,
\begin{align}
\label{eq:lemcond1}
P( {\pi}_k^{i } \geq j) & \leq  \mathcal{E}_{\delta_{\eta^0}} ( R_k^{i T + T} - 
R_k^{i T}  \geq j \, | \, R_k^{i T} - L_k^{i T} \geq  T (s_u - s_1) ) \\
\label{eq:lemcond2}
P( {\xi}_k^{i } \leq -j)  & \leq \mathcal{E}_{\delta_{\eta^0}} ( 
L_k^{i T + T} - L_k^{i T} \leq - j \, | \,  R_k^{i T} - L_k^{i T} \geq  T (s_u - s_1))
\end{align}
\begin{equation} 
\label{eq:lemcond3}
E[\pi_1^1] > E[\xi_1^1]
 \end{equation}
then $\forall \mu \in \mathcal{M}(\Sigma)$,
\begin{equation}
\label{eq:lemmaconcl}
\lim\limits_{t \rightarrow \infty} \mu \mathcal{P}^t =  \delta_{\mathbf{0}}.
\end{equation}
\end{lem}
In the statement of the lemma $P (\, \cdot \, )$ denotes the probability distribution of the random variables $\pi_k^i$ or $\xi_k^i$. Such random variables stochastically dominate from below the change of position of the rightmost and leftmost site of the massif every $T$ steps. We also recall that $\mathcal{E}_{\delta_{\eta^0}}$ has been defined in Definition \ref{def:evolutionmeasure}.  
The proof of the lemma is similar to the proof of Proposition 6.4 in \cite{ToomDiscr}.
\begin{proof}
It is sufficient to prove that $\forall \eta^0 \in \Sigma$, $\forall x \in \mathbb{Z}$, $\forall \epsilon >0$,
$\exists \, t_0$ such that  
\begin{equation}
\label{eq:lemtoprove}
\forall t > t_0, \, \, \, \, \, \, 
\delta_{\eta^0} \mathcal{P}^t (\eta_x=0) \geq 1 - \epsilon,
\end{equation}
from which condition (\ref{eq:lemmaconcl}) follows.

We define $c_1:=\frac{ 3 E[{\xi}_1^1] + E[{\eta}_1^1]}{4}$
and $c_2:=\frac{  E[{\xi}_1^1] + 3 E[{\eta}_1^1]}{4}$,
where $E[ \cdot ]$ denotes the expectation, and we observe that
if (\ref{eq:lemcond3}) holds, then $c_2 > c_1$.
Then for every $\eta^0 \in \Sigma$, 
$\forall k \in \mathcal{I}$,
$\forall n, m \in \mathbb{Z}^+$,
$\forall i_0, j_0 \in \mathbb{Z}$ such that
$j_0 - i_0 \geq  T (s_u - s_1) + m + n$,
there exists two constants $u , v \in [0,1)$ such that,
\begin{multline}
\label{eq:lemeq1}
\begin{split}
  \mathcal{E}_{\delta_{\eta^0}} 
( \forall i \geq & 1, \, \, \, 
L_k^{i T} \leq  \,  c_1 (i-1) T  + L_k^T + n, \, \, 
R_k^{i T} \geq c_2 (i -1) T  + R_k^T - m
\, \mid\, L_k^T=i_0, R_k^T=j_0) \\
\geq 
 & \, \,  P  ( \forall i \geq 1, 
\, \, \, 
 \sum\limits_{j=1}^i {\xi}_k^j \leq c_1 (i-1) T+ i_0 + n, \, \, 
 \sum\limits_{j=1}^i {\pi}_k^j \geq c_2 (i -1) T + j_0 - m)  \\
\geq   & \, \, 1 - u^m - v^n.
\end{split}
\end{multline}
In the previous expression $P( \cdot )$ denotes the probability
measure defined on the space of outcomes of the sum of the increments
${\xi}_k^t$, ${\pi}_k^t$.
The first inequality follows from  (\ref{eq:lemcond1}) and (\ref{eq:lemcond2}).
The second inequality follows from the properties of the one dimensional random walk, observing that by definition $E[\xi_k^i] < c_1$ and $E[\pi_k^i] > c_2$.
The two constants $u$ and $v$ depend on the probability distribution of the increments of the random walk.

We observe that if for all $t$ multiple of $T$, $R^t_k \geq c_2 t + j_0 -m$ and $L_k^t \leq c_1 t + i_0 + n$ (event in the first expression in (\ref{eq:lemeq1})) , then for all $t \in \mathbb{N}_0$,
$R^t_k \geq c_2 t + j_0 -m - s_u T $ and $L_k^t \leq c_1 t + i_0 + n - s_1 T$.  Hence, the state of all sites in the space-time region $Y^{m,n}_{[i_0, j_0]} := 
\{ (x,t) \, \, : \, \,  t \in \mathbb{Z}_+ \mbox{ and } c_1 t + n + i_0 -T s_1   \leq x \leq c_2 t - m + j_0  - T s_u \}$ is zero. This region is represented in Figure \ref{Fig:lemm} on the left. This follows from the observation that by definition of transition probability of the Percolation PCA the following property holds, namely,
\begin{equation}
R_k^t - L_k^t \geq  T (s_u - s_1) \implies \forall q < T, \, \, R_k^{t+q}  \geq  R_k^t - q s_u,  L_k^{t+q}  \leq   L_k^t - q s_1.
\end{equation}
\begin{figure}
\centering        
\includegraphics[width=0.45\textwidth]{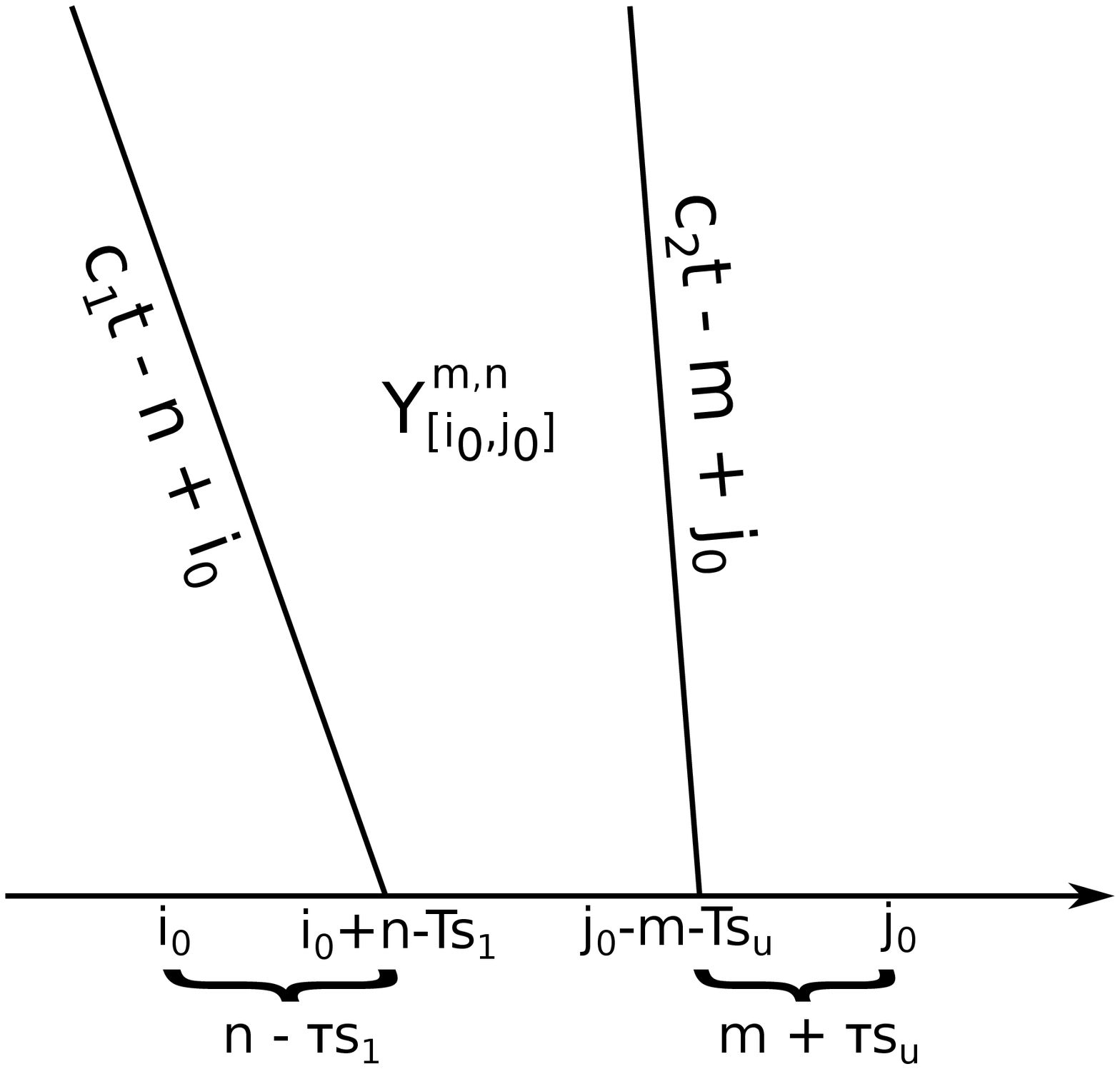} 
\includegraphics[width=0.45\textwidth]{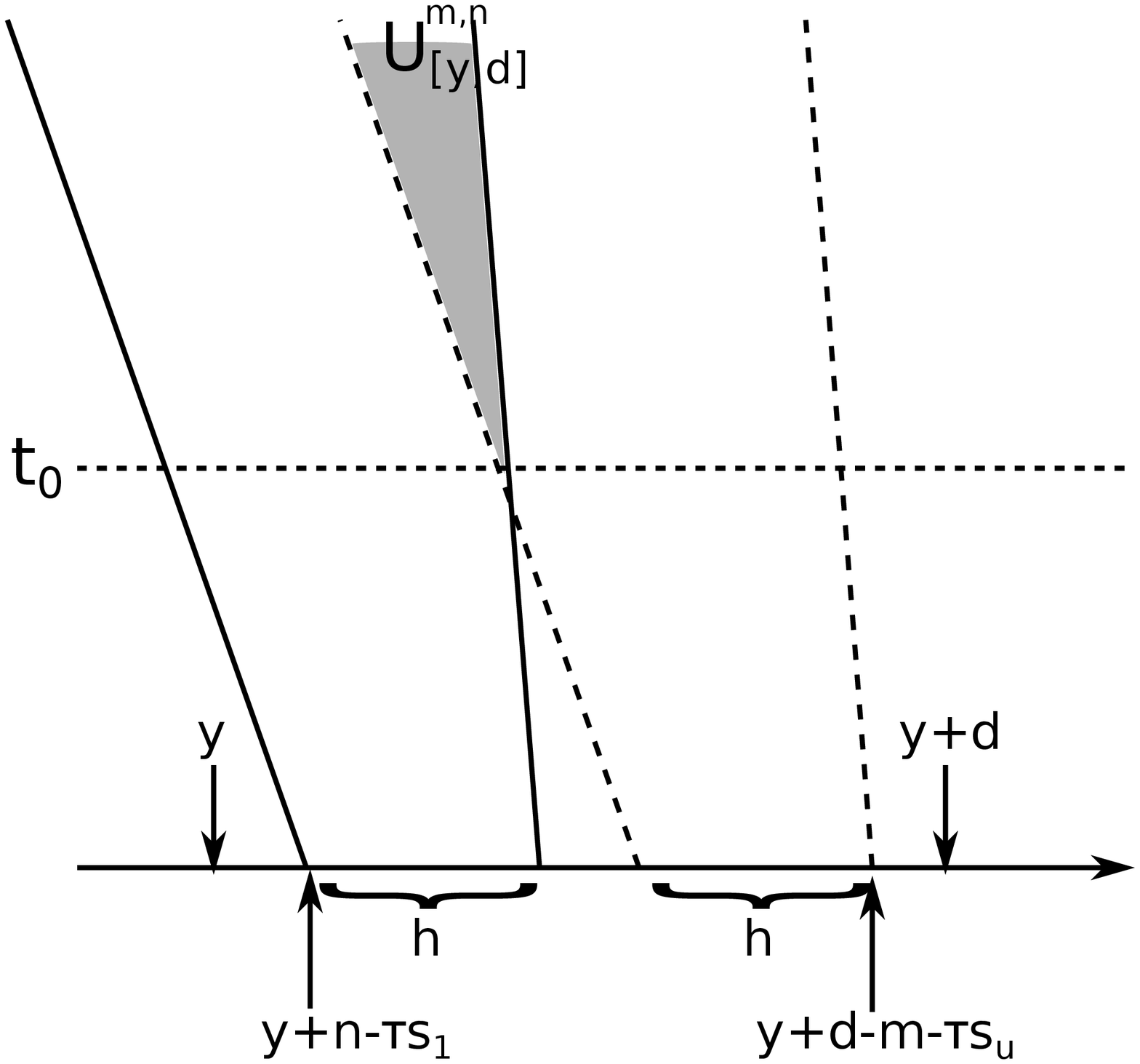}                		
\caption{The variable $h$ on the right is defined as $h:= j_0 - i_0 - T(s_u - s_1) - n - m$.}
\label{Fig:lemm}   
\end{figure}

Furthermore we observe that $\forall \eta^0 \in \Sigma$, $\forall x \in \mathbb{Z}$, $\forall n, m \in \mathbb{Z}_+$, 
the measure $\delta_{\eta^0}\mathcal{P}^{T}$ is such that the probability that
there exists a massif of zeros of length $j_0 - i_0 \geq  T (s_u - s_1) + n + m$
in $[y, d]$ goes to $1$ as $d \rightarrow \infty$. 
We choose then $n$ and $m$ such that $u^m + v^n < \frac{\epsilon}{2}$
and  $d$ large enough such that such probability is larger than
$1 - \frac{\epsilon}{2}$ for all $y$.

Simple geometrical considerations show that for any $y \in \mathbb{Z}$, $d \in \mathbb{Z}_+$, all regions $Y^{n,m}_{[i_0, j_0]}$, where $[i_0, j_0] \subset [y, y+d]$, have a non empty common region (dark region in Figure \ref{Fig:lemm} - right). We call $U^{m,n}_{[y,d]}$ this region.
From (\ref{eq:lemeq1}) and from the previous observations the following property holds,
\begin{equation}
\mathcal{E}_{\delta_{\eta^0}}  \left (  \, \forall (x,t) \in U_{[y,d]},  \, 
\eta_x^t =0 \, \, \right ) > 1 - \epsilon.
\end{equation}
Choosing $y$ and $d$ such that $(x, t) \in U_{[y,d]}$ implies (\ref{eq:lemtoprove}).
\end{proof}

\begin{proof}[\textbf{Proof of Theorem \ref{theo:maintheorem1}}]
We provide a lower bound for the critical probability of the Percolation PCA with neighbourhood
\begin{equation}
\label{eq:containedintoex}
\mathcal{U} := \{ s_1, s_1 + 1,  \ldots, s_{u}-1, s_u\},
\end{equation}
i.e. all elements between the two extremal ones are present. Our bound is a function of $s_u - s_1$. By Proposition \ref{prop:containedinto} such bound holds also for Percolation PCA with neighbourhood obtained removing some sites from 
(\ref{eq:containedintoex}).

The proof of the theorem is based on an application of Lemma \ref{lem:randomwalk}.
We fix a value $T \in \mathbb{Z}_+$ and by using the monotonicity property of the Percolation PCA, we define the random variables $\pi_k^t$ and $\xi_k^t$, whose probability distribution satisfies $\forall p \in [0,1]$ the conditions (\ref{eq:lemcond1}) and (\ref{eq:lemcond2}) of Lemma \ref{lem:randomwalk}.
We define,
$$p_T := \max\limits_{p \in [0,1]} \{ E[\pi^1_1] > E[\xi^1_1] \}.$$
From Lemma \ref{lem:randomwalk}, for all $p \geq p_T$ the Percolation PCA
is ergodic. From Definitions \ref{def:critprob} and \ref{def:ergodic},  $p_T \leq p_c$. We fix first $T=1$ and we derive $p_1$, later we consider 
$T=2$ and we derive $p_2$. Both $p_1$ and $p_2$ appear in the statement 
of the theorem. Higher is the value of  $T$ considered, more challenging is the estimation of $p_T$, as this involves the characterization of the 
increments of $L_k^t, R_k^t$ over a larger time interval.

\paragraph*{}
Fix then an integer $T \in \mathbb{Z}^+$ and consider an initial realisation $\eta^0 \in \Sigma$. Enumerate somehow the massifs of zeros having length not smaller than $T (s_u - s_1)$ and recall the definitions of the random variables
$R_k^t$, $L_k^t$,  $t \in \mathbb{N}_0$, $k \in \mathcal{I}$, provided before the statement of Lemma \ref{lem:randomwalk}.
For any $A \subset \mathbb{Z}$,  let $\mathbbm{1}^t_A : \tilde{\Sigma} \rightarrow \tilde{\Sigma}$ be the function that is equal to $1$ if the state of all sites in $A$ at time $t \in \mathbb{N}_0$  is zero and zero otherwise. Let
$\mathbbm{1}_A : {\Sigma} \rightarrow {\Sigma}$
(without the superscript) be the function that is equal to $1$ if the state of all
sites in $A$ is zero and zero otherwise. Observe that $1 - \mathbbm{1}^t_A$ and $1 - \mathbbm{1}_A$ are monotone functions.
Let also $\rho(x,y) \in \Sigma$ be the realisation having zeros in $[[x,y]]$ and 
ones everywhere else.
Then $\forall   \eta^0, \forall \eta \in \Sigma$, $\forall t \in \mathbb{Z}_+$, $\forall k \in \mathcal{I}$, $\forall x, y \in \mathbb{Z}$ such that 
$ y - x \geq T (s_u - s_1)$, $\forall j \in \mathbb{Z}_0$,
the following relations hold,
\begin{align}
\label{eq:stochdomination1}
\mathcal{E}_{\delta_{\eta^0}} ( & R_k^{t + T} - 
R_k^{t}   \geq j \, |  \, R_k^{t} =y,  L_k^{t} =x, \eta^t = \eta ) \\
\label{eq:stochdomination2}
& = \mathcal{E}_{\delta_{\eta^0}} ( \mathbbm{1}^{t + T}_{ [[x - T s_1, y - T s_u + j ]]} \, |  \, R_k^{t} =y,  L_k^{t} =x, \eta^t = \eta ) \\
\label{eq:stochdomination3}
& =\delta_{\eta} \mathcal{P}^T ( \mathbbm{1}_{ [[x - T s_1, y - T s_u + j ]]} ) \\
\label{eq:stochdomination4}
& \geq \delta_{\rho(x,y)} \mathcal{P}^T ( \mathbbm{1}_{ [[x - T s_1, y - T s_u + j ]]} ). 
\end{align}
Equation (\ref{eq:stochdomination2}) follows from the definition of
$R_k^t$, equation (\ref{eq:stochdomination3}) follows from the Markov
property of the probabilistic cellular automaton,
inequality (\ref{eq:stochdomination4}) follows from the monotonicity property of the Percolation PCA, as any realisation $\eta \in \Sigma$ having a massif of zeros in $[[x,y]]$ is such that $\eta \prec \rho(x,y)$.
Similarly, 
\begin{align}
\label{eq:stochdomination1L}
\mathcal{E}_{\delta_{\eta^0}} ( & L_k^{t + T} - 
L_k^{t}   \leq - j \, |  \, R_k^{t} =y,\,  L_k^{t} =x,\, \eta^t = \eta ) \\
\label{eq:stochdomination3L}
& =\delta_{\eta} \mathcal{P}^T ( \mathbbm{1}_{ [[x - T s_1 - j , y - T s_u ]]} ) \\
\label{eq:stochdomination4L}
& \geq \delta_{\rho(x,y)} \mathcal{P}^T ( \mathbbm{1}_{ [[x - T s_1 - j , y - T s_u ]]} ). 
\end{align}
We also observe that from the definition of transition probability of the Percolation PCA,  the quantities (\ref{eq:stochdomination4}) and (\ref{eq:stochdomination4L}) do not depend on the sites $x,y \in \mathbb{Z}$, as long as  $y- x \geq T (s_u - s_1)$.
Thus, we provide the following definitions of the probability distribution of the random variables $\pi_k^t$, $\xi_k^t$. Namely, fix $y$ and $x$ such that $y - x \geq T (s_u - s_1)$ and $\forall k, \forall t \in \mathbb{N}_0$ we define,
\begin{align}
\label{eq:pik}
P(\pi_k^t \geq j) & :=     \delta_{\rho(x,y)} \mathcal{P}^T ( \mathbbm{1}_{ [[x - T s_1,\, y - T s_u + j ]]})  \\
\label{eq:xik}
P(\xi_k^t \leq -j) & :=  \delta_{\rho(x,y)} \mathcal{P}^T ( \mathbbm{1}_{ [[x - T s_1 - j,\, y - T s_u ]]}).
\end{align}
With this definition, from (\ref{eq:stochdomination1}) - (\ref{eq:stochdomination4L})
the first two conditions  of
Lemma \ref{lem:randomwalk}, namely (\ref{eq:lemcond1}) and (\ref{eq:lemcond2}), are satisfied.
The maximum among all $p \in [0,1]$ such that condition (\ref{eq:lemcond3})
is satisfied is $p_{T} \leq p_c$.

We fix now $T=1$ and we provide an estimation for (\ref{eq:pik})
and (\ref{eq:xik}) for any $j \in \mathbb{Z}$. After this we determine which values of $p$ satisfy (\ref{eq:lemcond3}).
We consider the Percolation PCA starting from initial realisation $\rho(x,y)$ and we assign assign the label $1$ to the unique massif of zeros, namely $R_1^0 = y$ and $L_1^0=x$.
We recall that by definition,
\begin{equation}
R^1_1  \geq j + R_1^0 - s_u \iff \forall z \in [[L_1^0-s_1,  R_1^0  - s_u + j]],  
\, \, \eta_z^1 = 0
\end{equation}
(see also Figure \ref{Fig:3}).
Hence, 
$\forall j \in \mathbb{N}_0$, 
\begin{equation}
\delta_{\rho(x,y)} \mathcal{P} ( R^1_1  \geq j + R_1^0 - s_u) = (1-p)^j.
\end{equation}
This bound is obtained considering that almost surely 
$\forall z \in [[ x - s_1, y  - s_u]]$,
$\eta^1_z = 0$ and that $\forall z \in [[y - s_u + 1,  y  -s_u  +j]]$,
independently $\delta_{\rho(x,y)} \mathcal{P} (\eta_z=0) = (1-p)$.
Analogously, $\forall j \in \mathbb{N}_0$, 
\begin{equation}
\delta_{\rho(x,y)} \mathcal{P} ( L^1_1  \leq -j + L^0_1  - s_1) = (1-p)^j.
\end{equation}
Thus for all $j \in \mathbb{N}_0$, we define the probability distributions of $\pi^t_1$ and $\xi^t_1$ 
respectively as,
\begin{align}
\label{eq:tau1-1}
P( \pi^t_k  \geq j - s_u ) & := (1-p)^j \\
P( \xi^t_k  \leq - j - s_1) & := (1-p)^j.
\end{align}
With this definition, from the relations 
(\ref{eq:stochdomination1}) - (\ref{eq:stochdomination4L}),
the relations (\ref{eq:lemcond1}) and (\ref{eq:lemcond2}) are satisfied.
It remains to determine for which values of $p \in [0,1]$ the second condition of Lemma \ref{lem:randomwalk} holds.
By a simple computation, 
\begin{align}
\label{eq:tau1-1}
E[ \pi_1 ] = \frac{1-p}{p}  - s_u, \\
E[ \xi_1 ] =   \frac{1-p}{p} - s_1.
\end{align}
and 
\begin{equation}
\label{eq:tau1-2}
E[ \xi_1]  - E[\pi_1]  \geq 0 \, \,  \iff \, \,  p \geq p_1,
\end{equation}
where  $p_1 := \frac{2}{2 + s_u - s_1}$ appears on the statement of the theorem. Thus we proved that $p_c \geq p_1$.
\begin{figure}
\centering        
\includegraphics[width=\textwidth]{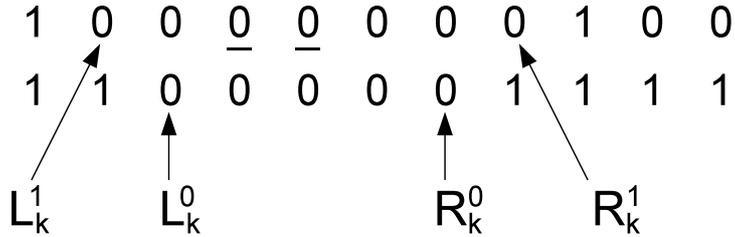}                		
\caption{In this example we consider a Percolation PCA with
$\mathcal{U}=\{ -1, 0, 1, 2 \}$. If the process starts from the realisation represented in the figure (row below), then the state of the sites above the small
horizontal ball is almost surely $0$ at time $1$ (row above).}
\label{Fig:3}   
\end{figure}

\paragraph*{} We fix now $T=2$ and we use the same argument.
Namely we consider the Percolation PCA starting from initial realisation
$\rho(x,y) \in \Sigma$ such that $y-x \geq 2 (s_u - s_1)$
and we assign label $1$ to the unique massif of zeros of $\rho(x,y)$.
We recall that by definition of $R_1^2$, 
$$R_1^2  \geq  j  +  R_1^0  - 2 s_u \iff \forall z \in  [[L_1^0 - 2 s_1, R_1^0 - 2 s_u + j]],  \, \, \, \eta^2_z=0.$$
From definition (\ref{eq:operator}) it follows that
\begin{multline}
\begin{split}
\label{eq:tau2-1}
\delta_{\rho(x,y)} \mathcal{P}^2(  R_1^2  \geq  j  +  y  - 2 s_u  )  & =     
\sum_{ \eta^1 \in A_R}
\delta_{\rho(x,y)} \mathcal{P}  ( C_{\eta^{1}}  )  \\ &
\prod_{z \in  [[ y -2 s_u,y- 2 s_u + j ]]}
T (    \eta^{2}_z=0  | \eta^1_{\mathcal{U}(z)} ),
\end{split}
\end{multline}
where $A_R: = \{0,1 \}^{[[y- 2 s_u + s_1, y  - s_u +j]]}$ and 
$C^R_{\eta^1} = \{ \eta^{\prime} \in \Sigma \mbox{ s.t. } 
\forall z \in  [[y - 2 s_u + s_1, y  - s_u +j]], 
\, \, \, \eta^{\prime}_z = \eta^{1}_z \,  \}$.
The sum is reduced to the elements of $A_R \subset \Sigma$
because the states of the sites in the interval  $[[y - 2 s_u, y  - 2s_u +j]]$ for $\eta^2$ depend only on the states of the sites in the finite interval  
$[[  y- 2s_u + s_1, y - s_u +j ]]$ for $\eta^1$.
A similar expression holds for the random variable $L_1^2$,
 \begin{multline}
\begin{split}
\label{eq:tau2-left}
\delta_{\rho(x,y)} \mathcal{P}^2(  L_1^2  \leq  - j  + x - 2 s_1  )  & =     
\sum_{ \eta^1 \in A_L}
\delta_{\rho(x,y)} \mathcal{P}  ( C^L_{\eta^{1}}  )  \\ &
\prod_{z \in  [[ x -2 s_1 - j, x- 2 s_u ]]}
T (    \eta^{2}_z=0  | \eta^1_{\mathcal{U}(z)} ),
\end{split}
\end{multline}
where $A_L: = \{0,1 \}^{[[x -  s_1 - j, x  - 2s_1 + s_u]]}$
and $C^L_{\eta^1} = \{ \eta^{\prime} \in \Sigma \mbox{ s.t. }
\forall z \in [[x - 2 s_1 - j, x  - 2s_1 + s_u]], \, \, \, \eta^{\prime}_z = \eta^{1}_z \}$.

The exact computation of the left hand side of (\ref{eq:tau2-1})
and of (\ref{eq:tau2-left}) for any $j$ is a difficult combinatorial problem, as for each of the $2^j$ possible realizations one should determine the corresponding product of transition probabilities. We present our estimations in the following proposition.
\begin{prop}
\label{prop:estimations}
Consider the realisation $\rho(x,y) \in \Sigma$ which has zeros in  $[[x,y]]$ and
ones everywhere else, where  $x,y \in \mathbb{Z}$ are such that $y-x \geq 2 (s_u - s_1)$. Assign label $1$ to the unique massif of zeros of $\rho(x,y)$ and recall
the definition of $(R_1^t)_{t \in \mathbb{N}_0}$, $(L_1^t)_{t \in \mathbb{N}_0}$.
Then,
\begin{align}
\label{eq:R20}
& \delta_{\rho(x,y)} \mathcal{P}^2  (  R_1^2 \geq  R_1^0 - 2s_u)  
=1,   \\
\label{eq:L20}
& \delta_{\rho(x,y)} \mathcal{P}^2  (  L_1^2 \leq  L_1^0 -  2s_1) 
=1, \\
\label{eq:R21}
& \delta_{\rho(x,y)} \mathcal{P}^2  (  R_1^2 \geq  1 +R_1^0 - 2s_u)  
\geq 1 - p^2,   \\
\label{eq:L21}
& \delta_{\rho(x,y)} \mathcal{P}^2  (  L_1^2 \leq  -1 + L_1^0 -  2s_1) 
\geq1 - p^2,   \\
\label{eq:R22}
& \delta_{\rho(x,y)} \mathcal{P}^2  (  R_1^2 \geq  2 + R_k^0 - 2s_u)  
\geq (1-p)^2 (1 + 2p), \\
\label{eq:L22}
& \delta_{\rho(x,y)} \mathcal{P}^2  ( L_1^2 \leq -2  + L_1^0 -  2s_1)  
\geq (1-p)^2 (1 + 2p),
\end{align}
 for any $3 \leq j \leq s_u - s_1$,
\begin{align}
\label{eq:tau2-estj}
& \delta_{\rho(x,y)} \mathcal{P}^2  ( R_1^2 \geq j +  R_1^0 - 2s_u   ) 
\geq  j \,  p (1-p)^j  + (1-p)^{j} + (1-p)^{2j},   \\
& \delta_{\rho(x,y)} \mathcal{P}^2  ( L_1^2 \leq -j + L_1^0 - 2s_1   ) 
 \geq  j \,  p (1-p)^j  + (1-p)^{j} + (1-p)^{2j},  
\end{align}
 and for any $j > s_u - s_1$,
\begin{equation}
\begin{split}
\label{eq:tau2-estj2}
\delta_{\rho(x,y)} \mathcal{P}^2  ( R_1^2 \geq j +  R_1^0 - 2s_u   ) 
& \geq   j \,  p (1-p)^j  +   (1-p)^j \\  & +  p (1-p)^{j+s_u - s_1} (j - s_u + s_1 - \frac{1}{p}) \\ & + 2 (1-p)^{2j},
\end{split}
\end{equation}
\begin{equation}
\begin{split}
\label{eq:tau2-estj2L}
\delta_{\rho(x,y)} \mathcal{P}^2  ( L_1^2 \leq -j + L_1^0 - 2s_1   ) 
& \geq   j \,  p (1-p)^j  +   (1-p)^j \\  & +  p (1-p)^{j+s_u - s_1} (j - s_u + s_1 - \frac{1}{p}) \\ & + 2 (1-p)^{2j}.
\end{split}
\end{equation}
\end{prop}
We postpone the proof of Proposition \ref{prop:estimations} to the next paragraph and we conclude the proof of Theorem \ref{theo:maintheorem1}.
We use the lower bounds provided in the proposition to define the probability distribution of the random variables $\pi^k_1$, $\xi^k_1$.
Namely $\forall j \in \mathbb{N}_0$, we define the probability of the event $\{\pi^t_1 \geq j - s_u\}$ (respectively $\{\xi^t_1 \leq -j - s_1\}$) as the lower bound of the probability of the event  $\{R_1^2 \geq j + R_1^0   - s_u)\}$ 
(respectively $\{L_1^2 \leq - j + L_1^0   - s_1)\}$) provided in the proposition.
With such definition, the expectation of the random variables $\xi^t_k$, $\pi^t_k$ is equal to
\begin{equation}
\label{eq:tau2-4}
E [ \pi^t_k]  = 2 \, \frac{( 1 - p ) }{p} - 2 s_u + \frac{(1-p)^6 + (1-p)^{2s_u - 2s_1 + 2}}{p (2-p)},
\end{equation}
\begin{equation}
\label{eq:tau2-5}
E [ \xi^t_k]  = - 2 \, \frac{ ( 1 - p ) }{p} - 2 s_1 - \frac{(1-p)^6 + (1-p)^{2s_u - 2s_1 + 2}}{p (2-p)}
\end{equation}
By simple computations, the maximum $p \in [0,1]$ such that the inequality
$E [ \pi_k] - E [ \xi_k] \geq 0$ is satisfied (condition \ref{eq:lemcond3} of Lemma \ref{lem:randomwalk}), corresponds to the value $p_2$ defined in the statement of Theorem \ref{theo:maintheorem2}. As the function $E [\pi^t_k] - E [ \xi^t_k]$
intersects the line $y=0$ 
only in one point of the interval $[0,1]$, $p_2$ is the unique solution of $E [\pi^t_k] - E [ \xi^t_k]=0$ that falls in this interval.
\end{proof}

\begin{proof}[\textbf{Proof of Proposition \ref{prop:estimations}}]
In the proof we present the estimation of the probability $\delta_{\rho(x,y)} \mathcal{P}^2$ of the events $\{R_1^2 \geq j +  R_1^0 - 2s_u\}_{j \in \mathbb{N}_0}$.
Using the same argument one can estimate the probability of the events
$\{L_1^2 \leq  -j + L_1^0 - 2s_1\}_{j \in \mathbb{N}_0}$.
By definition of $R_1^2$,
$$
\{ R_1^2 \geq R_1^0- 2 s_u + j \}
\iff \forall z \in [[L_1^0 -  2 s_1, R_1^0 - 2 s_u + j]], \, \, \eta^2_z=0.
$$
As observed previously, the state of the sites in $[[R_1^0 - 2 s_u, R_1^0- 2 s_u + j]]$ for $\eta^2$ depends only on the state of the sites in 
$[[  R_1^0 -  2 s_u + s_1, R_1^0 - s_u + j  ]]$ for $\eta^1$.
Furthermore we observed that the state of the sites in $[[L_1^0 -  2 s_1, R_1^0 - 2 s_u ]]$ is zero almost surely for $\eta^2$.
Hence, from equation (\ref{eq:tau2-1}), 
we obtain the following estimation
 (see also Figure \ref{Fig:est12}),
\begin{multline}
\label{eq:j=1}
\begin{split}
\delta_{\rho(x,y)}\mathcal{P}^2( R_1^2  \geq  1 +  R_1^0  - 2 s_u  ) & =
\delta_{\rho(x,y)}\mathcal{P}^2( \forall z \in [[L_1^0 - 2 s_1,   R_1^0 - 2 s_u + 1 ]],
\, \, \eta_x = 0)\\ 
& =   \delta_{\rho(x,y)}\mathcal{P}^2(\eta_{R_1^0 - 2 s_1+1}=0) \\ &
= \delta_{\rho(x,y)} \mathcal{P} (   \eta_{R_1^0 - 2s_u + 1} = 0    )  +  
 \delta_{\rho(x,y)} \mathcal{P}( \eta_{R_1^0 - 2s_u + 1} = 1  )  \, (1-p)  
\\ & = (1-p) + p (1-p) = 1 - p^2,
\end{split}
\end{multline}
which corresponds to the estimation (\ref{eq:R21}).
Similarly we obtain the estimation (\ref{eq:R22}),
\begin{multline}
\label{eq:j=2}
\begin{split}
\delta_{\rho(x,y)}\mathcal{P}^2 (   R_1^2  \geq  2 +  R_1^0  - 2 s_u  ) & =
\delta_{\rho(x,y)}\mathcal{P}^2  ( \forall z \in [[ L_1^0 - 2 s_1, R_1^0 - 2s_u + 2]])
\\ & =  \delta_{\rho(x,y)}\mathcal{P}^2  (\eta_{R_1^0 - 2s_u + 1} = 
\eta_{R_1^0 - 2s_u + 2} =0  ) \\ & = \delta_{\rho(x,y)}\mathcal{P}^1
(   \eta_{R_1^0 - s_u + 1} =0, \,  \eta_{R_1^0 - s_u + 2} =0   ) \, \, 1  
\\  & +    \delta_{\rho(x,y)}\mathcal{P}^1 
(  \eta_{R_1^0 - s_u + 1} =0, \,  \eta_{R_1^0 - s_u + 2} =1   )  \, \, {(1-p)} 
\\ & + \delta_{\rho(x,y)}\mathcal{P}^1 (\eta_{R_1^0 - s_u + 1} =1, \,  \eta_{R_1^0 - s_u + 2} =0   ) \, \, {(1-p)}^2 \\ & + 
    \delta_{\rho(x,y)}\mathcal{P}^1(\eta_{R_1^0 - s_u + 1} =1, \,  \eta_{R_1^0 - s_u + 2} =1   )  \,\,  {(1-p)}^2   \\ &  \geq  (1-p)^2 \, + \,  p (1-p)^2 \,  + \,  p (1-p)^3 \, +  \, p^2(1-p)^2.
\end{split}
\end{multline}
\begin{figure}
\centering
\includegraphics[width=0.8\textwidth]{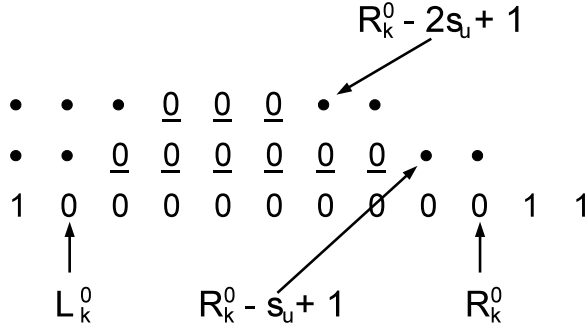}
\caption{In the figure we consider $\mathcal{U}=\{-1, 0, 1, 2 \}$. If the initial realisation of the Percolation PCA  is the one represented in the figure (lowest row), then almost surely the state of the sites above the short horizontal line is zero. }
\label{Fig:est12}
\end{figure}
We provide now the estimation (\ref{eq:tau2-estj})
considering all $j \geq 3$.
We introduce an index $m \in [[0, j-1]]$,
and we define the mutually disjoint cylinder sets (they will be defined later),
$$\{ C^{a,m} \}_{m \in [[0,j-1]]}\,  , \, \, \, 
\{ C^{b,m} \}_{m \in [[0,j-2]]}\,  , \, \, \,
C^{c}.$$ 
We denote by $C^d$ the set of realisations that are not in sets just defined, namely 
\begin{equation}
\label{eq:C^d}
C^d := \Sigma \, \, \setminus \, \,   \bigcup\limits_{m \in [[0, j-2]]}  C^{a,m} \cup  C^{b,m}  \cup C^c \cup C^{a,j-1}.
\end{equation}
For every $m  \in  [[0, j-2]]$ we estimate
$\delta_{\rho(x,y)} \mathcal{P}( C^{a, m})$ and
$\delta_{\rho(x,y)} \mathcal{P}( C^{b, m})$,
and we also estimate $\delta_{\rho(x,y)} \mathcal{P}( C^{a, j-1})$
and $\delta_{\rho(x,y)} \mathcal{P}( C^{c})$.
Furthermore for each of these sets we provide some bounds $B^{a,m}$,$B^{b,m}$,$B^{c}$. Namely for every $\eta^1 \in C^w$, where $w$ denotes generically $(a,m)$, $(b,m)$ or $c$, the following inequality holds,
\begin{equation}
\label{eq:tau2-2}
B^{w} \leq  \prod_{z \in  [[ R_1^0 -2 s_u,  \ldots, R_1^0- 2 s_u + j ]]}
T (    \eta^{2}_z=0  | \eta^1_{\mathcal{U}(z)} ),
\end{equation}
We use such estimations to provide a bound for (\ref{eq:tau2-1}), as shown in the following expression.
\begin{multline}
\begin{split}
\label{eq:tau2-5}
 \delta_{\rho(x,y)} \mathcal{P}^2(  R_1^2  \geq  j  +  R_1^0  - 2 s_u  )   \geq     \sum\limits_{m \in [[0, j-1]]} 
[ & \delta_{\rho(x,y)} \mathcal{P}  ( C^{a,m}  ) \, B^{a, m} +   \\ &
\delta_{\rho(x,y)} \mathcal{P}  ( C^{b,m}  ) \, B^{b,m} ] \,  +
\delta_{\rho(x,y)} \mathcal{P}  ( C^{c}  ) \, B^{c}  + \\&
\delta_{\rho(x,y)} \mathcal{P}  ( C^{d}  ) 
\end{split}
\end{multline}

We start with the introduction of the cylinder set $C^{a,m} \subset \Sigma$,
\begin{multline}
\begin{split}
C^{a,m}   := \{ \eta \in \Sigma \, \,   \mbox { s.t. } & \forall z \in  [[R^0_k - s_u + 1, R_1^0 - s_u + j]] \setminus \{  R^0_k-s_u+m+1 \},  
\\ & \eta_z=0  \mbox { and }  \eta_{R^0_k-s_u+m+1}=1 \},
\end{split}
\end{multline}
(see also Figure \ref{Subfig:c}).
By a simple computation,
\begin{equation}
\label{eq:PCam}
\delta_{\rho(x,y)} \mathcal{P} ( C^{a,m} ) = p \, ( 1- p )^{j-1}.
\end{equation}
Furthermore we observe that $\forall \eta \in C^{a,m}$,
the product over the transition probabilities 
of equation  (\ref{eq:tau2-1}) satisfies the following bound,
\begin{equation}
 \prod_{z \in [[  R_1^0 - 2 s_u , R_1^0- 2 s_u + j ]] } 
 T (    \eta^{2}_z=0  | \eta^1_{\mathcal{U}(z)} )   \geq B^{a,m},
\end{equation}
where
\begin{equation}
\label{eq:BCam}
B^{a,m} : =
\begin{cases}
(1-p)^{s_u-s_1+1}  & \mbox{ if }  0 \leq m \leq j - (s_u - s_1) - 1\\
(1-p)^{j-m}   	  & \mbox{ if }  j -  (s_u - s_1) \leq m \leq j -1\\
\end{cases}.
\end{equation}

\begin{figure}
\begin{center}
\begin{subfigure}[t!]{\linewidth}
\includegraphics[width=0.9\textwidth]{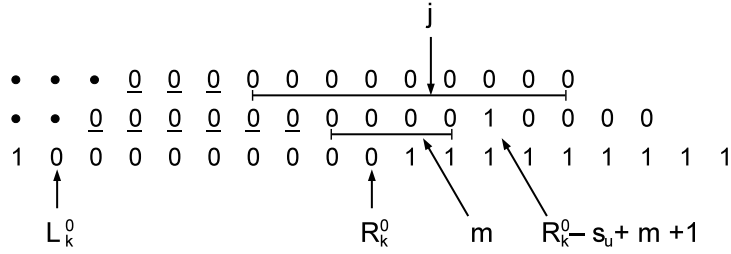}
\caption{}
\label{Subfig:c}
\end{subfigure}%
          
\begin{subfigure}[b]{\linewidth}
\includegraphics[width=0.9\textwidth]{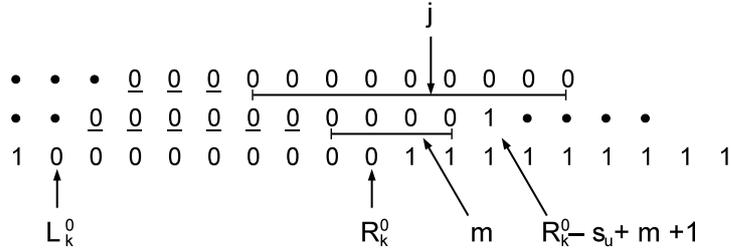}
\caption{}
\label{Subfig:d}
\end{subfigure}
          
\begin{subfigure}[b]{\linewidth}
\includegraphics[width=0.9\textwidth]{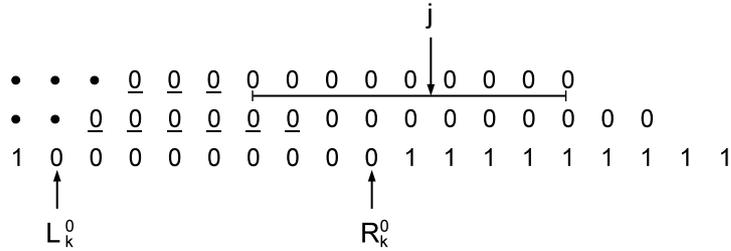}
\caption{}
\label{Subfig:e}
\end{subfigure}
\caption{In the figures above the neighbourhood is assumed to be $\mathcal{U}= \{  -1, 0, 1, 2  \}$. If the initial realisation of the Percolation PCA (first row) is the one represented in the figure, then the sites underlined by a short line on the second and third row have state zero almost surely. Figure (a): the second row represents a realisation $\eta^2$ belonging to the set $C^{a,m}$, $m=4$, $j=9$.
Figure (b): the second row from below represents a realisation $\eta^2$  belonging to the set $C^{b,m}$, $m=4$, $j=9$.
Figure (c): the second row from below represents a realisation $\eta^2$ belonging to the set $C^{c}$, $j=9$.}
\label{Fig3:computationbis}
\end{center}
\end{figure}
Then we introduce the cylinder sets $C^{b,m} \subset \Sigma$,
where $ 0 \leq m \leq j-2$.
\begin{multline}
\label{eq:PCbm}
\begin{split}
C^{b,m} := \{  \eta \in \Sigma \, \, \,  \mbox{ s.t.}  \, \, \, 
 & \forall z \in [[R_1^0 - s_u + 1, R_1^0 - s_u + m]],  \, \,  \eta_z=0, \\
& \eta_{R_1^0 - s_u + m + 1}=1, \\ &
\exists y \in [[R_1^0 - s_u + m +2, R_1^0 - s_u + j]]
 \mbox{ s.t. }   \eta_y = 1 \} 
 \end{split}
\end{multline}
(see also Figure \ref{Subfig:d}).
By using the definition of transition probability for the Percolation PCA we estimate the probability measure of this cylinder set
\begin{equation}
\label{eq:PCbm}
\delta_{\rho(x,y)} \mathcal{P} ( C^{b,m} ) = (1-p)^m \, p \,  [ 1- (1-p)^{j-m-1}] ,
\end{equation}
and we observe that $\forall \eta^1 \in C^{b,m}$
the following bound holds
\begin{equation}
\label{eq:TCbm}
\prod_{z \in [[  R_1^0 - 2 s_u , R_1^0- 2 s_u + j ]] }
T (    \eta^{2}_z=0  | \eta^1_{\mathcal{U}(z)} ) \geq (1-p)^{j-m} 
\end{equation}
Thus we define
\begin{equation}
\label{eq:BCbm}
B^{b,m} := (1-p)^{j-m}.
\end{equation}
The bound (\ref{eq:TCbm}) is obtained considering that
$T (\eta_z^2 = 0 \, | \, \eta_{\mathcal{U}(z)}) = 1 $
for all $z \in [[R_1^0 - 2 s_u + 1, R_1^0 - 2 s_u + m]]$
and $T (\eta_z^2 = 0 \, | \, \eta_{\mathcal{U}(z)}) \geq (1-p) $
for all $z \in [[R_1^0 - 2 s_u + 1 + m, R_1^0 - 2 s_u + j]]$.

Third, we define the cylinder set $C^c \subset \Sigma$,
\begin{equation}
C^c := \{ \eta^1 \in \Sigma \mbox{ s.t. }  \eta^1_z = 0  \, \, \, \, \forall
 z \in  [[R_1^0 -  s_u + 1, R_1^0 -  s_u + j]]\}.
\end{equation}
For this set, 
\begin{equation}
\label{eq:PCcm}
\delta_{\rho(x,y)} \mathcal{P} (C^c) = (1-p)^j,
\end{equation}
and $\forall \eta^1 \in C^c$,
\begin{equation}
\label{eq:TCcm}
\prod_{z \in [[  R_1^0 - 2 s_u , R_1^0- 2 s_u + j ]] }
T (    \eta^{2}_z=0  | \eta^1_{\mathcal{U}(z)} ) = 1 
\end{equation}
(see also Figure \ref{Subfig:e}).
Thus we define 
\begin{equation}
\label{eq:BCcm}
B^c := 1.
\end{equation}

Finally we recall the definition of $C^d$ provided in equation (\ref{eq:C^d}).
We observe that 
\begin{equation}
\label{eq:PCdm}
\delta_{\rho(x,y)} \mathcal{P} (C^d) = 1 - \sum\limits_{m=0}^{j-1}  \delta_{\rho(x,y)} \mathcal{P} (C^{a,m})  -  \sum\limits_{m=0}^{j-2} \delta_{\rho(x,y)} \mathcal{P} (C^{b,m} )  - 
\delta_{\rho(x,y)} \mathcal{P} (C^c)
\end{equation}
and that $\forall \eta^1 \in \Sigma$,
\begin{equation}
\label{eq:TCdm}
\prod_{x \in [[  R_1^0 - 2 s_u , R_1^0- 2 s_u + j ]] }
T (    \eta^{2}_x=0  | \eta^1_{\mathcal{U}(x)} )
\geq  (1-p)^j.
\end{equation}  
The inequality is obtained considering that from the definition (\ref{eq:transitionprob}) it follows that  $\forall z \in \mathbb{Z}$,
$T (    \eta^{2}_x=0  | \eta^1_{\mathcal{U}(x)} ) \geq (1-p)$.
Thus we define
\begin{equation}
\label{eq:BCdm}
B^d :=  (1-p)^j.
\end{equation}

We finally use the estimations
(\ref{eq:PCam}), (\ref{eq:BCam}), (\ref{eq:PCbm}), (\ref{eq:BCbm}),
(\ref{eq:PCcm}), (\ref{eq:BCcm}), (\ref{eq:BCdm}), (\ref{eq:BCdm}),
in (\ref{eq:tau2-5}) and we derive the lower bounds
(\ref{eq:tau2-estj}) and (\ref{eq:tau2-estj2}).
\end{proof}

\section{Convergence time of the finite process}
\label{sect:timeofconvergence}
In this section we prove Theorem \ref{theo:maintheorem2}.
In Section \ref{sect:ergodicity}
we describe the connection between
Percolation PCA and 
oriented percolation.
We mainly follow \cite{ToomDiscr, Uniform},
although propositions and statement
have been reformulated emphasising 
the differences between Percolation PCA
on a finite and infinite space.
In Section \ref{sect:percolationestimates}
we list some percolation estimates.
Some of these percolation estimates have been proved in 
\cite{DurretOr, Durret, Griffeath}
in the case of oriented bond percolation with
symmetric neighbourhood.
In this article we consider a similar model,
namely oriented site percolation with arbitrary neighbourhood.
The proofs of these estimates in our case are substantially the same of those
provided in \cite{Durret, DurretOr}. 
We sketch them illustrating the small differences.
In Section \ref{sect:proofTheo2} we
prove the theorem.
The proof of the right inequality of statement (a) 
of the theorem is an
application of the estimates presented in Section \ref{sect:percolationestimates}.
The proof of the left inequality can be found in
\cite{Stavskaja3}.
The proof of the right inequality of 
statement (b) is trivial.
The proof of the left inequality
uses some of the percolation estimates
and the estimation provided Proposition \ref{prop:inequality2}, which is stated in the same section.
The original contribution of the author
consists in the proof and the application of Proposition \ref{prop:inequality2}
to the proof of the statement (b), in the estimations based on path constructions used in the proof of statement (b) and
in the generalization of the percolation estimates to the proof of the statement (a).

\subsection{Relations with Oriented Percolation}
\label{sect:ergodicity}
In this section we describe a 
connection between the Percolation PCA and 
a certain percolation model.
This connection has been pointed out 
for the first time in \cite{Uniform}, as far as we know.
We consider a Percolation PCA with space
$S = \mathbb{S}_n$ or $S=\mathbb{Z}$,
as defined in Section \ref{sect:themodel}.
We define an auxiliary space $\Omega= {\{  0,1  \}}^V$,
we denote by $\omega \in \Omega$ its elements
and we introduce in this space the Bernoulli product measure $\mathbb{P}_p$.
Namely, the state of every component is $1$ with probability $p$ and 
$0$ with obability $1-p$ independently.
We declare a vertex $(x,y) \in V$ ``open'' if $w_{x,y} =1$ and ``closed'' otherwise.
The correspondence between the PCA and percolation consists in the fact that the probability that the state of the
site $x \in S$ is $1$ at time $t \in \mathbb{N}_0$ for
the probabilistic cellular automaton
equals the probability that the site
$(x,t) \in V$ is connected by a path of open vertices in
$\mathcal{G}_{\mathcal{U}}$ to the line $y=0$.
This is precisely the meaning of the statement of Proposition
\ref{prop:percolation}, which is stated below. 

In order to describe this connection rigorously,  we represent the Percolation 
PCA starting from an initial realisation $\eta^i \in \Sigma$ by introducing a deterministic mapping
$$\eta \, : \, \Omega \times \Sigma \longrightarrow \tilde{\Sigma}.$$
For every $(x,t) \in V$,  the component $\eta_x^t \, : \, \Omega \times \Sigma \rightarrow \{0,1\}$ of $\eta$ is defined as
\begin{multline}
\begin{split}
\label{eq:mapping}
&\eta^{t}_x : = 
\begin{cases}
\min \{   \omega_{x,t-1 }, 
\max_{k \in \mathcal{U}(x)} \{  \eta_k^{t-1}\,  \}	
	\}, &  \mbox{ if }  t \in \mathbb{Z}_+ \\
	\eta^0_x = \eta^i_x, & \mbox{ if }  t =0,
\end{cases}	
\end{split}
\end{multline}
where ${(\omega_{x,t})}_{x \in S, y \in \mathbb{N}}$ are elements of $\Omega$.
This mapping defines any $\eta_z^T$, $z \in  V$, $T \in \mathbb{Z}_+$ as a function of the variables $\omega_{x,y}$ associated to vertices belonging to the evolution cone of  $(z,T) \in V$, and of initial realisation $\eta_x^i$.
One should observe that, recalling (\ref{eq:transitionprob}) and using independence, for any $x \in S$, $t \in \mathbb{Z}_+$, $a \in \{0,1\}$, $\eta^{t-1}_{\mathcal{U}(x)}  \in \{0,1\}^{\mathcal{U}(x)}$,  $\eta^i \in \Sigma$,
\begin{equation}
\label{eq:correspondence}
\begin{split}
T_{x}( \, \eta^t_x  = a \, | \, \eta^{t-1}_{\mathcal{U}(x)}  \, ) = &  \mathbb{P}_p (\, \omega \in \Omega \mbox{ s.t. } \eta_x^t (\omega, \eta^i) = a \, | \, 
\omega \in \Omega \mbox{ s.t. } \eta^{t-1}_{\mathcal{U}(x)}(\omega, \eta^i) \, )
\\  : =  & \mathbb{P}^{\eta^i}_p ( \eta^t_x =a | \eta^{t-1}_{\mathcal{U}(x)}  ),
 \end{split}
\end{equation}
where in the last expression we rewrote the second quantity in a more compact form. This notation will be used also in the proof of the next proposition.
\begin{prop}
\label{prop:correspondence}
Consider the Percolation PCA with space $S \in \{ \mathbb{S}_n, \mathbb{Z}\}$,
represented by the operator $\mathcal{P} : \mathcal{M}(\Sigma) \rightarrow \mathcal{M}(\Sigma)$. Then, for any $\eta^i \in \Sigma$, $a \in \{0,1\}$,
\begin{equation}
\label{eq:correspondenceprop}
\delta_{\delta_{\eta^i}} \mathcal{P}^t (\eta_x=a)  = \mathbb{P}_p ( w \in \Omega \, \, \mbox{s.t.} \, \, \eta^t_x(\omega, \eta^i) =a).
\end{equation}
\end{prop}
\begin{proof}
For any $x \in S$, $t \in \mathbb{Z}^+$, we define
$$
\mathcal{U}^{t}(x) = \overset{t}{\overbrace{ \mathcal{U}\,  \circ \,  \mathcal{U} \, \circ \, \ldots \, \circ \, \mathcal{U}}} \,   (x),
$$
By using equation (\ref{eq:correspondence}), we observe that the following equalities hold.
\begin{align}
\begin{split} 
\mathbb{P}_p ( \eta^t_x =a)  & = 
\sum_{\eta^{t-1}_{\mathcal{U}(x)} \in \{0,1\}^{\mathcal{U}(x)}}
\mathbb{P}_p ( \eta^t_x   =a | \eta^{t-1}_{\mathcal{U}(x)}  )   \mathbb{P}_p ( \eta^{t-1}_{\mathcal{U}(x)}) \\
& = 
\sum_{\eta^{t-1}_{\mathcal{U}(x)} \in \{0,1\}^{\mathcal{U}(x)}} \, 
T_{x}( \, \eta^t_x  = a \, | \, \eta^{t-1}_{\mathcal{U}(x)}  \, ) 
\mathbb{P}_p ( \eta^{t-1}_{\mathcal{U}(x)} ) \\
& = \sum_{\eta^{t-1}_{\mathcal{U}(x)} \in \{0,1\}^{\mathcal{U}(x)}}
\sum_{\eta^{t-2}_{\mathcal{U}^2(x)} \in \{0,1\}^{\mathcal{U}^2(x)}} 
T_{x}( \, \eta^t_x  = a \, | \, \eta^{t-1}_{\mathcal{U}(x)}  \, )  \\ 
& \times \mathbb{P}_p ( \eta^{t-1}_{\mathcal{U}(x)} \,  | \,    \eta^{t-2}_{\mathcal{U}^2(x)}  ) \, \mathbb{P}_p (\eta^{t-2}_{\mathcal{U}^2(x)}  ) \\
& = \sum_{\eta^{t-1}_{\mathcal{U}(x)} \in \{0,1\}^{\mathcal{U}(x)}}
\sum_{\eta^{t-2}_{\mathcal{U}^2(x)} \in \{0,1\}^{\mathcal{U}^2(x)}} 
T_{x}( \, \eta^t_x  = a \, | \, \eta^{t-1}_{\mathcal{U}(x)}  \, )  \\ 
& \times T ( \eta^{t-1}_{\mathcal{U}(x)} \,  | \,    \eta^{t-2}_{\mathcal{U}^2(x)}) \, \mathbb{P}_p (\eta^{t-2}_{\mathcal{U}^2(x)}   ) \\
 & =  \ldots 
\end{split}
\end{align}
By proceeding with the expansion, by writing the previous term as a sum over all states of the cylinder set of $(x,t)$ of transition probabilities (\ref{eq:transitionprob})
and by recalling definition (\ref{eq:operator}), one recognizes that this
term is equal to the left side of (\ref{eq:correspondenceprop}).
\end{proof}
The next proposition has been proved in \cite{Uniform} and it is stated and used also in \cite{ToomDiscr}. 
\begin{prop}
\label{prop:percolation}
The function  $ \eta^{t}_x  : \Omega \times \Sigma \mapsto \{0, 1 \}$ is such that $\eta^{t}_x = 1$ \textit{iff} there exists a sequence 
$\{ x_0, x_{1}, x_{2}, \ldots x_{t} \} \subset \mathbb{Z}$  satisfying the three following properties,
\begin{enumerate}
\item $x_t = x$ and $x_{i-1} \in \mathcal{U}(x_{i})$ for any $i \in \{ 1, 2,  \ldots t \}$,
\item $\omega_{i-1, x_i} = 1$ for any $i \in \{1, 2, \ldots t \}$,
\item $\eta^i_{x_0} = 1$.
\end{enumerate}
\end{prop}
\proof We sketch the proof of the proposition.
Assume $\eta_x^t =1$ and assume that
properties 1, 2, 3 hold for a sequence of sites 
$x_{t-k}$, $x_{t-k+1}$, $\ldots$ $x_{t}$.
From (\ref{eq:mapping}) it follows that
$\eta^{t-k}_{x_{t-k}} = 1 \Leftrightarrow \omega_{x_{t-k-1}, t-k}=1$
and $\exists\,  x_{t-k-1} \in \mathcal{U}(x_{t-k})$
s.t. $ \eta^{t-k-1}_{x_{t-k-1}} = 1$.
This implies that there exists an element $x_{t-k-1} \in S$
such that properties 1, 2, 3 hold
for the sequence $x_{t-k-1}$, $x_{t-k}$, $\ldots$ $x_{t}$.
The proof of the proposition follows by induction.
\endproof
If we consider the case of infinite space, from the previous proposition it follows that ergodicity for the probabilistic cellular automaton is associated with the existence of an infinite path of open vertices in the auxiliary space.
Indeed, recall Definitions \ref{def:critprob} and \ref{def:ergodic} and observe that,
\begin{align}
\label{eq:ergodicpercolation}
p > p_c \implies \lim_{t \rightarrow \infty } \delta_{\mathbf{1}} \mathcal{P}^t (\eta_x=1) > 0  \\
p<p_c 	\implies  \lim_{t \rightarrow \infty } \delta_{\mathbf{1}} \mathcal{P}^t (\eta_x=1) =0,
\end{align}
Thus the probabilistic cellular automaton is non-ergodic if and only if the 
limit $t \rightarrow \infty$ of the probability that a vertex $(0,t)$ is connected to the line $y=0$ by an open path is positive.
 
If we consider the case of finite space with periodic boundaries, the previous proposition shows that there is a connection between the absorption time of the probabilistic cellular automaton and the existence of an open path in the auxiliary space. This connection is clarified in the next proposition.  Before its statement we introduce some more definitions.

From now on we use $\mathbb{P}_p^n( \cdot )$ to denote the Bernoulli product measure in the finite  space and  $\mathbb{P}_p( \cdot )$ to denote
the Bernoulli product measure in the infinite space.
\begin{mydef}
\label{def:opensitesandpaths}
Consider $S \in \{  \mathbb{S}_n, \mathbb{Z}\}$ and consider the event,
\begin{align}
\begin{split}
 \{ \omega \in \Omega \mbox{ s.t.} & \mbox{ there exists a path of open vertices in } \mathcal{G}_{\mathcal{U}}  \mbox { that connects } \\ 
& (x,t) \mbox{ to one of the vertices belonging to the line } y=0 \}.
\end{split}
\end{align}
If $S = \mathbb{Z}$ we denote this event by 
$ \{(s, t) \overset{\mathcal{G}_{\mathcal{U}}}{ \longrightarrow } S^0 \}$
and if $S = \mathbb{S}_n$ we denote this event by
$ \{(s, t) \overset{ \mathcal{G}_{\mathcal{U}}  (n) }{\longrightarrow} S^0 \}.$
\end{mydef}
Recall the definition of evolution measure (Definition \ref{def:evolutionmeasure}) and of 
absorption time (equation \ref{eq:absorbingtime}). 
Recall that $\tau_k$ can be considered as a function 
$\tau_k :  \Omega \times \Sigma \rightarrow \mathbb{N}$, 
as, from (\ref{eq:mapping}), $( \eta_x^t )_{x \in S, t \in \mathbb{N}}$ 
it is a mapping from $\Omega \times \Sigma$ to $\tilde{\Sigma}$.
\begin{prop}
\label{prop:taunconnection}
Consider the Percolation PCA on a finite space with periodic boundaries. For every $t \in \mathbb{N}_0$,
\begin{equation}
\mathcal{E}^n_{\delta_{\mathbf{1}}}
( \tau_n > t ) =
\mathbb{P}^n_p (\, \, \exists x \in [[-n, n-1]] \, \, \mbox{s.t.} \, \, 
(x,t) \overset{ \mathcal{G}_{\mathcal{U}}  (n) }
{\longrightarrow}  
S^0 \ ),
\end{equation}
where $S^0$ denotes the set of vertices of $V$ belonging to the line $y=0$.
\end{prop}
\begin{proof}
By the definition of $\tau_n$ (see Definition \ref{def:absorbtime}), $\tau_n > t$ if and only if $\exists x \in [[-n, n-1]]$
such that $\eta_x^t =1$.
From Proposition \ref{prop:percolation}, it follows that $\eta_x^t = 1$ if and only if $ (x,t) \overset{ \mathcal{G}_{\mathcal{U}}  (n) } {\longrightarrow}  S^0$.
\end{proof}
\begin{figure}
\begin{center}
\includegraphics[scale=0.3]{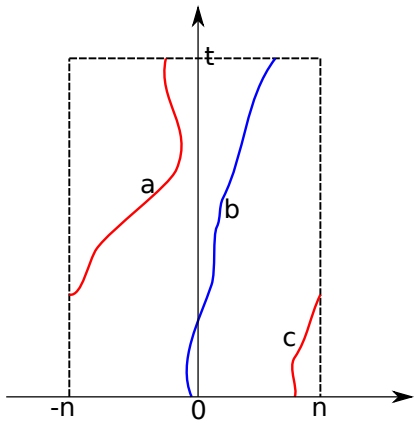}
\caption{The event $\{ \tau_n > t \}$  (recall Definition \ref{def:absorbtime})   
occurs if at least one open path joins one of the sites $(x,t)$ such that $x \in [[-n, n-1]]$ to one of the sites $(y,0)$, $y \in S$.
If the neighbourhood is periodic, then the path can leave from one the two vertical lines $x=-n$ or $x=n-1$ and re-appear at the same
high on the other line (e.g. see the path $a \circ c$). }
\label{Fig:events1}
\end{center}
\end{figure}
\noindent \textbf{Remark.} Recall the definition of the neighbourhood in the case of finite space with periodic boundaries, provided in equation (\ref{eq:neighborperiodic}). 
As boundaries are periodic, the site $(x,t)$ is connected to the line $y=0$ also if the 
path of open vertices leaves one of the vertical boundaries ($x = -n$ or $x= n-1$) from one side and it re-appears at the same high on the other side (see for example the path $a \circ c$ in Figure \ref{Fig:events1}).

\subsection{Percolation estimates}
\label{sect:percolationestimates}
In this section we list some properties
involving the cluster of vertices belonging to an open path
in $\mathcal{G}_{\mathcal{U}}$ starting from $(0,t)$. 
These properties have been proved in 
\cite{Durret, DurretOr, Griffeath} in case of a bond percolation model 
with symmetric neighbourhood of two elements.
In this article we consider a slightly different percolation model,
as sites instead of bonds can be open or closed and the neighbourhood
is an arbitrary (translation invariant) finite set.
The proofs of these propositions in the case considered in this
article are similar to those provided in \cite{Durret, DurretOr, Griffeath}.
We sketch their proof describing the small differences.

We start with some definitions. From now on we will consider $S = \mathbb{Z}$.
For every $t, m \in \mathbb{N}$ we define the sets,
\begin{equation}
\begin{split}
\xi^t_{m} &= \{   x \in \mathbb{Z} \, : \,  (0,t)   \overset{\mathcal{G}_{\mathcal{U}}}{\longrightarrow} \,  (x,t-m)  \}, \\
{\overline{ \xi }}^t_{m} &=  \{    x \in\mathbb{Z} \, : \,   \exists \, z \, \leq 0
\, \, \mbox{s.t.}  \,  (z,t)   \overset{\mathcal{G}_{\mathcal{U}}}{\longrightarrow} (x,t-m)  \}, \\
\overline{ \chi }^t_{m} &=  \{   x  \in \mathbb{Z} \, : \, \exists \, z \, \geq 0
\, \, \mbox{s.t.}  \, \,  (z,t)   \overset{\mathcal{G}_{\mathcal{U}}}{\longrightarrow} (x,t-m)  \}, \\ 
\end{split}
\end{equation}
Note that $\xi^t_{m} \subset \{ s_1 \, m, \, \, s_1 \, m +1, \,  s_1 \, m +2, \ldots, \,  s_u \, m \}$.
We define then the variables,
\begin{equation}
\label{eq:definitionperc}
\begin{split}
r^t_{m} \, & = \, \sup  \{ \xi_m^t \}, \\
\ell^t_{m} \, & = \, \inf  \{ \xi_m^t \}, \\
\overline{r}^t_{m} \, & = \, \sup  \{ \overline{\xi}_m^t \}, \\
\overline{\ell}^{t}_m \, &  = \, \inf  \{ \overline{\chi}_m^t \}, \\
\end{split}
\end{equation}
and we set $r^t_{m} = - \infty$, $\ell^t_{m} = \infty$ if $\xi^t_{m} = \varnothing$.
As the distributions of $r^t_{m}$, $\ell^t_{m}$,
$\overline{r}^t_{m}$, $\overline{\ell}^t_{m}$,
 $\xi^t_{m}$,  $\overline{\xi}^t_{m}$ and $\overline{\chi}^t_{m}$ depend only on 
the difference $t-m$, from now on
we will omit the dependence on $t$, that will be some positive integer. 
Furthermore we consider the space 
$\mathcal{G}_{\mathcal{U}}$
as before, but
with vertices $\mathbb{Z} \times \mathbb{Z}$
instead of $\mathbb{Z} \times \mathbb{N}$.
In the former case, if we consider only paths starting from $(0,t)$, we allow $(0,t)$ to belong to an infinite open path.
Thus we recover the notation
of \cite{DurretOr} ($r_m$, $\ell_m$, $\overline{r}_m$, $\xi_m$), with the difference
that in this article paths are oriented from up to down.

We observe that for every $t$, $m$, the probability that $\overline{\xi}^{m} = \varnothing$
is zero, as every vertex in $\{ (x,y) :  \mbox{s.t. }  y=t, \, x \geq 0\}$ has a non-zero
probability of being connected to  $S^0$ by an open path in $\mathcal{G}_{\mathcal{U}}$.
The same holds for the event $\overline{\chi}^{m,t} = \varnothing$.
By definition,
\begin{equation}
\begin{split}
\label{eq:rnrbarn}
r_m \leq  \overline{r}_m, \\
{\ell}_m \geq \overline{\ell}_m.  
\end{split}
\end{equation}
The following relations hold,
\begin{align}
\label{eq:proprl1}
\xi_m & = \overline{\xi}_m \cap [\ell_m, + \infty) = \overline{\chi}_m^t \cap (-\infty, r_m], \\ 
\label{eq:proprl2}
~~~~~~~~~~~~ \mbox{ on } \{\xi_m \neq \varnothing \}, ~~~~~~~ r_m & = \overline{r}_m, \\
\label{eq:proprl3}
~~~~~~~~~~~~ \mbox{ on } \{\xi_m \neq \varnothing\}, ~~~~~~~ \ell_m & = \overline{\ell}_m.
\end{align}
\begin{proof}
Equation (\ref{eq:proprl1}) is a corollary of equations 
(\ref{eq:proprl2}) and (\ref{eq:proprl3}). 
We sketch an argument for (\ref{eq:proprl3}),
that can be also found in \cite[Section 3]{DurretOr}. 
By reflection the same argument holds also for (\ref{eq:proprl2}).
It is trivial from the definition that $\xi_m \subset \overline{\xi}_m$
and that $\xi_m \subset (-\infty, r_m]$. We have to show that 
$\overline{\xi}_m \cap (-\infty, r_m] \subset \xi_m$.
In this case it is clear from Figure \ref{Fig:rnrml} that if there is a path from
some site $(y,t)$, $y >0$ to $(x, t-m)$, $x\leq r_m$,
then there is also a path from $(0,t)$ to $(x,t-m)$.
Then $x \in \xi_m$.
\end{proof}
\begin{figure}
\begin{center}
\includegraphics[scale=0.35]{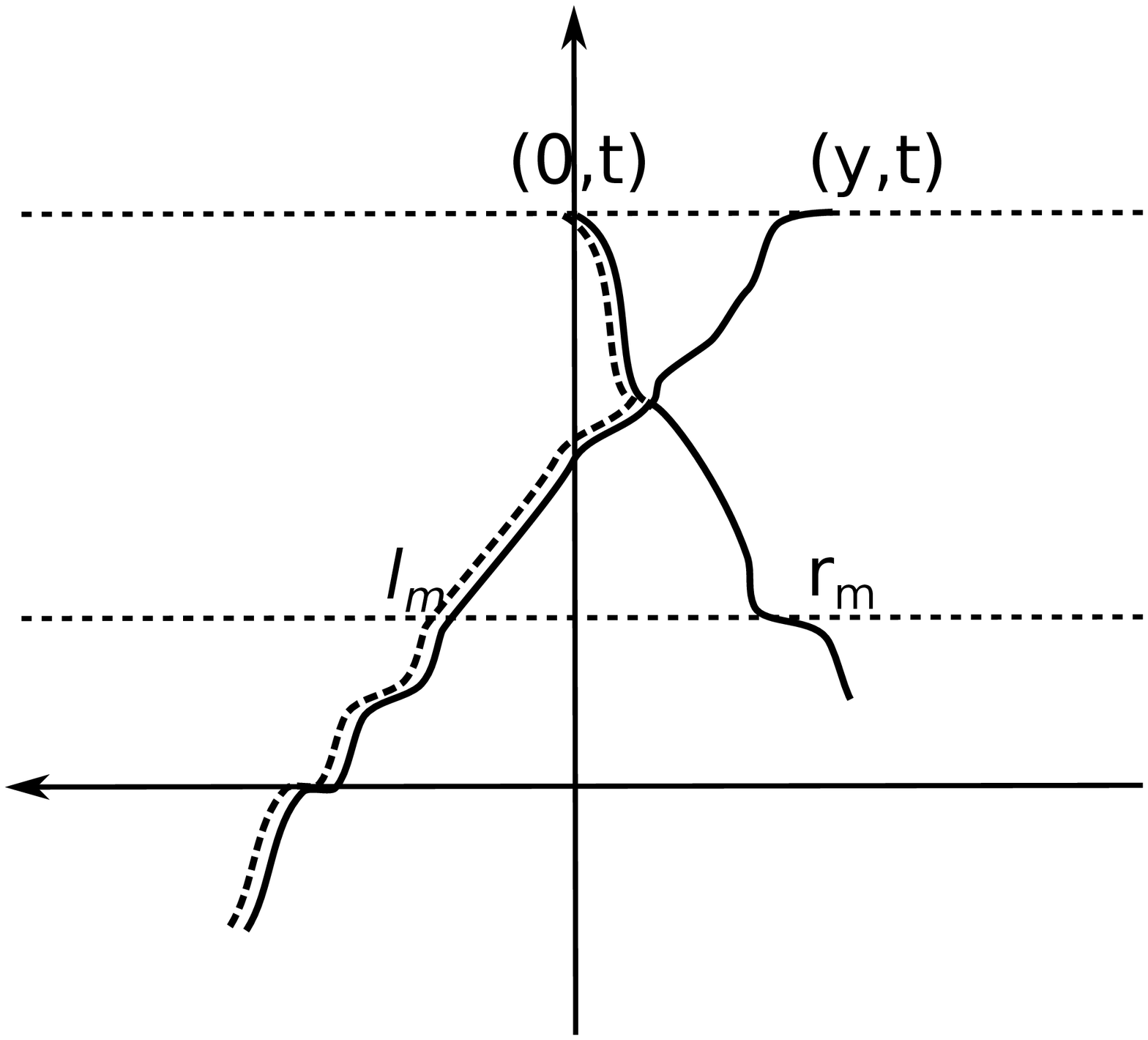}
\caption{  }
\label{Fig:rnrml}
\end{center}
\end{figure}
We introduce the following quantities, for all integers $n \geq m \geq 0$,
\begin{equation}
\begin{split}
\overline{r}_{m,n} & = \sup \{ x - \overline{r}_m \, : x \in \mathbb{Z} \, \mbox{ and }  \,
 \exists z \, \in \mathbb{Z}  \mbox{ s.t. }   \\ & z \leq \overline{r}_m \mbox{ and } (z, t-m) \overset{\mathcal{G}_{\mathcal{U}}}{\longrightarrow}  \, (x, t-n)    \}.
\end{split}
\end{equation}
\begin{equation}
\begin{split}
\overline{\ell}_{m,n} & = \inf \{ x - \overline{\ell}_m \, : x \in \mathbb{Z} \, \mbox{ and }  \,
 \exists z \, \in \mathbb{Z}  \mbox{ s.t. }   \\ & z \geq \overline{\ell}_m \mbox{ and } (z, t-m) \overset{\mathcal{G}_{\mathcal{U}}}{\longrightarrow}  \, (x, t-n)    \}.
\end{split}
\end{equation}
The following relations holds.
\begin{equation}
\label{eq:subadditive}
\overline{r}_{m} + \overline{r}_{m,n} \geq \overline{r}_{n}.
\end{equation}
\begin{equation}
\label{eq:subadditiveL}
\overline{\ell}_{m} + \overline{\ell}_{m,n} \leq \overline{\ell}_{n}.
\end{equation}
\begin{proof}
We prove (\ref{eq:subadditive}) and a similar argument holds for (\ref{eq:subadditiveL}).
One should observe that 
$\overline{r}_{m} +  \overline{r}_{m,n}$
is the rightmost point on the line $y = t- n$
which can be reached from any of the points $(x, t- m)$
with $x \leq \overline{r}_m$.
Instead $\overline{r}_{n}$ is the rightmost point
on the line $y = t- n$
which can be reached from any of the points $(x, t- m)$
with  $x \leq \overline{r}_m$ and with the additional restriction
that there exists an open path in $\mathcal{G}_{\mathcal{U}}$ from $(z, t)$
to $(x, t- m)$ for some $z \leq 0$. See also Figure \ref{Fig:rnrm}.
\end{proof}
\begin{figure}
\begin{center}
\includegraphics[scale=0.35]{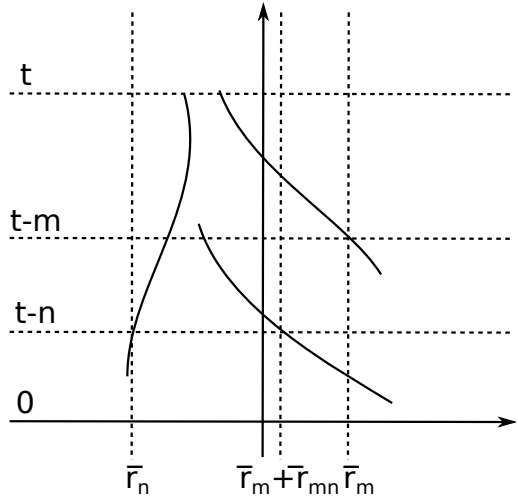}
\caption{Curves represent open paths.}
\label{Fig:rnrm}
\end{center}
\end{figure}
The next proposition involves the random variables defined above
and it corresponds to \cite{Durret}[Theorem 2.1]. It holds 
for a class of model called \textit{growth processes} that is more general than the class of models considered here. We refer to  \cite{Durret} for its proof,
which is based on the subadditivity property of (\ref{eq:subadditive}) and some arguments
similar to those used  in the proof of the Kingman's Subadditive Ergodic Theorem.
\begin{prop}
\label{prop:convergencesubadditive}
Let $\overline{r}_m$ and $\overline{\ell}_m$ be the quantities defined above. Then
there exist two constants $\alpha \in [-\infty, s_u]$ and   $\beta \in [s_1, +\infty]$ such that,
\begin{equation}
\label{eq:convergencesubadditive4}
\overline{r}_m / m \rightarrow \alpha \, \, \, \mbox {almost surely},
\end{equation}
\begin{equation}
\label{eq:convergencesubadditivel}
\overline{\ell}_m / m \rightarrow \beta \, \, \, \mbox {almost surely}.
\end{equation}
\end{prop}
Let  $E^t \subset \Omega$ be the following event,
$$
E^t := 
\{ \mbox{``there exists an infinite open path starting from (0,t)}'' \}.
$$
Then, if $p > p_c$, conditioning on $E^t$, from Proposition \ref{prop:convergencesubadditive}
and from equations (\ref{eq:proprl1}), (\ref{eq:proprl2}), (\ref{eq:proprl3})
the following properties hold,
\begin{equation}
\label{eq:r_n}
\lim_{m \rightarrow \infty} r_m / m = \alpha \, \, \, \mbox{almost surely},
\end{equation}
\begin{equation}
\label{eq:r_n}
\lim_{m \rightarrow \infty} \ell_m / m = \beta \, \, \, \mbox{almost surely},
\end{equation}
\begin{equation}
\label{eq:r_n}
\beta \leq \alpha.
\end{equation}
\begin{proof}
If $p>p_c$ then the event $E^t$ occurs with positive probability. Conditioning on $E^t$, for all $m \geq 0$ $r_m \geq \ell_m$. Furthermore,
from equations (\ref{eq:proprl2}) and (\ref{eq:proprl3}) it follows that $r_m = \overline{r}_m$ and $\ell_m = \overline{\ell}_m$.
\end{proof}
We define now the variable,
\begin{equation}
\label{eq:gamma}
\gamma := \alpha - \beta,
\end{equation}
which plays the role played by $\alpha$ in \cite{DurretOr}.
The proof of the next proposition can be found in \cite[Section 3]{DurretOr},
in case of bond percolation with symmetric neighbourhood.
As the statement is needed for the proof of Theorem \ref{theo:maintheorem2},
we sketch its proof, adapting it to the model considered in this article.
\begin{prop} 
\label{prop:pcgamma2}
Let $\gamma$ be the variable defined in equation (\ref{eq:gamma}). Then,
\begin{equation}
\label{eq:pcgamma0}
p_c = \inf \{p \, : \, \gamma(p) >0\}.
\end{equation}
\end{prop}
\begin{proof}
Observe that equation (\ref{eq:r_n}) implies that,
\begin{equation}
\label{eq:pcgamma1}
 \alpha < \beta \implies p \leq p_c.
\end{equation}
Then, to prove equation (\ref{eq:pcgamma0}), first it is necessary to show that,
\begin{equation}
\label{eq:pcgamma2}
\gamma > 0 \implies p >p_c.
\end{equation} 
Indeed, equations (\ref{eq:pcgamma1}) and (\ref{eq:pcgamma2}) imply 
that 
\begin{equation}
\label{eq:pcgamma3}
\sup  \{p \, : \, \gamma(p) < 0\} \leq p_c \leq \inf \{p \, : \, \gamma(p) >0\}.
\end{equation}
Hence, it remains to exclude the possibility that the interval $\{p \, : \gamma(p) = 0 \}$ has positive length. 
This fact is a consequence of the following property,
\begin{equation}
\label{eq:pcgamma4}
 p > p^{\prime} \mbox{ and } \alpha(p^{\prime}) > - \infty \implies 
 \alpha(p) - \alpha(p^{\prime})  \geq (p - p^{\prime}).
\end{equation}
and of the fact that $\beta(p)$ is non-decreasing with $p$. For the proof of 
(\ref{eq:pcgamma4}) we refer to \cite[Section 3]{DurretOr},
as the symmetry of the neighbourhood does not play any role in the proof.
The proof is based on the construction of two systems with parameter $p$ and 
$p^{\prime}$ on the same space by assigning an independent random variable
$U_{x,y}$ to each vertex $(x,y) \in V$ which is uniformly distributed on $(0,1)$. 
The vertex is open if $U_{x,y} < $ the parameter value and closed otherwise. The only difference from \cite{DurretOr} is that there these random variables are assigned to bonds and that the set of vertices of the graph is different, i.e. $\{(x,y)$   s. t. $x + y $ is even$\}$.  

In the remaining part of the proof  we prove equation (\ref{eq:pcgamma2}). 
Observe that if $\gamma > 0$, then $\overline{r}_m - \alpha m + \frac{\gamma}{2} m =
\overline{r}_m - \frac{\alpha + \beta}{2} m  \longrightarrow \infty$ and 
$\overline{\ell}_m - \beta m  - \frac{\gamma}{2} m = \overline{\ell}_m - \frac{\alpha + \beta}{2} m  \longrightarrow -\infty$
almost surely. 
Then there exists an integer $M < \infty$ such that,
\begin{align}
\label{eq:pcgamma5}
\mathbb{P}_p ( \forall m,  \, \, \overline{r}_m > \frac{\alpha + \beta}{2} m - M ) \geq 0.51, \\
\label{eq:pcgamma6}
\mathbb{P}_p ( \forall m, \, \, \overline{\ell}_m < \frac{\alpha + \beta}{2} m + M ) \geq 0.51.
\end{align}
Secondly we introduce the following notation. If $A \subset (-\infty, + \infty)$,
then we let 
\begin{align}
\label{eq:notation1}
\xi_m^A:& = \{ x \, : \, \exists y \in A \, \mbox{ s.t } (y,t) \rightarrow (x, t-m) \}, \\
\label{eq:notation2}
r_m^A : &=  \sup \xi_m^A, \\
\label{eq:notation3}
\ell_m^A :& = \inf \xi_m^A, \\
\label{eq:notation4}
\tau^A :&=\inf\{m \, : \, \xi^A_m = \varnothing \}.
\end{align}
Repeating the proof of (\ref{eq:proprl1}), (\ref{eq:proprl2}), (\ref{eq:proprl3})
(see also \cite[Section 3, equations 10]{DurretOr}),
it follows that
\begin{equation}
\begin{split}
\label{eq:pcgamma7}
\tau^{[-M, M]} & = \inf\{m \, : \, r_m^{[-M, M]} < \ell_m^{[-M, M]} \} \\
& = \inf\{m \, : \, r_m^{[-\infty, M]} < \ell_m^{[-M, \infty]} \} .
\end{split}
\end{equation}
The previous equality implies that,
$$ 
\{ \tau^{[-M, M]} = \infty \} \supset \{ \ell_m^{[-M, \infty)} \leq \frac{\alpha + \beta}{2}
m \leq r_m^{(-\infty, M]}, \, \, \, \forall m \}.
$$
As 
$$
\mathbb{P}_p( r_m^{(-\infty, M]} > \frac{\alpha + \beta}{2} m, \, \, \, \forall m) =
\mathbb{P}_p( r_m^{(-\infty, 0]} > \frac{\alpha + \beta}{2} m - M,\, \, \, \,  \forall m),
$$
and 
$$
\mathbb{P}_p( \ell_m^{[-M, +\infty)} < \frac{\alpha + \beta}{2} m, \, \, \, \forall m) =
\mathbb{P}_p( \ell_m^{[0, + \infty)} > \frac{\alpha + \beta}{2} m + M,\, \, \, \,  \forall m),
$$
it follows that
$$
\mathbb{P}_p( \xi_m^{[-M, M]} \neq \varnothing,\, \, \,  \forall m ) \geq 0.02.
$$
Since, $\forall M > 0$,
$$
\mathbb{P}_p( \xi^0_M \supset \mathbb{Z} \cap [-M,M]) > 0,
$$
it follows that $\mathbb{P}_p( E^t) >0$. Then $p>p_c$.
\end{proof}
The next estimates have been proved in \cite{Griffeath}. The proof can be found also in \cite[Section 7, equations (1) and (2)]{DurretOr}.  In particular equation (\ref{eq:exponentialspeed}) holds for a wide class of percolation models in the subcritical regime (see \cite{Aizenmann} for a proof in a very general setting).
\begin{prop}
\label{prop:subcritical}
Recall Definition \ref{def:opensitesandpaths}.
For every $p$, let  $a(p) > \alpha(p)$ and $b(p) < \beta(p)$.
If  $p < p_c$ there exist some positive constants $h, h_2, h_3, C_2, C_3$ (dependent on $p$) such that,
\begin{align}
\label{eq:exponentialspeed}
& \mathbb{P}_p (  (0,m) \overset{ \mathcal{G}_{\mathcal{U}}}{\longrightarrow} S^0  )  \leq \exp ( - h \, m), \\
\label{eq:exponentialspeed2}
& \mathbb{P}_p (  \overline{r}_m > a\, m )  \leq  C_2 \exp ( - h_2\,  m),  \\
\label{eq:exponentialspeed3}
& \mathbb{P}_p (  \overline{\ell}_m < b\, m )  \leq  C_3 \exp ( - h_3 \, m).
\end{align}
\end{prop}
\begin{proof}
We sketch the proof of (\ref{eq:exponentialspeed}), which is similar to the proof of
\ref{eq:exponentialspeed2}) and \ref{eq:exponentialspeed3}).
If $p< p_c$, then from equation (\ref{eq:pcgamma0}) $\alpha < \frac{\alpha + \beta}{2} < \beta$. Thus there exists an $N$ large enough such that $E[\overline{r}_{0,N}] < \frac{\alpha + \beta}{2} N$,  $E[\overline{\ell}_{0,N}] > \frac{\alpha + \beta}{2} N$.
By using the subadditivity property of $\overline{r}_{m,n}$ and $\overline{\ell}_{m, n}$ one can see that,
\begin{equation}
\begin{split}
\overline{r}_{mN} - \frac{\alpha +\beta}{2} m N & \leq S_m := \overline{r}_{0, N} - \frac{\alpha +\beta}{2} N + 
\overline{r}_{N, 2N} - \frac{\alpha +\beta}{2} N \\ & + \ldots + \overline{r}_{(m-1)N, mN} - \frac{\alpha +\beta}{2} N ,\\
\frac{\alpha +\beta}{2} m N  - \overline{\ell}_{mN}  & \leq S^{\prime}_m := \frac{\alpha +\beta}{2} N - \overline{\ell}_{0, N}  + 
\frac{\alpha +\beta}{2} N - \overline{\ell}_{N, 2N} \\ &  \ldots + \frac{\alpha +\beta}{2} N - \overline{\ell}_{(m-1)N, mN} .
\end{split}
\end{equation}
The right hand side of the two previous inequalities is a random walk with expectation
respectively $E[S_1] < 0 $, $E[S_1^{\prime}] < 0 $.
As $S_1 \leq s_u N$,  $S_1^{\prime}  \leq s_u N$,
then $\varphi(\theta) : = E[\exp(\theta S_1) ] < \infty$ 
and $\varphi^{\prime}(\theta) : = E[\exp(\theta S^{\prime}_1) ] < \infty$ 
for all $\theta > 0$.
From the considerations in \cite{DurretOr}
it follows that we can pick $\theta_0 >0$ with $\varphi(\theta_0) < 1$
and $\varphi^{\prime}(\theta_0) < 1$  such that,
$$
\mathbb{P}_p ( S_m \geq 0) \leq E[\exp(\theta_0 S_m ) ] = \varphi(\theta_0)^m,
$$
$$
\mathbb{P}_p ( S^{\prime}_m \geq 0) \leq E[\exp(\theta_0 S_m ) ] = \varphi^{\prime}(\theta_0)^m.
$$
This implies that $\mathbb{P}_p( \overline{r}_{mN} \geq \frac{\alpha + \beta}{2} m N) \longrightarrow 0$ and $\mathbb{P}_p( \overline{\ell}_{mN} \leq \frac{\alpha + \beta}{2} m N) \longrightarrow 0$ exponentially fast.
Observe also that as $\mathbb{P}_p (\xi_m = \varnothing ) \geq \mathbb{P}_p ( \overline{r}_m < \frac{\alpha + \beta}{2} m < \overline{\ell}_m )$, then
$$\mathbb{P}_p ( \xi_m \neq \varnothing) \leq 
\mathbb{P}_p ( \overline{r}_m \geq \frac{\alpha + \beta}{2} m ) + 
\mathbb{P}_p ( \overline{\ell}_m \leq \frac{\alpha + \beta}{2} m ).$$
This implies (\ref{eq:exponentialspeed}).
\end{proof}
We end this section recalling a property proved in \cite{Mityushin}.
As the reference is in Russian, we sketch its proof below.
\begin{prop}
\label{prop:Mityushin}
Recall Definition \ref{def:opensitesandpaths}. For every $t, n \in \mathbb{N}$,
\begin{equation}
\label{eq:ineqfininf}
\mathbb{P}^n_p ( (0,t)  \overset{ \mathcal{G}_{\mathcal{U}}(n)}{\longrightarrow} S^0 ) \leq  
\mathbb{P}_p ( (0,t) \overset{ \mathcal{G}_{\mathcal{U}}}{\longrightarrow} S^0  ).
\end{equation}
\end{prop}
\begin{proof}
Observe that in $\mathbb{Z} \times \mathbb{Z}$ all paths of length $t$ starting from $(0,t)$ lie within $\Delta = [s_1 t, s_u t] \times [0,t-1]$. At each point we have a random variable $\omega_{x,y}$ that is equal to $1$ with probability $p$ and to $0$ with probability $1-p$ and these random variables are mutually independent.
We consider the same set $\Delta$ but with a different set of random variables $z_{x,y}$. Each $z_{x,y}$ is equal to $1$ with probability $p$ and $0$ with probability $1-p$, but these random variables are not independent. Namely, for all $(x,y)$, the random variables $z_{x + 2kn, y}$ for all integers $k$ such that $(x + 2kn,y) \in \Delta$ 
have the same outcome (i.e. they are ``synchronized'').
This model is equivalent to the model on the cylinder $\Delta_n \times [0, t-1]$ (i.e. with periodic boundaries), where $\omega_{x,y}$ are independent, because in these two models their probabilistic spaces and sets of open paths starting at $(0,t)$ are isomorphic.

Let then $\mathcal{W}$ be the set of all possible paths of length $t$ from $(0,0)$. We will show that ``synchronization'' does not increase the probability of the existence of an open path of length $t$ on $\Delta$.

Let then $\theta_{x,y}$ be some random variables with values $0$ or $1$ associated with $(x,y) \in \Delta$. Consider the function $Z$, with arguments 
$\theta_{x,y}$,
$$
Z = \sum_{h \in \mathcal{W}} \prod_{(x,y) \in h} \theta_{x,y}.
$$
Then $Z \geq 0$ and $Z > 0$ if and only if there exists an open path. Suppose that at the beginning $\theta_{x,y} = \xi_{x,y}$,
for all $(x,y) \in \Delta$ and at each step we ``synchronize'' the variables $\theta_{a + 2kn, b}$ for a certain $(a,b)$ until we get $\theta_{x,y} = z_{x,y}$ for all $(x,y) \in \Delta$. We will show that each synchronization step does not increase $Z$.
To do this, we write 
$$
Z = \sum\limits_{k : (a + 2kn, b) \in \Delta} \theta_{a + 2kn, b} \, \, f_k (\tilde{\theta}) + g (\tilde{\theta}),
$$
where $\tilde{\theta}$ is the set of all $(x,y) \neq (a +2kn, b)$, i.e. they are independent from the group $\theta_{a + 2kn, b}$.
The $f_k$ and $g$ are some functions with non-negative integer values. Here we use the fact that a path can contain only one point of the form $(a + 2kn, b)$, so different
$\theta_{a + 2kn, b}$ don't multiply.
Before the ``synchronization'' step,
$$
Z = Z_1 = \sum\limits_{k : (a + 2kn, b) \in \Delta} \omega_{a + 2kn, b} \, \, f_k (\tilde{\theta}) + g(\tilde{\theta}),
$$
and after it,
$$
Z = Z_2 = z_{a,b} \cdot \sum\limits_{k : (a + 2kn, b) \in \Delta} f_k (\tilde{\theta}) + g(\tilde{\theta}).
$$
It is easy to show that, fixing any value of the set $\tilde{\theta}$, $P(Z_1 > 0) \geq P(Z_2 >0)$. Hence, the same is true when $\tilde{\theta}$ is not fixed.
\end{proof}

\subsection{Proof of Theorem \ref{theo:maintheorem2}}
\label{sect:proofTheo2}
Recall the definitions provided just before the statement of the theorem.
Along the whole proof we denote by $\mathbb{P}^n ( \, \cdot \, )$ the Bernoulli product measure in $\Sigma$, where the space is finite, and by $\mathbb{P}( \, \cdot \, ) $ the Bernoulli product measure in $\Sigma$, where the space is infinite. The proof is based on the estimation of $\mathcal{E}^n_{{\delta}_{\mathbf{1}}}  ( \tau_n > t )$,  
which gives the expectation,
\begin{equation}
\label{eq:time}
\mathbb{E}^{(n)}_{{\delta_1}} [ \tau_n ]  =  \sum_{ t =0 }^{\infty } \mathcal{E}^n_{{\delta}_{\mathbf{1}}} ( \tau_n > t ),
\end{equation}

We prepare the reader to the proof of the left inequality of the statement (b). The proofs of the other inequalities do not need an introduction, as they are simpler. The proof is based on the estimation of the probability of the event $\{ \tau_n > t \}$.
In order to provide this estimation, first we define the event $\mathcal{D}_{n,t,a}$,
whose probability is less than the probability of  $\{ \tau_n > t \}$.  
The event occurs if a path connects $[[-n, n-1]] \times \{t\}$ to the line $y=0$ without crossing the diagonal sides of a parallelogram (a rigorous definition is given later). This allows to reduce the estimation of $\{ \tau_n > t \}$ to the estimation of the probability of an event that is simpler to study, as periodic boundaries play no role.

As the neighbourhood of the model is in general non symmetric, the cluster of vertices belonging to an open path starting from $(0,t)$ (which is infinite with positive probability, as $p > p_c$) will
typically have a drift. Indeed, recall Proposition \ref{prop:convergencesubadditive}
and the fact that $\overline{r}_t  \sim  t (\frac{\alpha + \beta}{2}  + \frac{\gamma}{2})$ and $\overline{\ell}_t \sim t (\frac{\alpha + \beta}{2} - \frac{\gamma}{2})$, $\beta \leq \alpha$.  Thus, as $p$ is slightly larger than $p_c$, then typically the cluster of vertices will be centred around $\sim \frac{\alpha + \beta}{2} t$.
Hence, the diagonal sides of the parallelogram is chosen in such a way that in the limit $t \rightarrow \infty$ the cluster has typically a non-empty intersection 
with the parallelogram. With this choice, the probability of the event $\mathcal{D}_{n,t,a}$ does not go to zero too fast as $t$ grows.

Later we introduce a change of coordinates $T_b^t$ that allows to simplify the notation,
by transforming the graph in a new graph, where the cluster of vertices connected by an open path starting from $(0,t)$ (namely in the new graph $\frac{\alpha + \beta}{2} = 0$). We provide a lower bound for the probability of $\mathcal{D}_{n,t,a}$ by introducing a new event $\mathcal{H}_n$ and by using the FKG inequality to bound the probability of  $\mathcal{D}_{n,t,a}$ with a product of probabilities of events $\mathcal{H}_n$.

In the last part of the proof we estimate the probability of the event $\mathcal{H}_n$,
showing that it goes to $1$ fast enough with $n$ for any $p >p_c$. This estimation is stated in Proposition \ref{prop:inequality2}.  

\paragraph{Proof of part (a).} 
The proof of the left inequality of the statement (a) can be found in
\cite[Section 2]{Stavskaja3} together with an estimation of the constants, so we do not provide it here.
Indeed, the left inequality holds for any $p \in [0,1]$.
The proof of the right inequality of statement (a) is an application 
of the estimates presented in Section \ref{sect:percolationestimates}.
Starting from (\ref{eq:time}),
\begin{equation}
\label{eq:partea}  
\begin{split}
\mathbb{E}^{(n)}_{{\delta_1}} [ \, \tau_n   \, ]
&= \sum_{ t =1 }^{\infty } \mathbb{P}_p^n
 ( \bigcup_{s= -n}^{n-1} \{ (s,t) \overset{ \mathcal{G}_{\mathcal{U}}(n)}{\longrightarrow}  S^0 \}) \\
& \leq \sum_{ t =1 }^{\infty } \min \{ 1, 2  n \, \mathbb{P}^n(  (0,t) \overset{ \mathcal{G}_{\mathcal{U}}(n)}{\longrightarrow}  S^0 )  \} \\
& \leq  \sum_{ t =1 }^{\infty } \min \{ 1, 2  n \exp(- h t)  \}  \\
& \leq  \frac{\log( 2\, n)}{h} + K,
\end{split}
\end{equation}
where $K$ is some positive constant. 
In the first equality we used Proposition \ref{prop:taunconnection}, in the second 
inequality we used the union bound 
and translation invariance, in the second-last inequality we used (\ref{eq:exponentialspeed}) and (\ref{eq:ineqfininf}).
The algebraic tricks of (\ref{eq:partea}) have been used also in the proof of \cite[Proposition 8.6]{ToomDiscr}.
\paragraph{Proof of part (b).}
The proof of the right inequality of the statement (b) is trivial. We 
define a new process $(q_x^t)_{x \in \mathbb{S}_n}$ where
every $q_x^t$ is $1$ with probability $p$ and $0$ with probability $1-p$ independently.
Observe that for all $x \in \mathbb{S}_n$, $t \in \mathbb{Z}_+$,
$q_x^t \geq \eta_x^t$, as long as the two processes are driven by the same random process. Hence, the expected convergence time for the former is larger than the convergence time for the latter. By a simple computation, the expected convergence time for the system $q^t$ is $\frac{1}{1-p}^{2n}$. This implies the inequality.

We start with the proof of the left inequality.
For every $a \in \mathbb{R}$ we define the event,
\begin{equation}
\begin{split}
\mathcal{D}_{n ,t, a} :=  \{ & \exists \,   x \in [[-n , n-1]], \, \, 
\mbox{ such that } (x, t) \mbox{ is connected to } S^0   \\ &
\mbox{ by an open path in } 
\mathcal{G}_{\mathcal{U}}  \mbox{ that never crosses
the lines } \\ & y  =     n -1  - a (x-t),   \, \, 
y  =     - n - a (x-t)\, \, \,  x \in \mathbb{R}  \},
\end{split}
\end{equation}
which is a subset of $\Omega$,
recalling that $S^0$ denotes the set of vertices belonging to the line $y=0$.
See also Figure \ref{Fig4-5:Dn} - up for a representation.
\begin{figure}
\begin{center}
\includegraphics[scale=0.50]{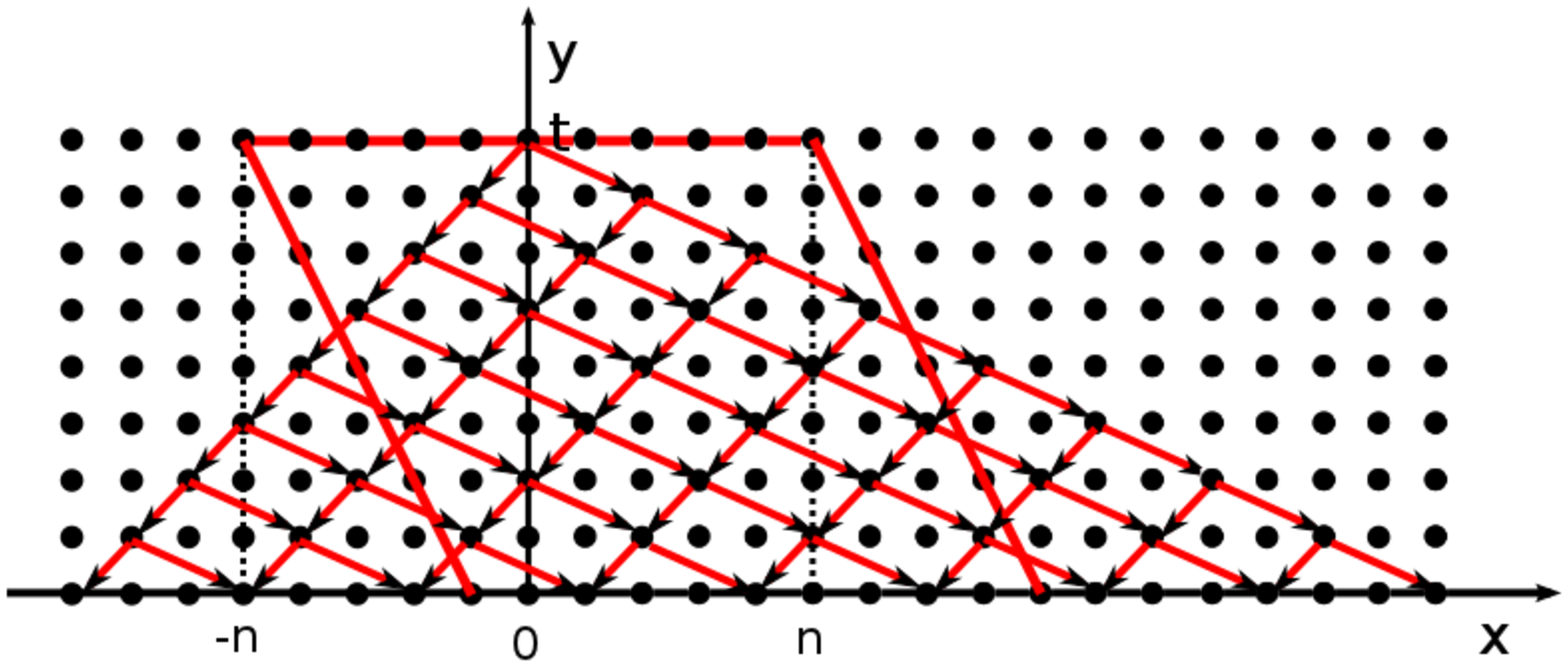}
\includegraphics[scale=0.55]{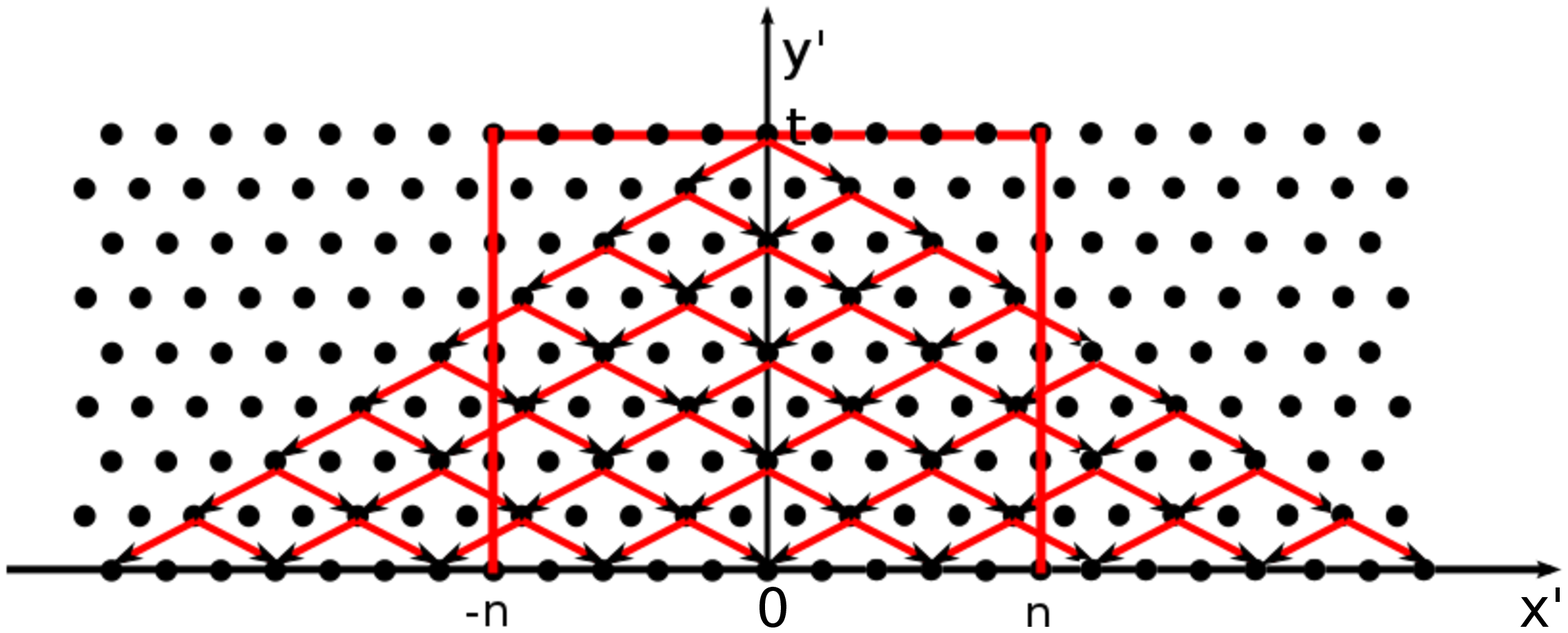}
\caption{\textit{Up}: Representation of $\mathcal{G}_{\mathcal{U}}$ in case of neighbourhood $\mathcal{U} = \{ -1, 2 \}$. For graphical reasons
only edges belonging to the evolution cone of $(0, t)$ have been drawn. In the figure $a = \frac{s_1 + s_u}{2}$. \textit{Down}: the same graph of the figure above, transformed via (\ref{eq:transf}) with parameter $b = a$.}
\label{Fig4-5:Dn}
\end{center}
\end{figure}
Recall Definition \ref{def:opensitesandpaths} and observe that, 
\begin{equation}
\label{eq:ineq}
\mathbb{P}_p ( \mathcal{D}_{n,t, a}  ) \leq
 \mathbb{P}^n_p (\, \, \exists x \in [[-n, n-1]] \, \, \mbox{s.t.} \, \, 
(x,t) \overset{ \mathcal{G}_{\mathcal{U}}  (n) }{\longrightarrow} S^0) =
\mathcal{E}^n_{\delta_{\mathbf{1}}} ( \tau_n > t  ).
\end{equation}
Observe that the quantity on the left is defined in the infinite system and the quantities
in the middle and on the right are defined on the finite system with periodic boundaries.
We provide a proof of the statement below.
\begin{proof}[Proof of (\ref{eq:ineq})]
Consider two graphs, $\mathcal{G}_{\mathcal{U}}^i$ and $\mathcal{G}_{\mathcal{U}}^f$. The former is defined on the infinite space $\mathbb{Z}\times \mathbb{N}_0$ and the latter on the finite space $\mathbb{S}_n \times \mathbb{N}_0$ with periodic boundaries, as defined in Section \ref{sect:themodel}. Let $Q_{a,t} \subset \mathbb{Z} \times \mathbb{N}_0$ be the region inside the parallelogram identified by the points $(-n, t)$, $(n-1, t)$, $(-n + at, 0)$, $(n-1 + at, 0)$ (see Figure \ref{Fig:proof}).
The event $\mathcal{D}_{n,a,t} \subset \{0,1\}^{\mathbb{Z} \times \mathbb{N}_0 }$ 
occurs if an open path connects $[-n, n-1] \times \{t\}$ to 
$[-n+at, n-1 + a \, t] \times \{ 0 \}$ without ever crossing the diagonal sides of the parallelograms. 
We couple the model on the finite space and the model on the infinite space in the following way. Namely, call $\omega_{x,y} $, for all $(x,y) \in \mathbb{Z} \times \mathbb{N}_0$, and $z_{x,y}$, for all $(x,y) \in \mathbb{S}_n \times \mathbb{N}_0$,
the random variables taking values $0$ or $1$ independently. The coupling is such that for all $(x,y) \in Q_{a,t}$, $\omega_{x,y} = z_{x^{\prime}, y^{\prime}}$,
where $x^{\prime} = | x + n |_{2n} -n$, $y^{\prime} = y$,
where $|x|_n$ denotes $x \mod n$. The random variables $\omega_{x,y}$ associated to sites $(x,y)$ not contained in $Q_{a,t}$ are not coupled.
Observe that for every $(x,y) \in Q_{a,t}$ there exists a unique
$(x^{\prime}, y^{\prime})$ in $\mathbb{S}_n$ and vice versa.
Recalling that boundaries of $\mathcal{G}_{\mathcal{U}}^f$ are periodic, one can observe from Figure \ref{Fig:proof} that, as long as there exists an open path in $\mathcal{G}_{\mathcal{U}}^i$ connecting the top to the bottom of $Q_{a,t}$ and never crossing its diagonal sites (e.g. the path represented by a continuous curve in the figure),
there exists also an open path in $\mathcal{G}_{\mathcal{U}}^f$ connecting
$[-n, n-1] \times \{t\}$ to $[-n, n-1] \times \{0\}$ (i.e. the path represented by a dashed curve in the figure). This implies the statement.
\begin{figure}
\begin{center}
\includegraphics[scale=0.35]{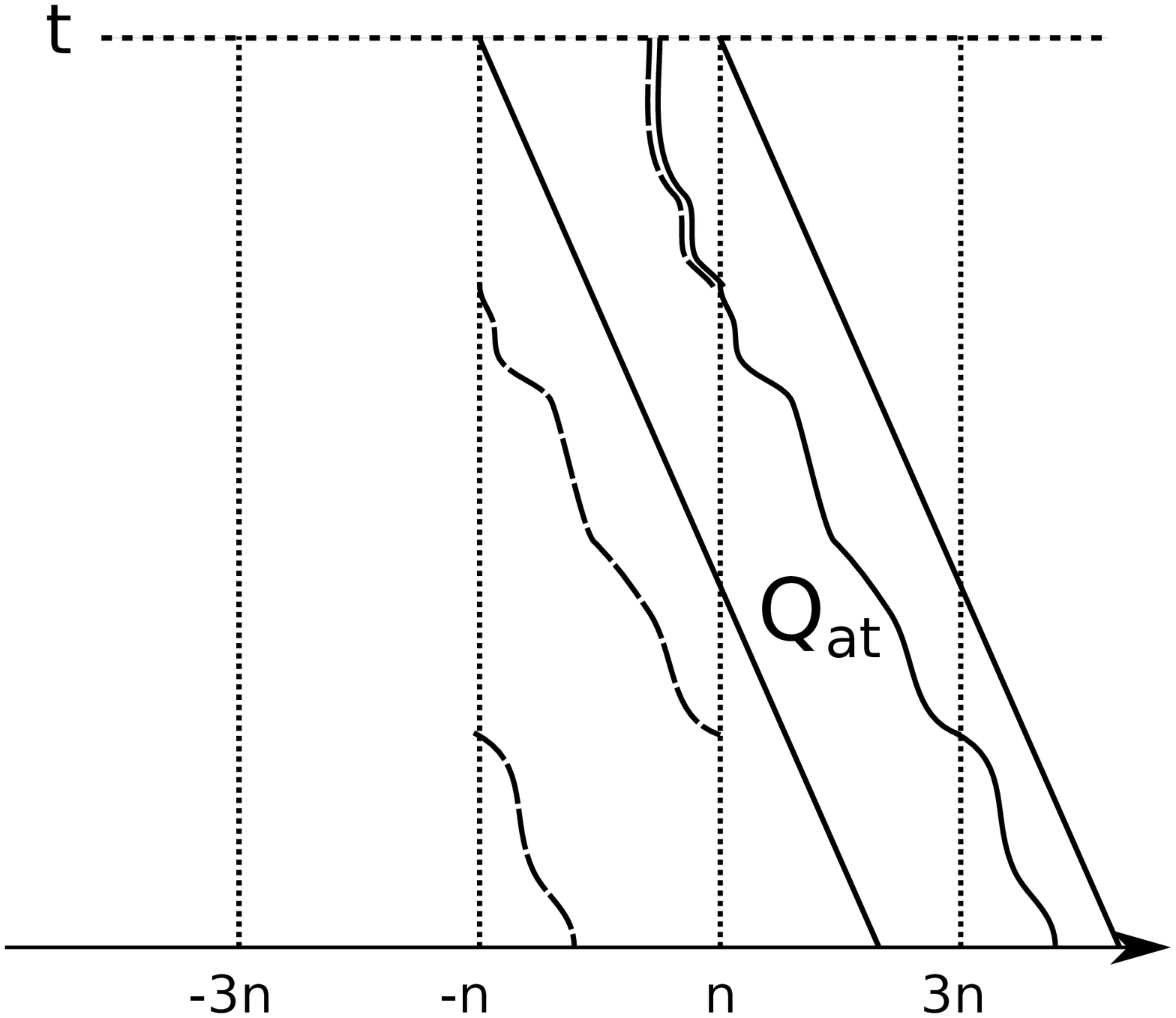}
\caption{ }
\label{Fig:proof}
\end{center}
\end{figure}
\end{proof}
Consider now the following change of coordinates,
\begin{equation}
\label{eq:transf}
\left\lbrace 
\begin{array}{cc}
\begin{split}
x^{\prime} & =    x - b (t - y)  \\
y^{\prime} & = y \\
\end{split}
\end{array}
\right.  ,
\end{equation}
under which the graph $\mathcal{G}_{\mathcal{U}}$ is transformed
into the new graph $T^t_b \,  \mathcal{G}_{\mathcal{U}}$.
We denote by $T^t_b \,  \mathcal{D}_{n, t, a}$ 
the event $\mathcal{D}_{n, t, a}$, defined for the graph 
$T^t_b \,  \mathcal{G}_{\mathcal{U}}$,
(i. e. replace $\mathcal{G}_{\mathcal{U}}$ with $T^t_b \,  \mathcal{G}_{\mathcal{U}}$
in the definition of the event above). See Figure \ref{Fig4-5:Dn} for an example.  
%From now on we will use the same
%notation for all events. This means that
%if a certain event $\mathcal{E}$ is defined for 
%the graph $\mathcal{G}_{\mathcal{U}}$, then 
%the event $T^t_b \, \mathcal{E}$ 
%is defined for the transformed graph 
%$T^t_b \mathcal{G}_{\mathcal{U}}$.
 The following equation holds,
\begin{equation}
\label{eq:transformation}
\mathbb{P}_p ( T^t_b \,  \mathcal{D}_{n, t, a} ) = \mathbb{P}_p (  \,  \mathcal{D}_{n, t, a-b} ),
\end{equation}
as the change of coordinates preserves connection between vertices.
Now we introduce the event $\mathcal{H}_n$,
\begin{equation}
\label{eq:evH}
\begin{split}
\mathcal{H}_n =  \{ &\exists \, y\,  , \, y^{\prime} \, \,  \mbox{s.t.} \, \,   y \in [[4n, 6n]], \, y^{\prime} \in [[0, 2n]]   \\
						      & \mbox{and }  (-n, y) \overset{\mathcal{G}_{\mathcal{U}}}{\longrightarrow}  (n, y^{\prime}) \},
\end{split}
\end{equation}
which is represented in Figure \ref{Fig6-7:events}-right. 
The following proposition is about this event. 
\begin{prop}
\label{prop:inequality2}
For any $p > p_c$ there exist  positive constants $A, b$ (dependent on $p$) such that
for any $t \in \mathbb{N}$ and for $n$ large enough,
\begin{equation}
\label{eq:inequality2}
\mathbb{P}_p(T_{\frac{ \alpha + \beta}{2}}^t \mathcal{H}_n )  \geq  \, 1 - A \,  \exp(  - b \, n).
\end{equation}
\end{prop}
As before, the event $T_{\frac{ \alpha + \beta}{2}}^t \mathcal{H}_n$ denotes the occurrence
of $\mathcal{H}_n$ in the graph $T_{\frac{ \alpha + \beta}{2}}^t \mathcal{G}_{\mathcal{U}}$.
We recall that $\alpha$ and $\beta$ are defined in Section \ref{sect:percolationestimates}.
We first use Proposition \ref{prop:inequality2} to conclude the proof of the theorem and later we prove the Proposition \ref{prop:inequality2}.
Define then the new event $\mathcal{F}_{n,t}$,
which is represented in Figure \ref{Fig6-7:events}.
$\mathcal{F}_{ n, t }$ occurs \textit{iff} (a) and (b) hold:
\begin{enumerate}
\baselineskip=17pt
\item[(a)] 
for every odd $j \in [[0, \frac{t}{ 2 n }  ]]$
there is a vertex 
$(-n  , y)$, with $y \in [[    2 n j ,    2 n (j+1) ]]$,
connected to $(  n , y^{\prime})$
by an open path in 
$\mathcal{G}_{\mathcal{U}}$,
with $y^{\prime}  \in [[   2 n (j-2) ,  2 n (j-1)   ]]$,
\item[(b)] for any even $j \in [[0, \frac{t}{ 2 n} ]]$
there is a vertex $(n, y)$,
with $y \in [[  2 n j , 2 n (j+1) ]]$,
connected by an open path in 
$\mathcal{G}_{\mathcal{U}}$
to $( - n, y^{\prime})$,
with $y^{\prime} \in [[  2 n  (j-2) , 2 n (j-1) ]]$.
\end{enumerate}
Note first that,
\begin{equation}
\label{eq:ineq22}
\mathbb{P}_p (T_{\frac{ \alpha + \beta}{2}}^t \, \mathcal{F}_{n, t}) \leq \mathbb{P}_p ( T^t_{\frac{ \alpha + \beta}{2}} \, \mathcal{D}_{n, t, 0 }),
\end{equation}
because if  $\mathcal{F}_{n, t }$ occurs, then the top of the box $2n \times t$
is connected to the bottom by a path
that never goes out from the box (compare figures \ref{Fig6-7:events}-left and \ref{Fig6-7:events}-middle).
Secondly, we observe that the event $T_{\frac{\alpha + \beta}{2}}^t \mathcal{F}_{n, t}$
equals the intersection of $\lfloor \frac{t}{n} \rfloor$
events of type $T_{\frac{\alpha + \beta}{2}}^t\,  \mathcal{H}_{n} $, represented in Figure
\ref{Fig6-7:events}-right.
\begin{figure}
\begin{center}
\includegraphics[scale=0.40]{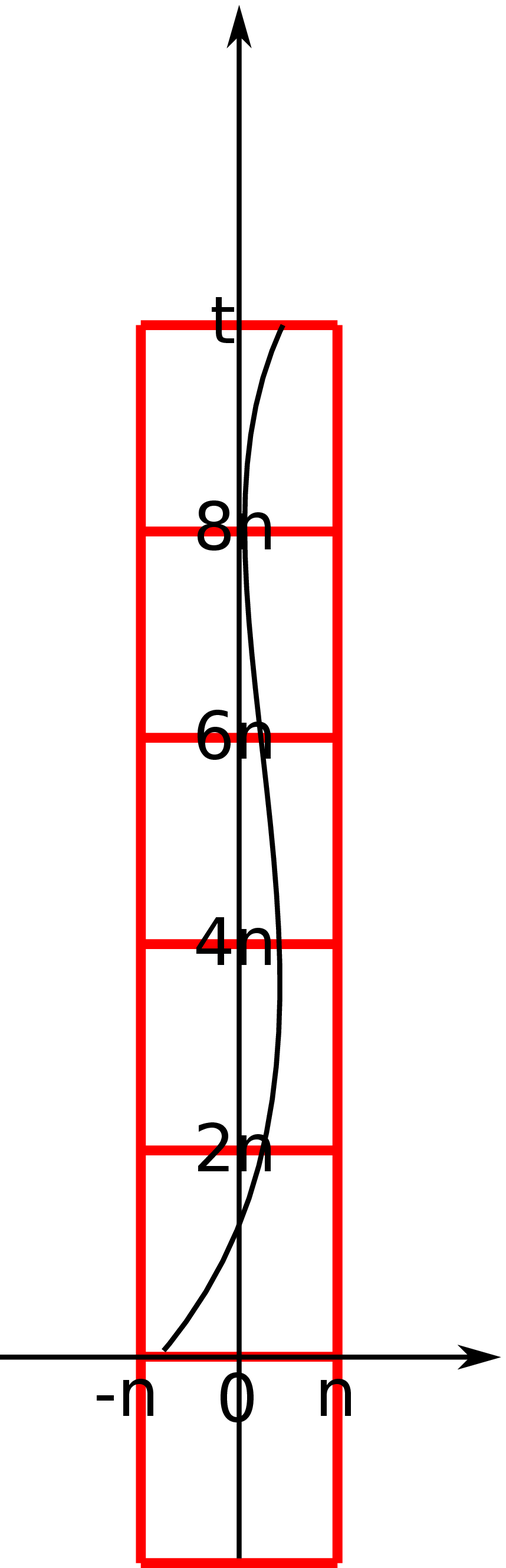}
\includegraphics[scale=0.40]{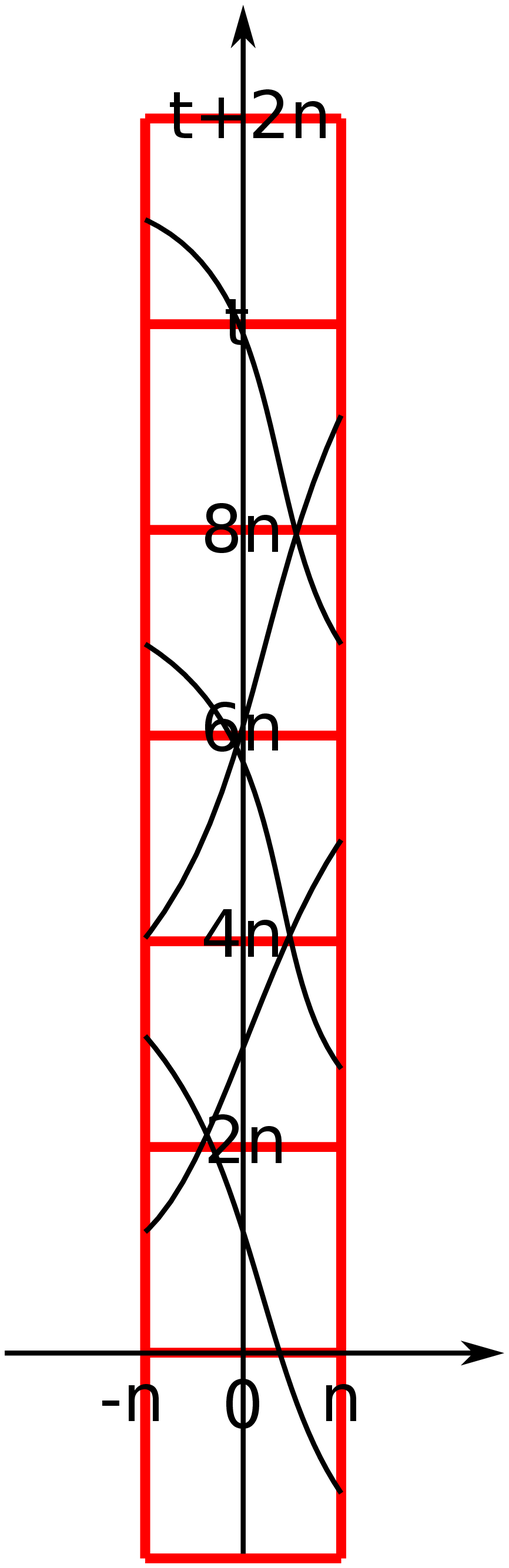}
\includegraphics[scale=0.40]{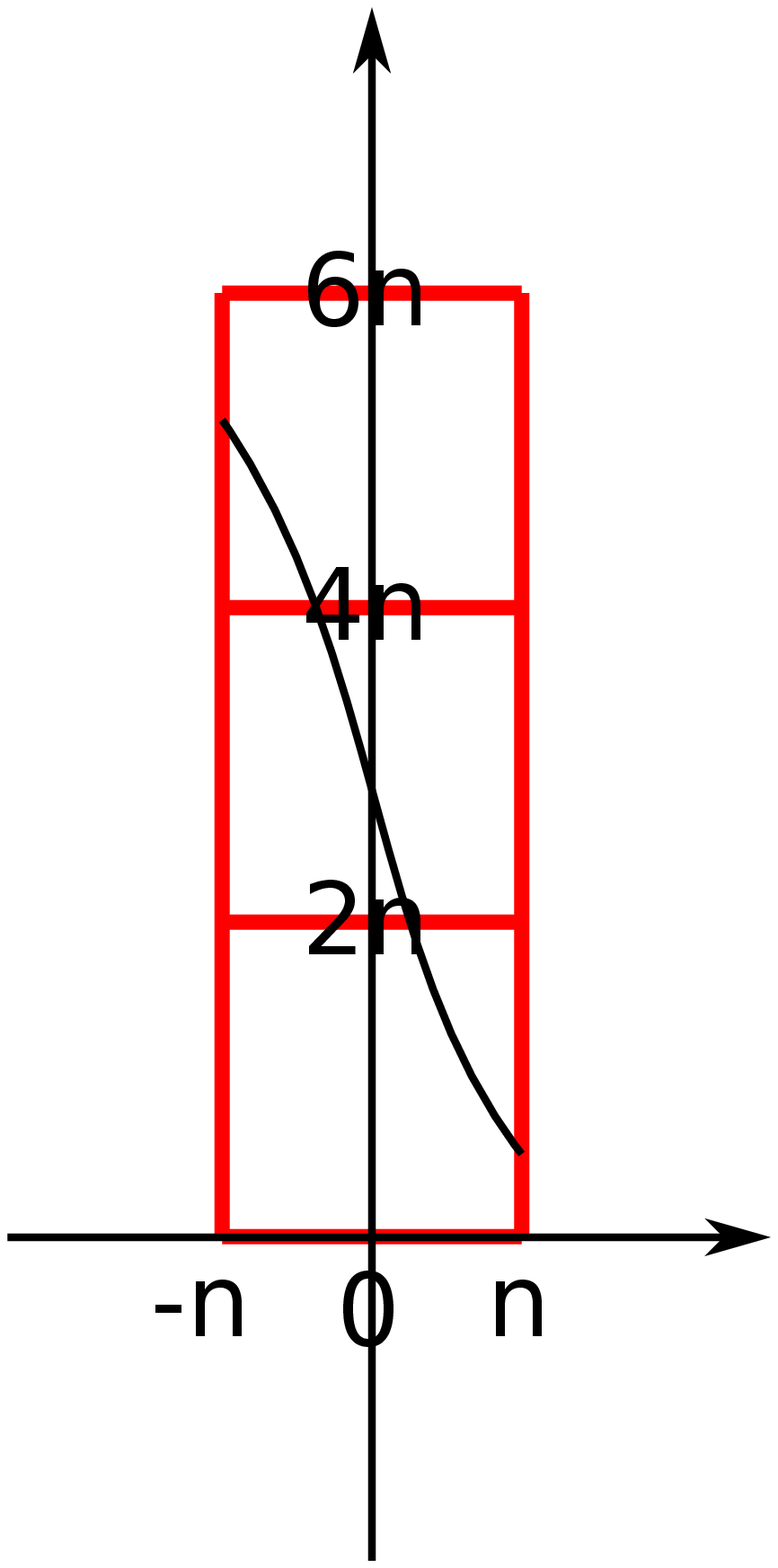}
\caption{\textit{Left}: representation of the event $\mathcal{D}_{n,t,0}$. \textit{Center}: representation of the event $\mathcal{F}_{n, t}$. \textit{Right}: representation of the event $\mathcal{H}_n$. In all figures the details on the structure of the graph have been omitted.}
\label{Fig6-7:events}
\end{center}
\end{figure}
As the event $\mathcal{H}_n$ increasing, 
the FKG inequality is applicable, i.e.
\begin{equation}
\label{eq:FKG}
  { \mathbb{P}_p ( T_{\frac{\alpha + \beta}{2}}^t \, \mathcal{H}_n )} ^{\lfloor \frac{t}{n} \rfloor}  \leq \mathbb{P}_p ( T_{\frac{\alpha + \beta}{2}}^t \, \mathcal{F}_{n, t} )
\end{equation}
Then using (\ref{eq:inequality2}) finally we get,
\begin{equation}
\label{eq:FKG2}
\mathbb{P}_p(  T_{\frac{\alpha + \beta}{2}}^t\,  \mathcal{F}_{n, t}  ) \geq \ {( 1 -A \cdot n \exp(  - n b  ) )}^{ \lfloor \frac{t}{n} \rfloor }
\end{equation}
Then, from (\ref{eq:time}) and for $n$ large enough,
\begin{equation}
\begin{split}
 \mathbb{E}^{(n)}_{{\delta_{\mathbf{1}}}}[\tau_n] 
 &  \geq \sum_{ t =1 }^{\infty } { \mathbb{P}_p ( T_{\frac{\alpha + \beta}{2}}^t \, \mathcal{D}_{n, t, 0}) } \\
  &  \geq \sum_{ t =1 }^{\infty } { \mathbb{P}_p( T_{\frac{\alpha + \beta}{2}}^t\,  \mathcal{F}_{n, t})   ) } \\
 &  \geq \sum_{ t =1 }^{\infty } { ( 1 - A \cdot \exp (- b n   ) )}^{\lfloor \frac{t}{ n } \rfloor} \\
 &  \geq j  (1 - \frac{A \cdot e^{-b n} j }{n} ),
\end{split} 
\end{equation}
where  $j$ is an arbitrary integer. In the previous expression we have used Proposition
\ref{prop:taunconnection}, (\ref{eq:ineq})
and (\ref{eq:transformation})
for the first inequality, (\ref{eq:ineq22}) for the second inequality
and (\ref{eq:FKG2}) for the third one. 
Choosing finally  $j = \lfloor \frac{ n e^{b n}}{2A} \rfloor $, the part (b) of the theorem follows.

\begin{proof}[\textbf{Proof of Proposition \ref{prop:inequality2}}]
We prepare the reader to the proof of the proposition and later we present the proof.
We consider two graphs,
$T_{ \frac{\alpha + \beta}{2} }^t \, \mathcal{G}_{\mathcal{U}} = ( V^{1}, \vec{E}_{\mathcal{U}}^{1})$ and
$T_{\frac{s_1 + s_u}{2}}^t \, \mathcal{G}_{\mathcal{U}} = ( V^{2}, \vec{E}_{\mathcal{U}}^{2 })$, recalling the definitions of $\alpha$ and $\beta$ in Section \ref{sect:percolationestimates}
and  the definition of the transformation $T_{\,\cdot\,}^t$ provided in (\ref{eq:transf}).
Observe that vertices of both graphs could take non integer positions.
The proof is divided in two parts. 

In the \textbf{first part}  we generalize the dynamic-block argument presented in \cite{DurretOr} to the percolation model considered in this article. The idea of the construction is the same of \cite{DurretOr}, although parameters of the construction have been adapted to the lack of symmetry. The lack of symmetry involves the structure of the graph $\mathcal{G}_{\mathcal{U}}$ and the slope of $\overline{r}_m$ and $\overline{\ell}_m$, as in general $\alpha \neq -\beta$. Two different spatial transformations have been used in order to recover the symmetric setting and simplify the construction, namely $T_{ \frac{s_1 + s_2}{2} }^t$ and $T_{ \frac{\alpha + \beta}{2} }^t$. 

The argument is based on a coupling between realisations of the graph
$T_{\frac{s_1 + s_u}{2}}^t \, \mathcal{G}_{\mathcal{U}}$
and those in $T_{ \frac{\alpha + \beta}{2} }^t \, \mathcal{G}_{\mathcal{U}}$.
The construction depends on a rescaling parameter $L$ and it is such that the realisation on $T_{\frac{s_1 + s_u}{2}}^t$ is a function of the realisation on $T_{ \frac{\alpha + \beta}{2} }^t \, \mathcal{G}_{\mathcal{U}}$. In  $T_{ \frac{\alpha + \beta}{2} }^t \, \mathcal{G}_{\mathcal{U}}$ every site is open with probability $p$ or closed with probability $1-p$ independently. On the contrary, the states of sites in $T_{\frac{s_1 + s_u}{2}}^t$ are not independent.
The construction is such that if the event $\mathcal{H}_n$ occurs in the former graph, then the event $\mathcal{H}_{Ln}$ occurs in the latter graph. Secondly, if $p>p_c$, then for every $\epsilon$, by choosing $L$ is large enough, every site in $T_{\frac{s_1 + s_u}{2}}^t$ is open with probability larger than $1-\epsilon$.

The \textbf{second part} we define a sub-graph 
of $T_{\frac{s_1 + s_u}{2}}^t \, \mathcal{G}_{\mathcal{U}}$,
that we call $\mathcal{L}$, for which it is easy to construct a dual graph. 
We use Peierls argument for the dual graph and we show that  $\mathbb{P}_p ( \mathcal{H}^{\mathcal{L}}_{n}) \geq 1 - A \,  \exp(  - b \, n)$. As far as we know, this estimation has not been provided in other works. The dual graph construction can be found in \cite{Uniform}.
This implies that 
$\mathbb{P}_p ( T_{ \frac{s_1 + s_u}{2} }^t \,  \mathcal{H}_{n}) \geq 1 - A \,  \exp(  - b \, n)$. Recalling the properties of the construction, it follows that
$\mathbb{P}_p ( T_{ \frac{\alpha + \beta}{2} }^t \,  \mathcal{H}_{Ln}) \geq 1 - A \,  \exp(  - b \, n)$. By rearranging the constants, the statement of the proposition follows.

We start now with the proof of the proposition.
\subparagraph{Part 1: Dynamic blocks construction}
We divide $T_{ \frac{\alpha + \beta}{2} }^t \, \mathcal{G}_{\mathcal{U}}$
into \textit{macro-regions} $R_{x,y}$ centred around the point
$C_{x, y}$, where $(x,y) \in V^{2}$ and 
\begin{equation}
\label{eq:coordinates}
\begin{split}
C_{x, y} &= ( x   \frac{\gamma}{s_u - s_1} (1-\delta) , y L    ),\\
R_{x, y} &= C_{x, y} + [ (-1 - \delta ) \frac{\gamma}{2} L, (1 + \delta) \frac{\gamma}{2} L] \times [0, -  (1 + \delta) L].
\end{split}
\end{equation}
We recall that from equation (\ref{eq:pcgamma0}) $\gamma = \alpha - \beta >0$ 
for all $p>p_c$.
 The constants $\delta$ and $L$ are positive and have to be properly chosen.
In order the argument to work rigorously,
$(1 - \delta) \gamma L$ and $L$ should be even integers.
To not complicate the exposition here we ignore these details, the same as in \cite{DurretOr}.
Each vertex $(x,y) \in V^{2}$
is associated to a random variable  $\varphi_{x,y}$ which takes value 
$1$ if a certain event $\mathcal{B}_{x,y}$ occurs in the region $R_{x,y}$ of $(V^{1}, \vec{E}_{\mathcal{U}}^{1})$ or $0$ otherwise.
In order to define such event we introduce the following points in space (see also Figure \ref{Fig8:SketchRen}),
for every $s \in \mathcal{U}$,
\begin{equation}
\label{eq:coordinates}
\begin{split}
u  &= (\frac{\delta \gamma L}{2}\, \, ,\, \,  0 ) , \\
v  &= (\frac{ 3 \delta \gamma L }{4}\, \, , \, \, 0 )  , \\
-u  &= (-\frac{\delta \gamma L}{2}\, \, ,\, \,  0 ) , \\
-v  &= (-\frac{ 3 \delta \gamma L }{4}\, \, , \, \, 0 )  , \\
u_s^R &=    ( \frac{\delta \gamma L}{2} +  ( s - \frac{s_1 + s_u}{2} ) \cdot \frac{  (1 - \delta) \gamma L}{s_u - s_1}\, \, , \, \, -L (1 + \delta) ), \\
v_s^R &=     ( \frac{ 3 \delta \gamma L }{4} +  ( s - \frac{s_1 + s_u}{2} ) \cdot \frac{  (1 - \delta) \gamma L}{s_u - s_1}\, \,  , \, \, -L (1 + \delta)),  \\
u_s^L &=   ( - \frac{\delta \gamma L}{2}  + ( s - \frac{s_1 + s_u}{2} ) \cdot \frac{  (1 - \delta) \gamma L}{s_u - s_1}\, \, , \, \, -L (1 + \delta) ) , \\
v_s^L &= ( - \frac{ 3 \delta \gamma L }{4} + (s - \frac{s_1 + s_u}{2} ) \cdot \frac{  (1 - \delta) \gamma L}{s_u - s_1}\, \, , \, \,  -L (1 + \delta) ) , \\  
\end{split}
\end{equation}
and for every $s \in \mathcal{U} \setminus \{ s_1, s_u \}$,
\begin{equation}
\label{eq:coordinates}
\begin{split}
u_s^U &= ( - \frac{\delta \gamma L}{2}  + ( s - \frac{s_1 + s_u}{2} ) \cdot \frac{  (1 - \delta) \gamma L}{s_u - s_1} + \frac{\gamma}{2} (1 + \delta ) L\, \, , \, \, 0), \\
v_s^U &= ( - \frac{ 3 \delta \gamma L }{4} + (s - \frac{s_1 + s_u}{2} ) \cdot \frac{  (1 - \delta) \gamma L}{s_u - s_1} + \frac{\gamma}{2} (1 + \delta ) L\, \, , \, \,  0) ,
\end{split}
\end{equation}
\begin{figure}
\begin{center}
\includegraphics[scale=0.43]{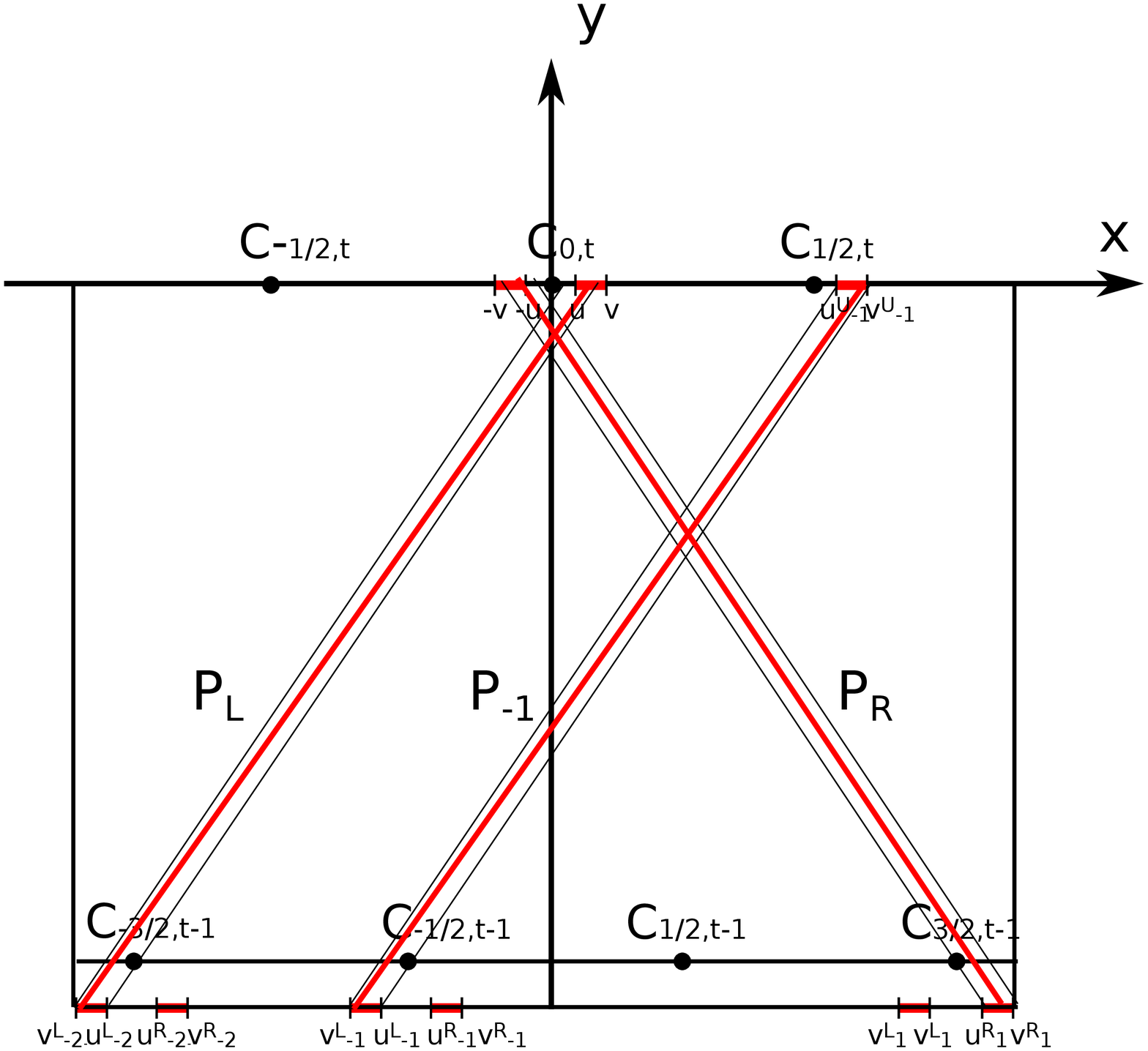}
\includegraphics[scale=0.43]{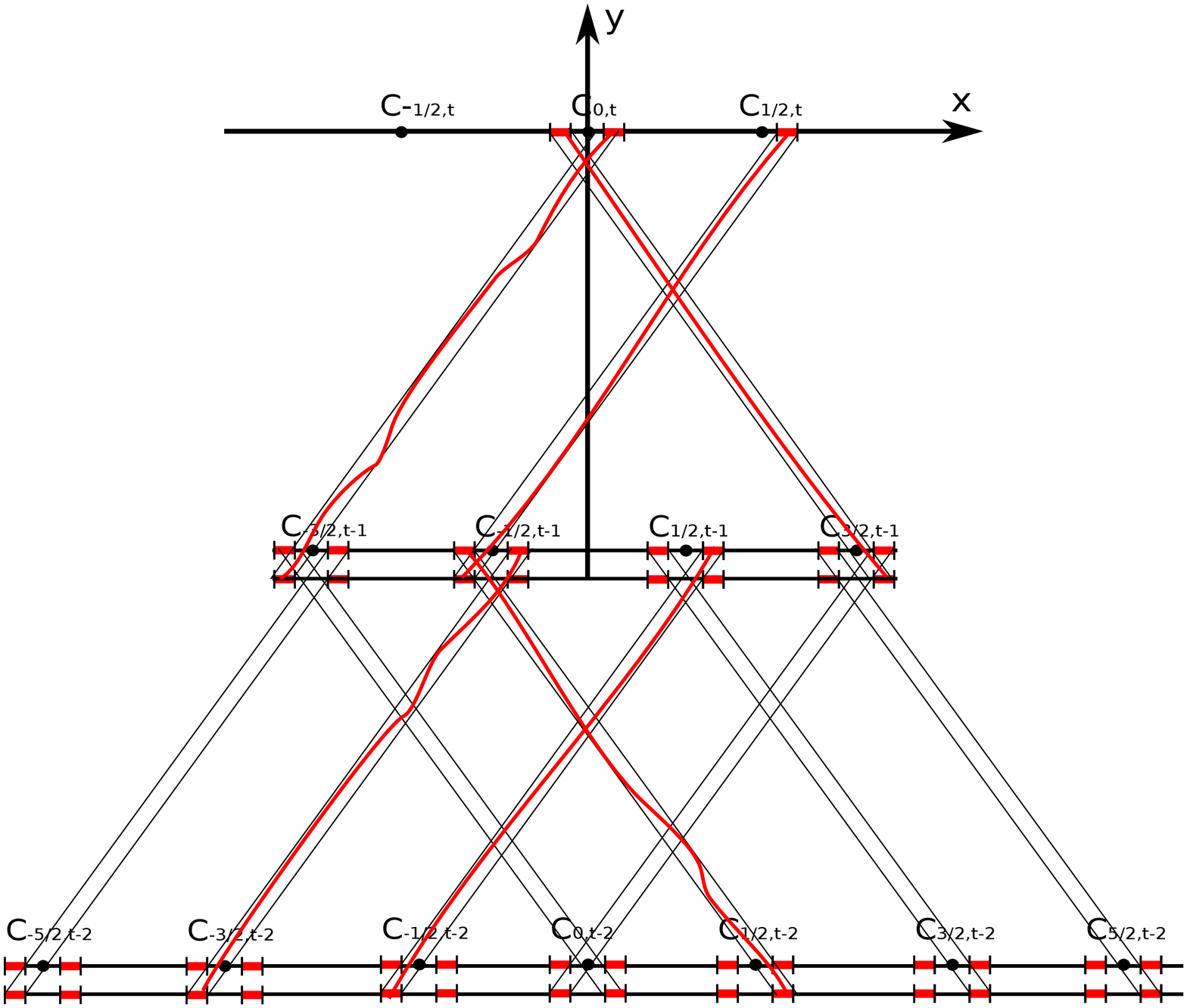}
\caption{\textit{Up}: The rectangle in the figure represents the region $R_{0,t}$ of the graph $T^t_{\frac{s_u + s_1}{2}} \, \mathcal{G}_{\mathcal{U}}$, for some positive integer $t$, in case $\mathcal{U} = \{-2, -1, 1\}$. }
\label{Fig8:SketchRen}
\end{center}
\end{figure}
As one can see in the example in Figure \ref{Fig8:SketchRen},
these points identify some \textit{target zones}
on the right and on the left of points $C_{x,y}$, $(x, y) \in V^2$.
Consider now the parallelograms obtained connecting the following quadruplets of points
 (see also Figure \ref{Fig8:SketchRen}),
\begin{equation}
\begin{split}
\label{eq:quad1}
& P_R =  (-v, \, -u, \,  u_{s_u}^R,\,  v_{s_u}^R),  \\
& P_L = (u, \,  v,\,  u_{s_1}^L, \, v_{s_1}^L) \\
& P_s = (u_s^L, v_s^L, u_s^U, v_s^U ),
\end{split}
\end{equation}
for all $s \in \mathcal{U} \setminus \{ s_1, s_u \}$.
Define the translated parallelograms  $P_R(x,y) = P_R + C_{x,y}$,
$P_L(x,y) = P_R + C_{x,y}$, $P_s(x,y) = P_s + C_{x,y}$
for all $s \in \mathcal{U} \setminus \{ s_1, s_u \}$.
\begin{mydef}
The event $\mathcal{B}_{x,y}$ occurs if and only if the top of all parallelograms
$P_R(x,y)$,  $P_L(x,y)$ and $P_s(x,y)$, for all $s \in \mathcal{U} \setminus \{ s_1, s_u \}$, is connected to the the bottom side by an open path in $T^t_{\frac{\alpha + \beta}{2}} \, \mathcal{G}_{\mathcal{U}}$
that remains always inside the parallelogram.
\end{mydef}
This event is represented in Figure \ref{Fig8:SketchRen}.
This construction is such that the following properties are satisfied. Namely,
\begin{enumerate}
\item the random variables $\varphi_{x,y} $ are $s_u - s_1$-dependent. 
With this we mean that 
$\varphi_{x,y}$ and $\varphi_{x^{\prime},y^{\prime}}$, with
$(x,y)$, $(x^{\prime}, y^{\prime}) \in V^2$,
are independent if $| x - x^{\prime} | > s_u - s_1$ or $| y-y^{\prime}| >1$.
\item Denote by $z_1 \ldots z_m$ the vertices of a path in
$T^t_{\frac{s_1+s_u}{2}} \, \mathcal{G}_{\mathcal{U}}$
and assume that the path is open, i.e. 
$\varphi_{z_i}=1$ for all $i \in \{1, 2, \ldots m \}$.
Then there exists an open path in $T^t_{\frac{\alpha + \beta}{2}} \, \mathcal{G}_{\mathcal{U}}$
that connects a vertex in the $C_{z_1} + [-v, v]$
to a vertex in $C_{z_m} + [-v, v]$
and which remains always inside the parallelograms
that connect  $C_{z_i} + [-v, v]$ to $C_{z_{i+1}} + [-v, v]$,
for all $i \in \{1, 2, \ldots m \}$
(note that  $C_{z} + [-v, v]$
denotes  the segment $[-v, v]$ translated by $C_{z}$).
\item if $\delta, \epsilon >0$ and $p>p_c$,  we can pick $L$
large enough so that for any $(x,y) \in V^2$, $\mathbb{P}_p (\varphi_{x,y} =1 )  > 1 - \epsilon$.
\end{enumerate}
\begin{proof}[Proof of the properties]
We sketch the proof of the three properties above. The proof can be found also in \cite[Section 9]{DurretOr} in the case of bond percolation and symmetric neighbourhoods.
Property 1 follows from the fact that if $R_{x,y}$ and $R_{x^\prime, y^\prime}$ have empty intersection, then the variables $\varphi_{x,y}$ and $\varphi_{x^{\prime}, y^{\prime}}$ are independent.
Property 2 follows by construction (see Figure \ref{Fig8:SketchRen}).
In the example in the figure we represent the graph $T^t_{\frac{s_u + s_1}{2}} \mathcal{G}_{\mathcal{U}}$ assuming $\mathcal{U} = \{-2, -1, 1\}$
as a neighbourhood. One should observe that if the events
 $\mathcal{B}_{0,t}$ and $\mathcal{B}_{-1,t-1}$ occur, then 
at least one vertex belonging to the interval $C_{0,t} + [-v, v]$
is connected to at least one of the vertices belonging to any of the intervals
$C_{-\frac{3}{2},t-2} + [-v, v]$,
$C_{-\frac{1}{2},t-2} + [-v, v]$,
$C_{\frac{1}{2},t-2} + [-v, v]$.

We now prove the third property.  Recall that Proposition \ref{prop:convergencesubadditive} implies that in the transformed graph $T^t_{\frac{\alpha + \beta}{2}} \, \mathcal{G}_{\mathcal{U}}$,
$\overline{r}_n /n \overset{n \rightarrow \infty}{\longrightarrow} \gamma/2$ a.s. and 
$\overline{\ell}_n /n \overset{n \rightarrow \infty}{\longrightarrow}  -\gamma/2$ a. s. 
We will prove that $\forall \epsilon > 0$, the probability that in all the parallelograms in the box there is a connection from the top to the bottom that never crosses the diagonal sides is larger than $1 - \epsilon$.
Let then $e$ be the number of parallelograms in the box $R_{0,0}$. This number depends on the neighbourhood $\mathcal{U}$. We consider the parallelogram $P_R$ and we prove that for every $\epsilon$ there exists $L$ large enough such that the probability that there is no such open path in the parallelogram is less than $\frac{\epsilon}{e}$. As this probability is the same for all parallelograms, this implies that the probability that such open path is present in all parallelograms is $> 1 - \epsilon$.

Consider the parallelogram $P_R$ defined above and recall then the definitions provided in equations (\ref{eq:notation1} - \ref{eq:notation4}).  
Let then $\tilde{r}_n := \sup \xi_n^{(-\infty,  -0.7 \delta \gamma L ]}$
and observe that  $-0.7 \delta \gamma L  \in [-v, -u]$.
Let $\overline{r}_n := \sup \xi_n^{(-\infty,  0 ]}$ and observe that
$ \{ \tilde{r}_n + 0.7 \delta \gamma L:  n \geq 0 \} =_d \{
\overline{r}_n: n \geq 0 \}$.
As $\overline{r}_n / n \rightarrow \frac{\gamma}{2}$ a.s. in the transformed graph 
$T^t_{\frac{\alpha+ \beta}{2}}$, then we can pick $L$ large enough such that with probability $\geq 1 - \frac{\epsilon}{2 e}$ we have that,
\begin{equation}
\begin{split}
\label{eq:proofprop1}
\tilde{r}_{(1+\delta) L} & > - 0.7 \delta \gamma L + ( 1 + 0.98 \delta) \frac{\gamma}{2} L \\ & = -0.71 \delta \gamma L + (1 + \delta) \frac{\gamma}{2} L,
\end{split}
\end{equation}
and for all $m \leq (1 + \delta) L$,
\begin{equation}
\tilde{r}_{m} \leq -0.6 \delta \gamma L + m \frac{1 + 1.08 \delta}{1 + \delta} \frac{\gamma}{2}.
\end{equation}
The two previous equations imply that there is an open path path
from $(- \infty, -0.7 \delta \gamma L] \times \{0\}$ to
$[  -0.71 \delta L \gamma + (1 + \delta) L \frac{\gamma}{2}, 
-0.56 \ \delta L  \gamma + (1 + \delta) L \frac{\gamma}{2}]
\times \{- (1 + \delta) L \}$ which does not cross the line 
$[-u, v_{s_u}^R]$. 
It remains to show that this path does not cross the line $[-v, u_{s_u}^R]$.

We observe that in order a path to travel from the line  $[-v, u^R_{s_u}] $
to $[-0.7 \delta L \gamma + \frac{\gamma}{2} (1 +\delta) L , \infty  ) \times \{-(1 + \delta) L \}$ a path must have an average slope $a  > \frac{\gamma}{2}$.
Thus recall equation (\ref{eq:exponentialspeed2}) and observe that in the transformed graph $T^t_{\frac{\alpha+ \beta}{2}}\mathcal{G}_{\mathcal{U}}$,
$$
\mathbb{P}_p ( \overline{r}_m > a m) \leq C e^{-h_2 m}.
$$
Consider then $M$ large enough such that,
$$
\sum_{m=M}^{\infty} C \exp(-h_2 m) \leq \frac{\epsilon}{4e}.
$$
The probability that one of the points on $[-v, u^R_{s_u}]$ with $ -(1 + \delta) L + M \leq y \leq 0$ is connected to  $[-0.7 \delta L \gamma +  (1 +\delta) \frac{\gamma}{2}  L ,  \infty)  \times \{- (1+\delta) L \}$ is then $\leq  \frac{\epsilon}{4e}$.
Furthermore, observe that the number of points on $[-v, u^R_{s_u}]$ with $ -(1 + \delta) \leq y \leq -(1+\delta) L + M$ does not depend on $L$ and that the distance of any of them from the set  $[-0.7 \delta L \gamma +  (1 +\delta) \frac{\gamma}{2}  L ,  \infty)  \times \{- (1+\delta) L \}$ is proportional to $L$. Thus we can pick $L$ large enough so that the probability that there exists an open path connecting any of these points to $[-0.7 \delta L \gamma +  (1 +\delta) \frac{\gamma}{2}  L ,  \infty]$ is less then $\frac{\epsilon}{4e}$. Combining the two estimations, we conclude that the probability that the line $[-v, u^R_{s_u}]$ is connected by an open path
to $[-0.7 \delta L \gamma +  (1 +\delta) \frac{\gamma}{2}  L ,  \infty]$ is less than $\frac{\epsilon}{2e}$. 

Summarising, we showed that with probability $\geq 1 - \frac{\epsilon}{e}$,
there is an open path from $(- \infty, -0.7 \delta \gamma L] \times \{0\}$ to
$[  -0.71 \delta L \gamma + (1 + \delta) L \frac{\gamma}{2}, 
-0.56 \ \delta L  \gamma + (1 + \delta) L \frac{\gamma}{2}]
\times \{- (1 + \delta) L \}$ which does not cross the line 
$[-u, v_{s_u}^R]$ and the line $[-v, u^R_{s_u}]$ is not connected by an open path to
$[-0.7 \delta L \gamma +  (1 +\delta) \frac{\gamma}{2}  L ,  \infty] \times \{-(1 + \delta) L \} $.
This implies that the probability that there exists a path joining the top to the bottom of $P_R$ without ever crossing its diagonal lines is $\geq 1 - \frac{\epsilon}{e}$.
Repeating the argument for all parallelograms in the box, we conclude that if $L$ is large enough then with probability at least $1 - \epsilon$ the event $\mathcal{B}_{0,0}$
occurs.
\end{proof}

\subparagraph{Part 2: Peierls argument}
Now we use the Peierls argument
for the $(s_u-s_1)$-dependent oriented percolation model
to prove that there exists $p_1 > p_c$ and positive constants $A^{\prime}$, $b^{\prime}$ (dependent on $p$) such that
for all $p \in (p_1, 1]$,
\begin{equation}
\label{eq:inequality3}
\mathbb{P}_p ( T^t_{\frac{s_1 + s_u}{2}} \, \mathcal{H}_n ) \geq 1 - A^{\prime} \cdot n e^{ -b^{\prime} n}.
\end{equation}
Let us explain first why this is sufficient to prove the proposition. Later we prove (\ref{eq:inequality3}).

Recall the third property of the dynamic-block construction presented above and observe that if $p> p_c$, then we can pick $L$ large enough such that, for every
$(x,y)$ belonging to the set of vertices of $T^t_{\frac{s_1 + s_2}{2}}\mathcal{G}_{\mathcal{U}}$, $\mathbb{P}_p( \mathcal{B}_{x,y} ) > p_1$.
Recall that the state of sites belonging to $T^t_{\frac{s_1 + s_2}{2}}\mathcal{G}_{\mathcal{U}}$ is a function of the realization in the graph $T^t_{\frac{\alpha + \beta}{2}} \mathcal{G}_{\mathcal{U}}$.
From the second property of the dynamic-block construction, if such sites are open with probability $> p_1$, then 
(\ref{eq:inequality3}) implies that with probability not less than $1 - A^{\prime} \cdot n e^{ -b^{\prime} n}$ the event  $\mathcal{H}_n$ occurs in $T^t_{\frac{s_1 + s_2}{2}}$.
Hence, from the second property of the dynamic-block construction, the event $\mathcal{H}_{L n}$ occurs in the graph 
$T^t_{\frac{\alpha + \beta}{2}} \, \mathcal{G}_{\mathcal{U}}$
with probability not less than $1 - A^{\prime} \cdot n e^{ -b^{\prime} n}$. 
One can rearrange the value of $b^{\prime}$ getting rid of the factor $n$, for $n$ large enough. Finally, by defining new constants $A = A^{\prime} /  L$ and $b = b^{\prime} L$,
the statement of Proposition \ref{prop:inequality2} follows.

We start proving (\ref{eq:inequality3}). We define a new graph $\mathcal{L}$,
that is a sub-graph of $T_{\frac{s_1 + s_u}{2}}\,  \mathcal{G}_{\mathcal{U}}$, whose 
vertices $(x,y)$ are,
\begin{equation}
\label{eq:Vprime}
V^{\prime} = \{ (x, y)  : \, \, x  = (s_u - s_1) z -  (y-t) \frac{s_u - s_1}{2}, \, \, z \in \mathbb{Z}, y \in \mathbb{Z}\},
\end{equation}
and whose edges connect vertices $(x,y)$ to $(x \pm \frac{s_u - s_1}{2}, y-1)$.
The reason shy we introduce $\mathcal{L}$ is that, as every site has only two neighbours, it is easier to construct its dual graph. The new graph $\mathcal{L}$ is represented in the example in Figure \ref{Fig:transformatonpeierls} on the right.
\begin{figure}
\begin{center}
\includegraphics[scale=0.35]{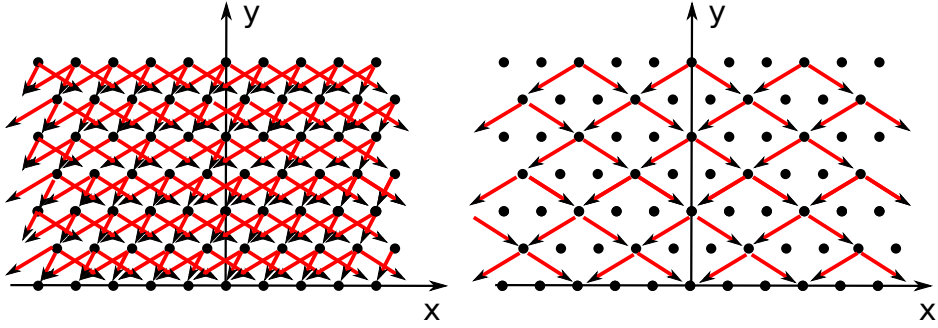}
\caption{\textit{Left}: representation of $T^t_{ \frac{s_1 + s_u}{2}} \mathcal{G}_{\mathcal{U}}$,
in case of $\mathcal{U}= \{ -2, - 1, 1 \}$.  \textit{Right}: points correspond to vertices
of $T^t_{ \frac{s_1 + s_u}{2}} \mathcal{G}_{\mathcal{U}}$, arrows represent 
edges of $\mathcal{L}$, points connected by an arrow correspond to vertices of $\mathcal{L}$.
The graph $\mathcal{L}$, defined in the text, is a subset of $T^t_{ \frac{s_1 + s_u}{2}} \mathcal{G}_{\mathcal{U}}$.}
\label{Fig:transformatonpeierls}
\end{center}
\end{figure}
As $\mathcal{L}$ is a sub-graph of $T_{\frac{s_1 + s_u}{2}}\,  \mathcal{G}_{\mathcal{U}}$,
the following inequality holds,
\begin{equation}  
\label{eq:ineq14}
\mathbb{P}_p ( \mathcal{H}^{\mathcal{L}}_{n}  ) \leq \mathbb{P}_p ( T^t_{\frac{s_u + s_u}{2}} \, \mathcal{H}_{n}).
\end{equation}
In the previous expression, the superscript $^{\mathcal{L}}$ is used to denote that event 
$\mathcal{H}_{n}$, defined in (\ref{eq:evH}),
occurs on the graph $\mathcal{L}$. 
Call then $\mathcal{L}_{D}$
the dual graph of $\mathcal{L}$.
The graph is represented on the right of Figure \ref{Fig9-10:EventsHn} and its costruction is due to
\cite{ToomDiscr, Uniform}.
\begin{figure}
\begin{center}
\includegraphics[scale=0.27]{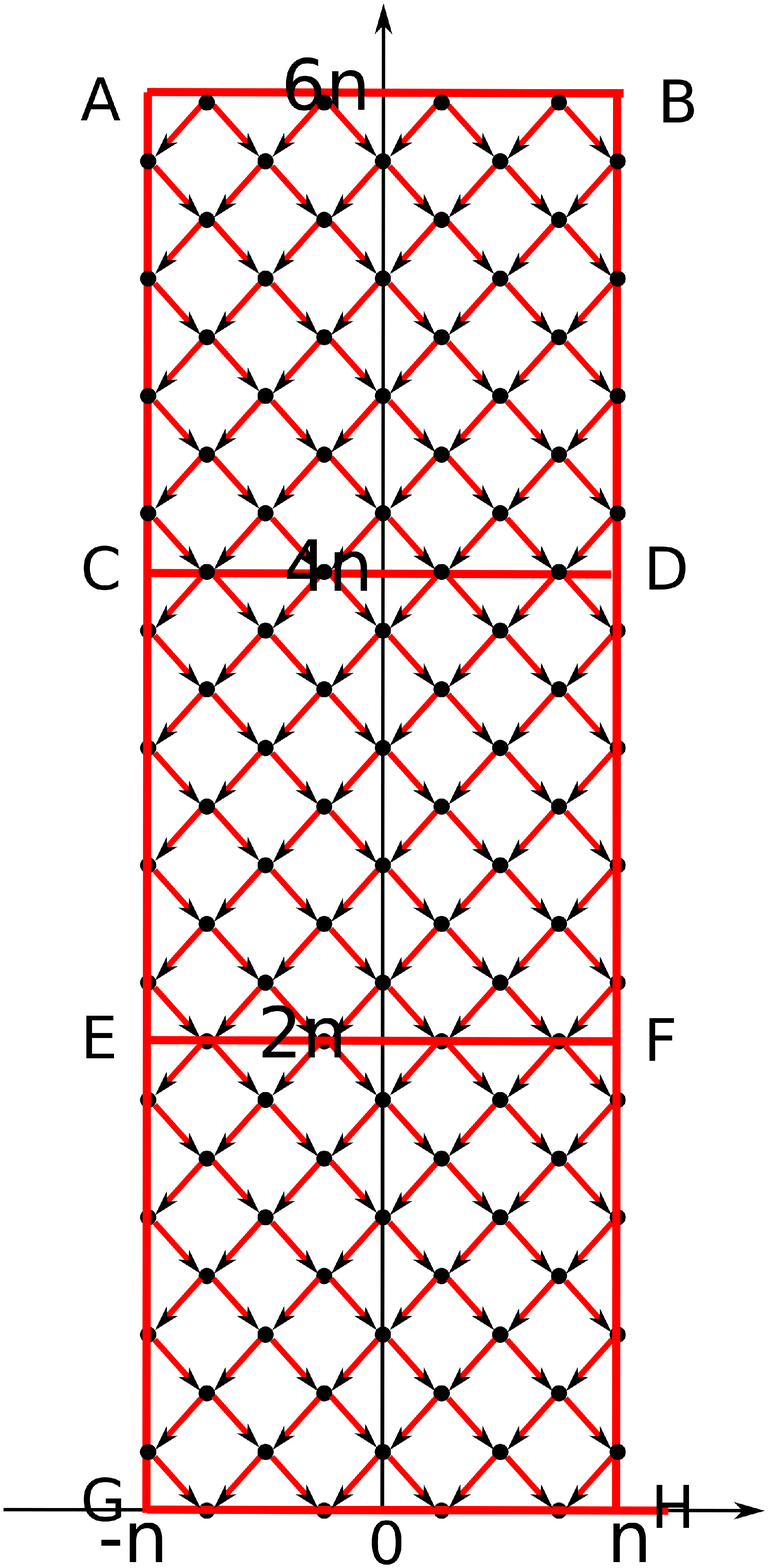}
\includegraphics[scale=0.27]{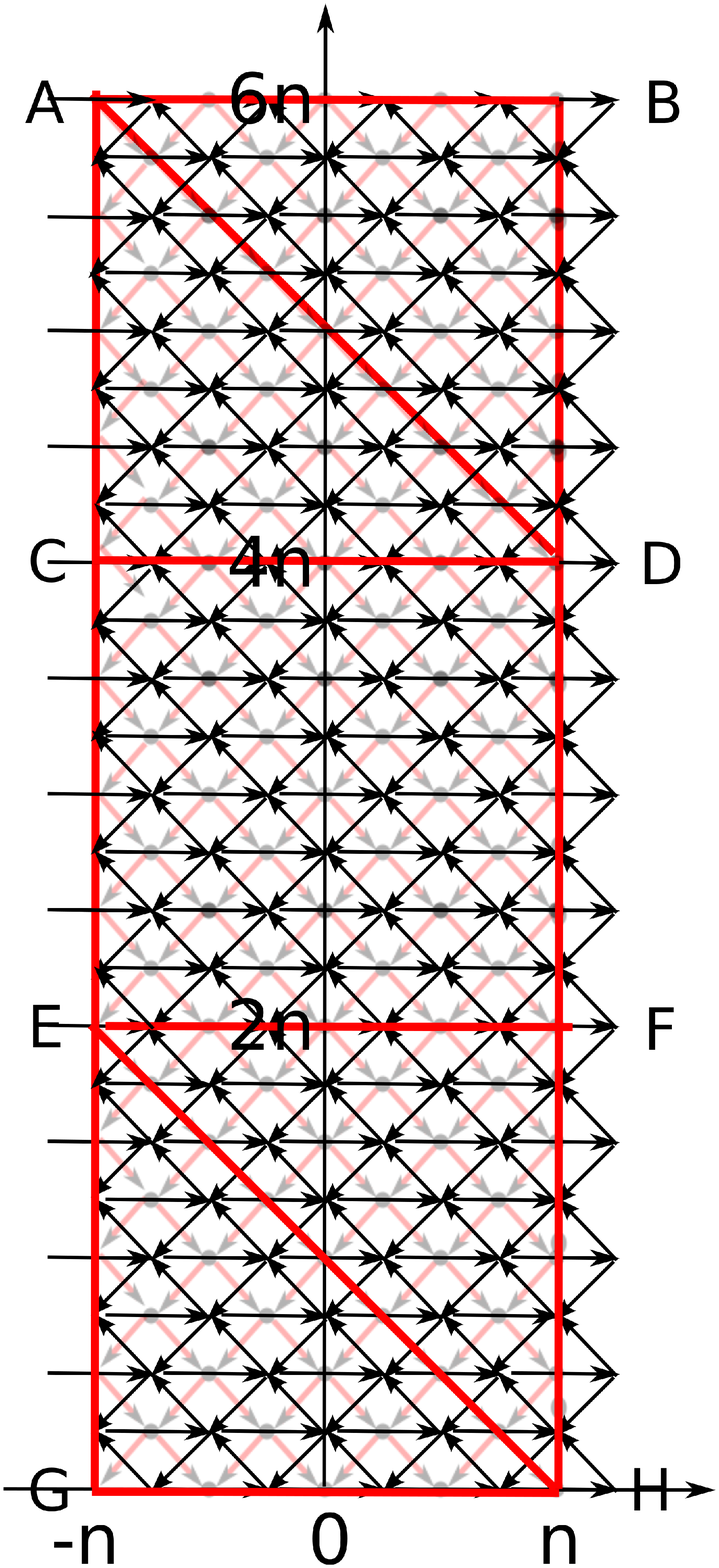}
\caption{Note: the horizontal axis has been rescaled by $\frac{s_u - s_1}{2}$ in both graphs.
$A= (-n,6n), B = (n, 6n), C= (-n, 4n), D= (n, 4n), E= (-n, 2n), F= (n, 2n), G= (-n, 0), H= (n,0)$.
\textit{Left}: representation of the 
graph $\mathcal{L}$. The event $\mathcal{H}_n$ occurs iff the side 
$AB$ is connected to the side $GH$ by an open path in $\mathcal{L}$
that does not cross the sides $AG$ and $BH$.
\textit{Right}: representation of the graph $\mathcal{L}^D$, as defined in the text.
The event $\mathcal{H}_n$ does not occur \textit{iff} one of the sides
$CE$ or $EH$ is connected to one of the sides $AD$ or $DF$ by an open path in the dual lattice.}
\label{Fig9-10:EventsHn}
\end{center}
\end{figure}
The dual graph is composed of three types of edges, namely
edges pointing down-left, those pointing up-left and those pointing right.
Every edge pointing right is positioned over a vertex of the original graph $\mathcal{L}$.
Edges down-left and up-left are always \textit{open},
edges pointing right are \textit{open} if and only if the corresponding vertex of the original graph is \textit{closed}.
A path in the dual graph is \textit{open} if and only if all its edges are open.
The following proposition connects the occurrence of the event $\mathcal{H}_n$ in $\mathcal{L}$
with the occurrence of a second event on the dual lattice.
\begin{prop}
\label{prop:percdual}
Consider Figure \ref{Fig9-10:EventsHn}. For every $n \in \mathbb{N}$, there exists an open path in $\mathcal{L}$
connecting $AC$ to $FH$
\textit{iff} there is no open  path in the dual lattice
connecting one of the sides $CE$ or $EH$ to one of the sides $AD$ or $DF$.
\end{prop}
\begin{proof}
We provide a graphical proof. Consider Figure \ref{Fig9-10:EventsHn}.
On the left we have represented the graph $\mathcal{L}$ and on the right we have represented its dual. Consider a realisation in the auxiliary space $\Omega$ 
and recall that if a site is open in $\mathcal{L}$, then the corresponding horizontal
edge is closed in the dual graph and vice versa.
The reader should observe that, as long as there is an open path connecting $AC$ to $FH$ in $\mathcal{L}$, no open path in the dual graph connecting one of the sides $CE$ or $EH$ to one of the sides $AD$ or $DF$ can exist. On the other hand, as long there
exists an open path in the dual graph connecting one of the sides $CE$ or $EH$
to one of the sides $AD$ or $DF$, no open path in $\mathcal{L}$ connecting $AC$ to $FH$ can exist.
\end{proof}
Both the proposition and the dual construction are analogous to the one presented in
\cite{Uniform}. We use this proposition to provide a lower bound for $\mathbb{P}_p (\mathcal{H}^{\mathcal{L}}_n)$.
Consider then a vertex $z$ on $CE$ or on $EH$.
Call $C_{z,h}$ the set of paths connecting the vertex $z$
to one of the sides $AD$ or $DF$ and having $h$ edges pointing to the right.
Call $N_{z,h}$ the total number of such paths.
Consider one of these paths and call $dl$ the number of its edges pointing down-left
and $ul$ the number of edges pointing up-left.
As the last edge of the path cannot be on the left of the first edge,
 $2h -  u l - d l \geq 0$.
This implies that for each of these paths sum $h + u l  + d l $
is bounded from above by $3h$.
As there are only $3$ different types of steps,
for any vertex $z$ located on $CE$ or on $EH$,
$N_{z,h} \leq 3^{3h}$.
Thus $N_{z,h} \leq 3^{3h}$ for every $z$.
Denote by $\overline{ \mathcal{H}_n^{\mathcal{L}}}$  the complementary of $\mathcal{H}_n^{\mathcal{L}}$.
Recall Proposition \ref{prop:percdual} and observe the fact that, in order $CG$ to be connected to $AD$ or to $DH$,
at least $\lfloor \frac{2n}{s_u  - s_1} \rfloor$ horizontal steps to the right are needed. Then,
$\mathbb{P}_p( \overline{ \mathcal{H}_n^{\mathcal{L}}} ) = 
\mathbb{P}_p( \bigcup\limits_{z \in {CE \cup EH}} \bigcup\limits_{h=2n}^{\infty} \bigcup\limits_{c \in C_{z,h}}   \{ c \mbox{ is open } \})$.
Observe also that, given a path $c \in C_{z,h}$, $\mathbb{P}_p(  c \mbox{ is open }) \leq (1-p) ^\frac{h}{2}$,
considering only the state of one every two edges to the right, as states of edges located over non-neighbour sites are independent.
By using the union bound, we determine an upper bound for $\mathbb{P}_p( \overline{ \mathcal{H}_n^{\mathcal{L}}} )$,
\begin{equation}
\label{eq:ineq4}
\mathbb{P}_p( \overline{ \mathcal{H}_n^{\mathcal{L}}}  ) 
\leq \sum_{z \in CE \cup EH}  \sum_{h=\lfloor \frac{2n}{s_u  - s_1} \rfloor}^{\infty}  N_{z,h} \,  (1-p)^{h/{2}}\leq A^{\prime} \cdot  n \exp ( - b^{\prime}  n  ),
\end{equation}
where the second inequality is true with $A^{\prime}, b^{\prime}$ positive constants
if  $p > 1 - \frac{1}{3}^6$. 
\end{proof}

\section{Appendix: Numerical Simulatios}
We consider Percolation PCA with space $\mathbb{S}_n$ and periodic boundaries.
We run the process $R$ times and we define 
$$P( p) : = N( R, T, n, p) / R,  $$
where $N(R,T,n, p)$ is the number of times the origin has state $1$ at time $T$.
In Figure \ref{Fig:simulations} we plot
the function $P(p)$ for different choices of the neighbourhood for a small range of $p$.
The parameters considered are $n=100000$, $T=100000$, $R=2000$
 and $n=500000$, $T=500000$ and $R=200$. 
\begin{figure}[!]
\begin{center}
\includegraphics[width=0.48\linewidth]{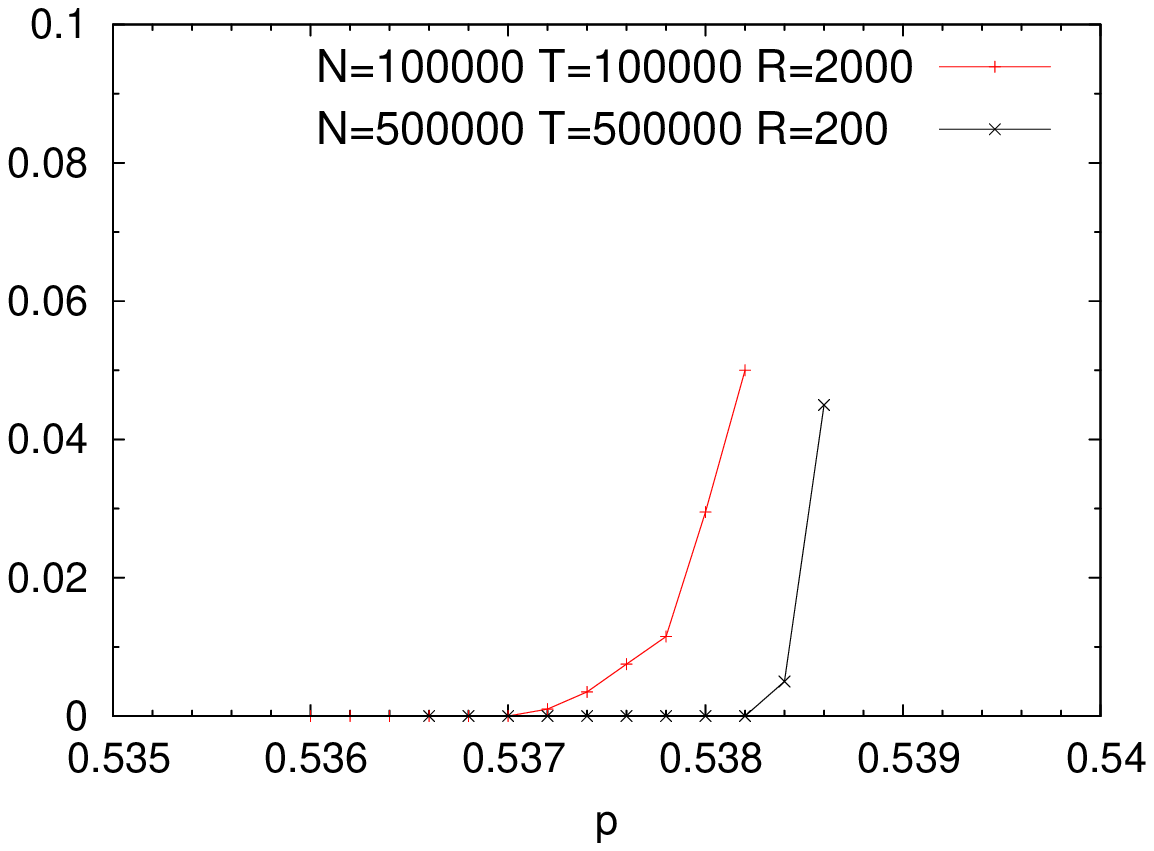}
\includegraphics[width=0.48\linewidth]{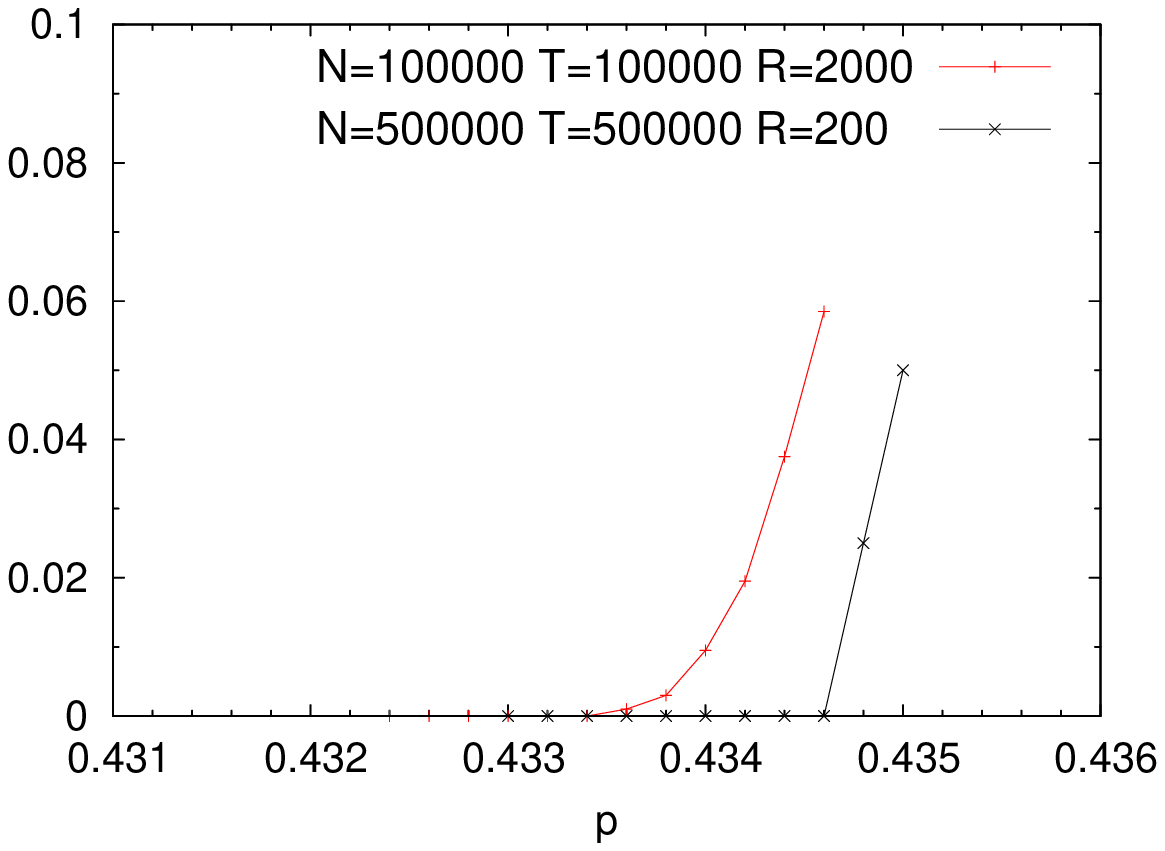}
\includegraphics[width=0.48\linewidth]{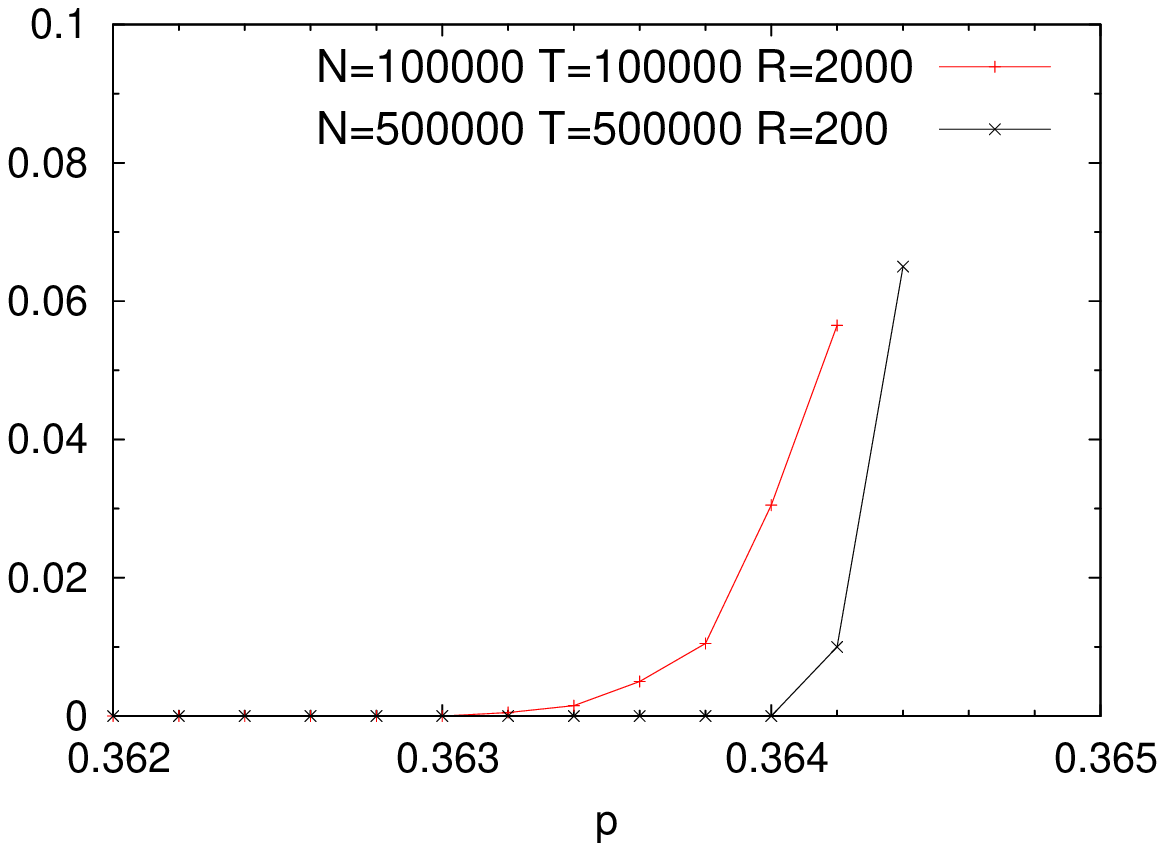}
\includegraphics[width=0.48\linewidth]{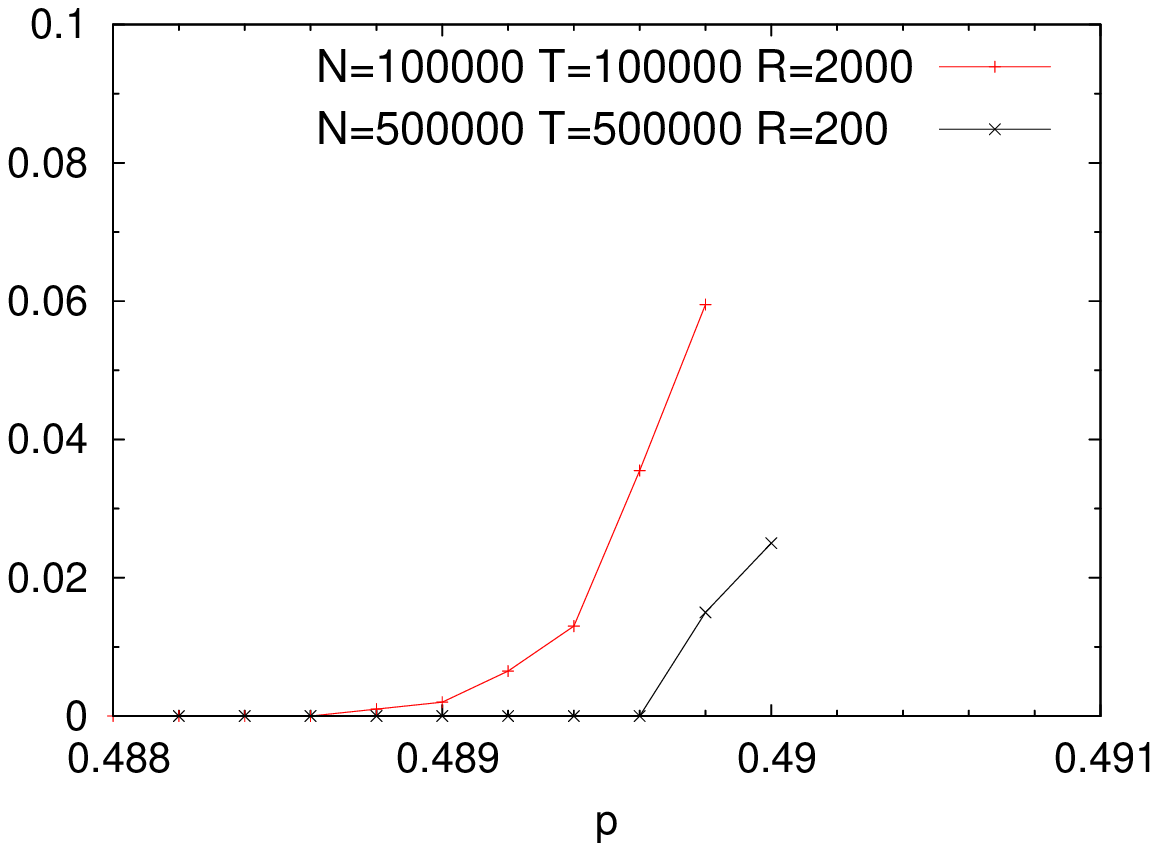}
\includegraphics[width=0.48\linewidth]{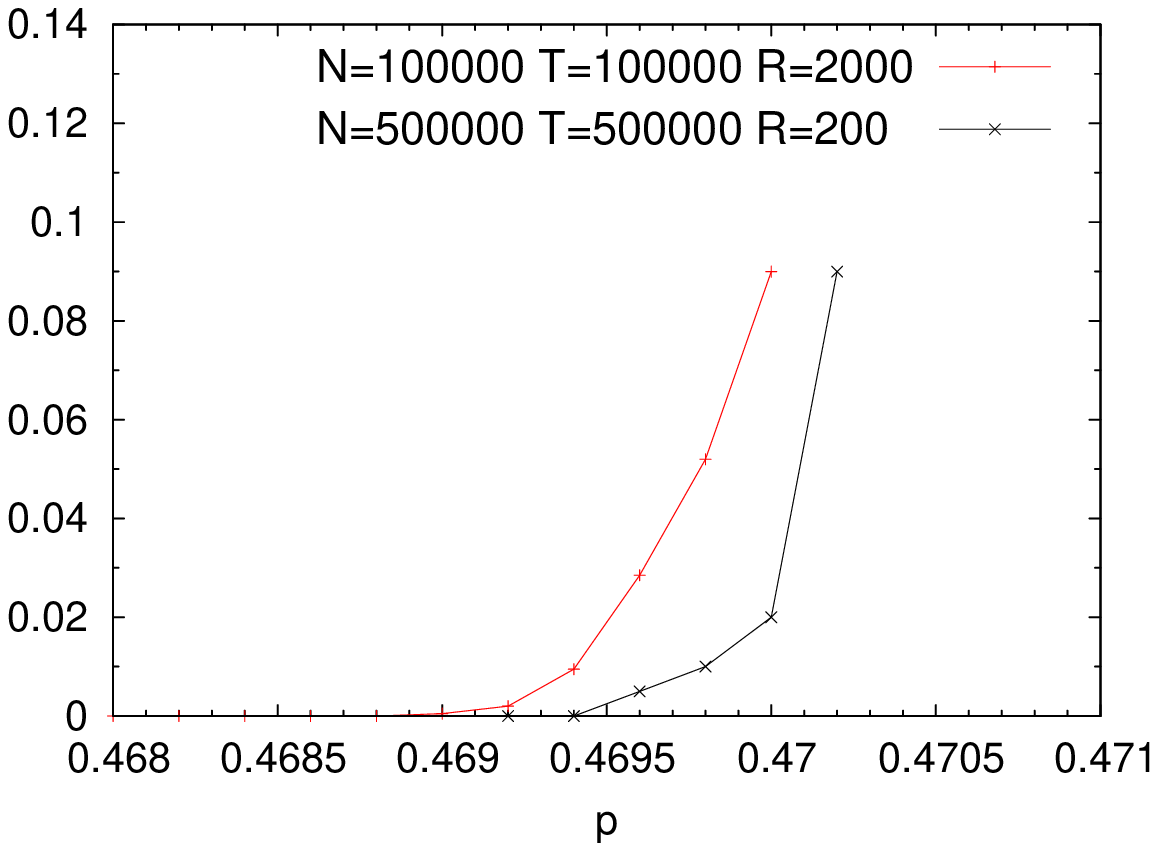}
\caption{Up Left: $\mathcal{U}=\{-1, 0, 1\}$. Up Right:
$\mathcal{U}=\{-1, 0, 1, 2\}$. Middle Left: $\mathcal{U}=\{-1, 0, 1, 2, 3\}$.
Middle Right: $\mathcal{U}=\{-1, 0,  2\}$. Down: $\mathcal{U}=\{-1, 0,  3\}$.}
\label{Fig:simulations}
\end{center}
\end{figure}

\section*{Acknowledgments}
The author is grateful to Artem Sapozhnikov
for teaching him many techniques used in the proof of Theorem 2.2.
The author thanks Andrei Toom,  Artem Sapozhnikov and
the anonymous referees whose comments
considerably helped to improve the presentation.
The author thanks Leonid Mityushin
for sending him the English version
of the proof of Proposition 9.
The author is grateful to J\"{u}rgen Jost
for giving him the possibility to study Probabilistic
Cellular Automata.
The author also thanks the organizers of the workshop
Probabilistic Cellular Automata - Eurandom
(June 2013, Eindhoven), 
where many inspiring discussions
favoured the development of this work.

\end{document}